\documentclass[11pt]{amsart}
\usepackage{amsmath,amsfonts,amsthm,amssymb,xcolor}
\usepackage{fancybox}
\usepackage{pdfsync}
\usepackage{fullpage}
\usepackage{hyperref}
\usepackage{algorithm} 
\usepackage[noend]{algpseudocode}
\usepackage{algorithmicx}
\usepackage{thmtools}
\usepackage{thm-restate}
\usepackage[foot]{amsaddr}

\DeclareMathOperator{\poly}{poly}

\def\showauthornotes{0}

\usepackage{hyperref}
\usepackage{amsmath,amssymb,amsthm,amstext,amsfonts,bbm,graphicx,xspace,nicefrac}
\usepackage{color,stmaryrd,enumerate,latexsym,bm,amsfonts,
  subfigure,wrapfig,verbatim, tabularx,textcomp}
\usepackage[small]{caption}
\usepackage{comment} 
\usepackage{epsfig} 
\usepackage{latexsym,nicefrac,bbm}
\usepackage{xspace}
\usepackage{color,fancybox,graphicx,url,subfigure}
\usepackage{enumitem, fullpage}
\usepackage{booktabs}
\usepackage{commath}
\usepackage{mdframed}

\newcommand{\defeq}{\stackrel{\textup{def}}{=}}

\newtheorem{theorem}{Theorem}[section]

\newtheorem{lemma}[theorem]{Lemma}
\newtheorem{corollary}[theorem]{Corollary}

\newtheorem{fact}[theorem]{Fact}

\theoremstyle{definition}
\newtheorem{definition}[theorem]{Definition}
\newtheorem{remark}[theorem]{Remark}

\DeclareMathOperator*{\nnz}{nnz}
\usepackage{xspace,nicefrac,amsmath}
\usepackage{algorithm} 
\usepackage[noend]{algpseudocode}
\usepackage{algorithmicx}
\usepackage{thmtools}
\usepackage{thm-restate}

\newcommand{\eps}{\varepsilon}
\renewcommand{\epsilon}{\varepsilon}

\newcommand{\nfrac}[2]{\nicefrac{#1}{#2}}
\def\abs#1{\left| #1 \right|}

\newcommand\rea{\mathbb R}

\newcommand{\marginlabel}[1]%
{\mbox{}\marginpar{\it{\raggedleft\hspace{0pt}#1}}}

\newcommand\map[2]{\mathcal{M}_{#1 \rightarrow #2}}
\newcommand\obj{\calE\xspace}
\newcommand{\circapprox}{\preceq^{\text{cycle}}}

\DeclareMathOperator{\polylog}{polylog}

\newcommand{\cost}{\texttt{cost}}

\newcommand\calE{\mathcal{E}}

\newcommand\calG{\mathcal{G}}
\newcommand\calH{\mathcal{H}}

\newcommand\calM{\mathcal{M}}

\definecolor{Mygray}{gray}{0.8}

 \ifcsname ifcommentflag\endcsname\else
  \expandafter\let\csname ifcommentflag\expandafter\endcsname
                  \csname iffalse\endcsname
\fi

\ifnum\showauthornotes=1
\newcommand{\todo}[1]{{\bf \color{red} TODO: #1}}
\else
\newcommand{\todo}[1]{}
\fi

\ifnum\showauthornotes=1
\newcommand{\Authornote}[2]{{\sf\small\color{red}{[#1: #2]}}}
\newcommand{\Authoredit}[2]{{\sf\small\color{red}{[#1]}\color{blue}{#2}}}
\newcommand{\Authorcomment}[2]{{\sf \small\color{gray}{[#1: #2]}}}
\newcommand{\Authorfnote}[2]{\footnote{\color{red}{#1: #2}}}
\newcommand{\Authorfixme}[1]{\Authornote{#1}{\textbf{??}}}
\newcommand{\Authormarginmark}[1]{\marginpar{\textcolor{red}{\fbox{
#1:!}}}}
\else
\newcommand{\Authornote}[2]{}
\newcommand{\Authoredit}[2]{}
\newcommand{\Authorcomment}[2]{}
\newcommand{\Authorfnote}[2]{}
\newcommand{\Authorfixme}[1]{}
\newcommand{\Authormarginmark}[1]{}
\fi

\newlength{\pgmtab}  
\let\originalleft\left
\let\originalright\right
\renewcommand{\left}{\mathopen{}\mathclose\bgroup\originalleft}
  \renewcommand{\right}{\aftergroup\egroup\originalright}

\def\defeq{\stackrel{\mathrm{def}}{=}}
\def\setof#1{\left\{#1  \right\}}

\def\diag#1{\textsc{Diag}\left( #1 \right)}

\newcommand{\union}{\cup}

\def\abs#1{\left|#1  \right|}

\def\norm#1{\left\| #1 \right\|}

\newcommand\ppsi{\boldsymbol{\mathit{\psi}}}

\newcommand\ddelta{\boldsymbol{\delta}}
\newcommand\ttau{\boldsymbol{\tau}}

\def\aa{\pmb{\mathit{a}}}
\newcommand\bb{\boldsymbol{\mathit{b}}}
\newcommand\cc{\boldsymbol{\mathit{c}}}
\newcommand\dd{\boldsymbol{\mathit{d}}}

\newcommand\ff{\boldsymbol{\mathit{f}}}

\renewcommand\gg{\boldsymbol{\mathit{g}}}
\newcommand\hh{\boldsymbol{\mathit{h}}}

\renewcommand\ll{\boldsymbol{\mathit{l}}}

\newcommand\rr{\boldsymbol{\mathit{r}}}
\renewcommand\ss{\boldsymbol{\mathit{s}}}
\def\tt{\boldsymbol{\mathit{t}}}
\newcommand\uu{\boldsymbol{\mathit{u}}}
\newcommand\vv{\boldsymbol{\mathit{v}}}
\newcommand\ww{\boldsymbol{\mathit{w}}}
\newcommand\yy{\boldsymbol{\mathit{y}}}

\newcommand\xx{\boldsymbol{\mathit{x}}}

\renewcommand\AA{\boldsymbol{\mathit{A}}}
\newcommand\BB{\boldsymbol{\mathit{B}}}
\newcommand\CC{\boldsymbol{\mathit{C}}}
\newcommand\DD{\boldsymbol{\mathit{D}}}

\newcommand\II{\boldsymbol{\mathit{I}}}

\newcommand\NN{\boldsymbol{\mathit{N}}}
\newcommand\MM{\boldsymbol{\mathit{M}}}
\newcommand\LL{\boldsymbol{\mathit{L}}}

\newcommand\RR{\boldsymbol{\mathit{R}}}
\renewcommand\SS{\boldsymbol{\mathit{S}}}
\newcommand\TT{\boldsymbol{\mathit{T}}}
\newcommand\UU{\boldsymbol{\mathit{U}}}
\newcommand\WW{\boldsymbol{\mathit{W}}}
\newcommand\VV{\boldsymbol{\mathit{V}}}
\newcommand\XX{\boldsymbol{\mathit{X}}}

\newcommand{\gghat}{\boldsymbol{\widehat{{g}}}}

\newcommand\MMtil{\boldsymbol{\widetilde{{M}}}}

\newcommand\MMhat{\boldsymbol{\widehat{M}}}
\newcommand\NNtil{\boldsymbol{\widetilde{{N}}}}

\newcommand\NNhat{\boldsymbol{\widehat{N}}}

\newcommand\LSS{\textsc{LSS}}

\newcommand\Otil{\widetilde{O}}

\newcommand{\opt}{\textsc{Opt}}

\DeclareMathOperator{\argmax}{argmax}

\newenvironment{tight_enumerate}{
\begin{enumerate}
 \setlength{\itemsep}{2pt}
 \setlength{\parskip}{1pt}
}{\end{enumerate}}
\newenvironment{tight_itemize}{
\begin{itemize}
 \setlength{\itemsep}{2pt}
 \setlength{\parskip}{1pt}
}{\end{itemize}}

\newcommand{\vone}{\boldsymbol{\mathbf{1}}}
\newcommand{\vzero}{\boldsymbol{\mathbf{0}}}

\newcommand{\etal}{\emph{et al.}}

\newcommand\Dtil{{\widetilde{{\Delta}}}}
\newcommand\Dbar{{\bar{{\Delta}}}}
\newcommand\Dopt{{{{\Delta^{\star}}}}}
\newcommand\res{{{{res}}}}

 {
	\begin{enumerate}}{\end{enumerate}}

\newlength{\tpush}
\newcommand{\handout}[5]{
   \noindent
   \begin{center}
   \framebox{ \vbox{ \hbox to \textwidth { {\bf \coursenum\ :\  \coursename} \hfill #5 }
       \vspace{3mm}
       \hbox to \textwidth { {\Large \hfill #2  \hfill} }
       \vspace{1mm}
       \hbox to \textwidth { {\it #3 \hfill #4} }
     }
   }
   \end{center}
   \vspace*{4mm}
   \newcommand{\lecturenum}{#1}
   \addcontentsline{toc}{chapter}{Lecture #1 -- #2}
}

\allowdisplaybreaks

\usepackage[
  backend=biber,
  backref=true,
  backrefstyle=none,
  date=year,
  doi=false,
  giveninits=true,
  hyperref=true,
  maxbibnames=10,
  sortcites=false,
  style=alphabetic,
  url=false, 
]{biblatex}
\newcommand{\deeksha}{\Authornote{Deeksha}}
\newcommand{\bab}{\Authornote{Brian}}
\newcommand{\sushant}{\Authornote{Sushant}}

\ifnum\showauthornotes=1
\newcommand{\todolow}[1]{{\bf \color{orange} TODOLOW: #1}}
\newcommand{\rasmus}[1]{{\bf \color{olive} Rasmus: #1}}
\else
\newcommand{\todolow}[1]{}
\newcommand{\rasmus}[1]{}
\fi

\begin{document}

\title{Almost-linear-time Weighted $\ell_p$-norm Solvers in Slightly
  Dense Graphs via Sparsification}
  

\author{Deeksha Adil$ ^1$ }
  \address{$^{1,4}$University of Toronto, $^2$TTI Chicago, $^3$ETH Zurich.} 
  \email[1,4]{ $\{$deeksha$\mid$sachdeva$\}$@cs.toronto.edu}
\author{Brian Bullins$^2$}
\email[2]{bbullins@ttic.edu}
\author{Rasmus Kyng$^3$}
\email[3]{kyng@inf.ethz.ch}
 \author{Sushant Sachdeva$^4$}

\maketitle
\thispagestyle{empty}
\begin{abstract}
  We give almost-linear-time algorithms for constructing sparsifiers
  with $n \poly(\log n)$ edges that approximately preserve weighted
  $(\ell^{2}_2 + \ell^{p}_p)$ flow or voltage objectives on graphs. For
  flow objectives, this is the first sparsifier construction for such
  mixed objectives beyond unit $\ell_p$ weights, and is based on
  expander decompositions. For voltage objectives, we give the first
  sparsifier construction for these objectives, which we build using
  graph spanners and leverage score sampling. Together with the
  iterative refinement framework of [Adil et al, SODA 2019], and a new
  multiplicative-weights based constant-approximation algorithm for
  mixed-objective flows or voltages, we show how to find
  $(1+2^{-\text{poly}(\log n)})$ approximations for weighted
  $\ell_p$-norm minimizing flows or voltages in
  $p(m^{1+o(1)} + n^{4/3 + o(1)})$ time for $p=\omega(1),$ which is
  almost-linear for graphs that are slightly dense
  ($m \ge n^{4/3 + o(1)}$).
\end{abstract}
\newpage

\setcounter{page}{1}
\section{Introduction}

Network flow problems are some of the most extensively studied
problems in optimization (e.g. see~\cite{AhujaMO93, Schrijver02,
  GoldbergT14}). A general network flow problem on a graph $G(V, E)$
with $n$ vertices and $m$ edges can be formulated as
\begin{equation*}
\min_{\BB^{\top}\ff = \dd} \cost(\ff),
\end{equation*}
where $\ff \in \rea^{E}$ is a flow vector on edges satisfying net
vertex demands $\dd \in \rea^{V},$ $\BB \in \rea^{E \times V}$ is the
signed edge-vertex incidence matrix of the graph, and $\cost(\ff)$ is
a cost measure on flows.
The weighted $\ell_{\infty}$-minimizing flow problem, i.e.,
$\cost(\ff) = \|\SS^{-1} \ff\|_{\infty},$ captures the celebrated
maximum-flow problem with capacities $\SS;$ the weighted
$\ell_{1}$-minimizing flow problem, $\cost(\ff) = \norm{\SS \ff}_{1}$
captures the transshipment problem generalizing shortest paths with
lengths $\SS$; and
$\cost(\ff) = \ff^{\top}\RR\ff = \|\RR^{\frac{1}{2}}\ff\|_{2}^{2}$
captures the electrical flow problem~\cite{SpielmanT04}.

Dual to flow problems are voltage problems, which can be formulated as
\[\min_{\dd^{\top}\vv = 1 } \cost'(\BB\vv),\]
Analogous to the flow problems, picking
$\cost'(\BB\vv) = \norm{\SS\BB\vv}_{1}$ captures the capacitated
min-cut problem, $\cost'(\BB\vv) = \|\SS^{-1}\BB\vv \|_{\infty}$
captures vertex-labeling~\cite{KyngRSS15}, and
$\cost'(\BB\vv) = (\BB\vv)^{\top} \RR^{-1}\BB\vv =
\|\RR^{-\frac{1}{2}} \BB\vv \|_{2}^{2}$ captures the electrical
voltages problem.

The seminal work of Spielman and Teng~\cite{SpielmanT04} gave the
first nearly-linear-time algorithm for computing
($1+1/\poly(n)$)-approximate solutions to electrical (weighted
$\ell_2$-minimizing) flow/voltage problems. This work spurred the
``Laplacian Paradigm'' for designing faster algorithms for several
classic graph optimization problems including maximum flow
~\cite{ChristianoKMST10, Sherman13, KelnerLOS14}, multi-commodity flow
\cite{KelnerLOS14}, bipartite matching~\cite{Madry13},
transshipment~\cite{Sherman17a}, and graph
partitioning~\cite{OrecchiaSV12}; culminating in almost-linear-time or
nearly-linear-time low-accuracy algorithms (i.e. $1+\eps$
approximations with $\poly(\frac{1}{\eps})$ running time dependence)
for many of these problems.

Progress on high-accuracy algorithms (i.e. algorithms that return
$(1+1/\poly(n))$-approximate solutions with only a $\poly(\log n)$
factor overhead in time) for solving these problems has been harder to
come by, and for many flow problems has been based on interior point
methods \cite{DaitchS08}.  E.g. the best running time for maximum flow
stands at
$\Otil(\min(m\sqrt{n}, n^{\omega}+ n^{2+1/6}))$~\cite{LeeS14,CLS19}
and $\Otil(m^{4/3})$ \deeksha{Does this have $o(1)$ in exponent?}for unit-capacity graphs~\cite{Madry13, LiuS20a,
  LiuS20b}.  Other results making progress in this direction include
works on shortest paths with small range negative
weights~\cite{CohenMSV17}, and
matrix-scaling~\cite{CohenMTV17, AllenLOW17}. \sushant{Rasmus: what other
  references need to go here?}
Recently, there has been progress on the dense case.
In \cite{vdBLNPSSSW20}, the authors developed an algorithm for weighted bipartite
matching and transshipment running in $\Otil(m + n^{3/2})$ time.
This is a nearly-linear-time algorithm in moderately dense
graphs.

Bubeck~\etal~\cite{BubeckCLL18} restarted the study of faster
high-accuracy algorithms for the weighted $\ell_p$-norm objective,
$\cost(\ff) = \norm{\SS\ff}_{p},$ 
a natural intermediate objective between $\ell_2$ and $\ell_{\infty}.$
This result improved the running time significantly over classical
interior point methods \cite{NesterovN94} for $p$ close to 2.
Adil~\etal~\cite{AdilKPS19}
gave a high-accuracy algorithm for computing
$\ell_{p}$-norm minimizing flows in time
$\min\{m^{\frac{4}{3} + o(1)}, n^{\omega}\}$ for
$p \in (2, \sqrt{\log n}$].
%
Building on their work, Kyng~\etal~\cite{KyngPSW19} gave an
almost-linear-time high-accuracy algorithm for \emph{unit-weight}
$\ell_{p}$-norm minimizing flows $\cost(\ff) = \norm{\ff}_{p}^{p}$ for
large $p \in (\omega(1),\sqrt{\log n}].$ More generally, they give an
almost-linear time-high-accuracy algorithm for \emph{mixed}
$\ell_2^{2} + \ell_p^{p}$ objectives as long as the $\ell_p^{p}$-norm
is unit-weight, i.e.,
\[\cost(\ff) = \|\RR^{\frac{1}{2}} \ff\|_2^2 + \norm{\ff}_{p}^{p}.\]
Their algorithm for $(\ell_2^{2} + \ell_p^p)$-minimizing flows was subsequently used as a key ingredient in recent
results improving the running time for high-accuracy/exact maximum
flow on unit-capacity graphs to $m^{4/3+o(1)}$~\cite{LiuS20a,
  LiuS20b}.

In this paper, we obtain a nearly-linear running time for weighted
$\ell^{2}_2 + \ell^{p}_p$ flow/voltage problems on graphs. Our algorithm requires 
$p(m^{1+o(1)} + n^{4/3 + o(1)})$ time for $p=\omega(1)$ which is
almost-linear-time for $p \leq m^{o(1)}$ in slightly dense graphs,
($m \ge n^{4/3 + o(1)}$). \sushant{Should we write the bound for all
  $p?$}

Our running time $m^{1+o(1)} + n^{4/3+o(1)}$ is even better than the $\Otil(m + n^{3/2})$ time
obtained for bipartite matching in \cite{vdBLNPSSSW20}.
Our result beats the $\Omega(n^{3/2})$ barrier
  that arises in \cite{vdBLNPSSSW20} from the use of interior point
  methods that maintain a vertex  dual solution using dense updates across
  $\sqrt{n}$ iterations.
  The progress on bipartite matching relies on
  highly technical graph-based inverse
maintenance techniques that are tightly interwoven with interior point
method analysis.
In constrast, our sparsification methods provide a clean interface to
iterative refinement, which makes our analysis much more simple and
compact.

\medskip{\textbf{Graph Sparsification.}
Various notions of graph sparsification -- replacing a dense graph
with a sparse one, while approximately preserving some key properties
of the dense graph  -- have been key ingredients in faster
low-accuracy algorithms.
%
Bencz\'{u}r and Karger~\cite{BenczurK96} defined cut-sparsifiers that
approximately preserve all cuts, and used them to give faster
low-accuracy approximation algorithms for maximum-flow. Since then,
several notions of sparsification have been studied extensively and
utilized for designing faster algorithms~\cite{PelegS89, SpielmanT11,
  Racke08, SpielmanS11, Madry10, Sherman13, KelnerLOS14, RackeST14,
  CohenP15, KyngLPSS16, DurfeePPR17, ChuGPSSW18}.

Sparsification has had a smaller direct impact on the design of faster
high-accuracy algorithms for graph problems, limited mostly to the
design of linear system solvers~\cite{SpielmanT04, KoutisMP11,
  PengS14, KyngLPSS16}.
Kyng~\etal~\cite{KyngPSW19} constructed sparsifiers for weighted
$\ell_2^{2}$ + unweighted $\ell_{p}^{p}$-norm objectives for flows.
In this paper, we develop almost-linear time algorithms for building
sparsifiers for weighted $\ell_2^{2} + \ell_{p}^{p}$ norm objectives
for flows and voltages,
\[
  \cost(\ff) = \|\RR^{\frac{1}{2}} \ff \|_2^{2} + \|\SS \ff \|_p^{p}
  ,\text{ and } \cost'(\BB\vv) = \|\WW^{\frac{1}{2}} \BB \vv \|_2^{2} + \|\UU \BB \vv \|_p^{p},
\]
and utilize them as key ingredients in our faster high-accuracy
algorithms for optimizing such objectives on graphs.
Our construction of sparsifiers for flow objectives builds on the
machinery from~\cite{KyngPSW19}, and our construction of sparsifiers
for voltage objectives builds on graph spanners~\cite{AlthoferDDJS93,
  PelegS89, BaswanaS07}.

\section{Our Results}\label{sec:results}

Our main results concern flow and voltage problems for mixed
$(\ell^{2}_2 + \ell_p^{p})$-objectives for $p \geq 2$. Since our
algorithms work best for large $p$, we restrict our attention to
$p = \omega(1)$ in this overview. Section~\ref{sec:Meta-Algo} provides
detailed running times for all $p \geq 2$.  We emphasize that by
setting the quadratic term to zero in our mixed
$(\ell^{2}_2 + \ell_p^{p})$-objectives, we get new state of the art
algorithms for $\ell_p$-norm miniziming flows and voltages.

\medskip\noindent{\textbf{Mixed $\ell_2$-$\ell_p$-norm minimizing flow.}
Consider a graph $G = (V,E)$ along with non-negative diagonal matrices
$\RR, \SS \in \rea^{E \times E},$ and a
gradient vector $\gg \in \rea^{E}$, as well as demands
$\dd \in \rea^{V}$. We refer to the diagonal entries of $\RR$ and
$\SS$ as $\ell_2$-weights and $\ell_p$-weights respectively.  Let $\BB$ denote
the signed edge-vertex incidence of $G$ (see Appendix \ref{sec:prelims}).  We wish to solve the following
minimization problem with the objective
$\obj(\ff) = \gg^{\top}\ff + \|\RR^{\nfrac{1}{2}} \ff\|_2^2 +
\norm{\SS \ff}_p^p$
\begin{equation}
  \min_{
    \BB^{\top} \ff = \dd
  }
  \obj(\ff)\label{eq:flowprob}
\end{equation}
We require
$\gg \perp \left\{\ker(\RR) \cap \ker(\SS) \cap \ker(\BB)\right\}$ so
that the problem has bounded minimum value, and $\dd \perp \vone$ so a
feasible solution exists.  These conditions can be checked in linear
time and have a simple combinatorial interpretation.
Note that the choice of graph edge directions in $\BB$ matters for the
value of $\gg^{\top}\ff,$. The flow on an edge is allowed to be
both positive or negative.

\medskip\noindent{\textbf{Mixed $\ell_2$-$\ell_p$-norm minimizing
    voltages.}  Consider a graph $G = (V,E)$ along with non-negative
  diagonal matrices $\WW \in \rea^{E \times E}$ and
  $\UU \in \rea^{E \times E},$ and demands $\dd \in \rea^{V}$.  We
  refer to the diagonal entries of $\WW$ and $\UU$ as $\ell_2$-conductances
  and $\ell_p$-conductances respectively.  In this case, we want to
  minimize the objective
  $\obj(\vv) = \dd^{\top}\vv + \|\WW^{\nfrac{1}{2}} \BB \vv\|_2^2 +
  \norm{\UU \vv}_p^p$ in minimization problem
\begin{equation}
  \min_{
    \vv
  }
  \obj(\vv)\label{eq:voltageprob}
\end{equation}
In the voltage setting, we only require  $\dd \perp \vone$ so the
problem has bounded minimum value.

\medskip\noindent{\textbf{Obtaining good solutions.} For both these problems, we
study high accuracy approximation algorithms that provide feasible
solutions $\xx$ (a flow or a voltage respectively), that approximately
minimize the objective function from some starting point $\xx^{(0)}$,
i.e., for some small $\eps > 0,$ we have
\[
  \obj(\xx) - \obj (\xx^\star) \leq \eps ( \obj(\xx^{(0)}) - \obj
  (\xx^\star) )
\]
wher $\xx^\star$ denotes an optimal feasible solution.  Our algorithms
apply to problems with quasipolynomially bounded parameters, including
quasipolynomial bounds on non-zero singular values of matrices we work with.
Below we state our main algorithmic results.
\begin{theorem}[Flow Algorithmic Result]\label{thm:flow-short}
  Consider a graph $G$ with $n$ vertices and $m$ edges, equipped with
  non-negative $\ell_2$ and $\ell_p$-weights, as well as a
  gradient and demands, all with quasi-polynomially bounded entries.
  For $p=\omega(1)$, in
  $p(m^{1+o(1)} + n^{4/3 + o(1)}) \log^{2} \nfrac{1}{\eps}$ time we
  can compute an $\epsilon$-approximately optimal flow solution to
  Problem~\eqref{eq:flowprob} with high probability.
\end{theorem}
This improves upon \cite{AdilKPS19, AdilPS19, AdilS20}
which culminated in a $pm^{4/3 + o(1)} \log^{2} \nfrac{1}{\eps}$ time
algorithm.
\begin{theorem}[Voltage Algorithmic Result]
  \label{thm:voltage-short}
  Consider a graph $G$ with $n$ vertices and $m$ edges, equipped with
  non-negative $\ell_2$ and $\ell_p$-conductances, as well as demands, all with quasi-polynomially bounded entries.
  For $p=\omega(1)$, in
  $p(m^{1+o(1)} + n^{4/3 + o(1)}) \log^{2} \nfrac{1}{\eps}$ time we
  can compute an $\epsilon$-approximately optimal voltage solution to
  Problem~\eqref{eq:voltageprob} with high probability.
\end{theorem}
\sushant{Should we write these theorems with the $p$ dependence in the
  power?}

\medskip\noindent{\textbf{Background: Iterative Refinement for Mixed
    $\ell_2$-$\ell_p$-norm Flow Objectives.}
  Adil~\etal~\cite{AdilKPS19} developed a notion of iterative
  refinement for mixed $(\ell^{2}_2 + \ell^{p}_p)$-objectives which in
  the flow setting, i.e.  Problem~\eqref{eq:flowprob}, corresponds to
  approximating $\obj'(\ddelta) = \obj(\ff+\ddelta)$ using another
  $(\ell_2^2+\ell_p^p)$-objective which roughly speaking corresponds
  to the 2nd degree Taylor series approximation of $\obj'(\ddelta)$
  combined with an $\ell_p$-norm term $\norm{\SS \ddelta}_p^p$, while
  ensuring feasibility of $\ff+\ddelta$ through a constraint
  $\BB \ddelta = \vzero$.  We call the resulting problem a
  \emph{residual} problem.  Adil~\etal~\cite{AdilKPS19} showed that
  obtaining a constant-factor approximate solution to the residual
  problem in $\delta$ is sufficient to ensure that $\obj(\ff+\ddelta)$
  is closer to the optimal solution by a multiplicative factor
  depending only on $p$.  In~\cite{AdilPS19}, this result was
  sharpened to show that such an approximate solution for the residual
  problem can be used to make $(1-\Omega(1/p))$ multiplicative
  progress to the optimum, so that $O(p \log(m/\epsilon))$ iterations
  suffice to produce an $\epsilon$-accurate solution.

  In order to solve the residual problem to a constant approximation,
  Adil~\etal~\cite{AdilKPS19} developed an accelerated multiplicative
  weights method for $(\ell^{2}_2 + \ell^{p}_p)$-flow objectives, or
  more generally, for mixed $(\ell^{2}_2 + \ell^{p}_p)$-regression in
  an underconstrained setting.

  \medskip\noindent{\textbf{Sparsification results.}  Our central
    technical results in this paper concern sparsification of residual
    flow and voltage problems, in the sense outlined in the previous
    paragraph.  Concretely, in nearly-linear time, we can take a
    residual problem on a dense graph and produce a residual problem
    on a sparse graph with $\Otil(n)$ edges, with the property that
    constant factor solutions to the sparse residual problem still
    make $(1-\Omega(m^{-\frac{2}{p-1}}p))$ multiplicative progress on
    the original problem.  This leads to an iterative refinement that
    converges in $O(p m^{\frac{2}{p-1}} \log(m/\epsilon))$ steps.
    However, the accelerated multiplicative weights algorithm that we
    use for each residual problem now only requires $\Otil(n^{4/3})$
    time to compute a crude solution.

\medskip\noindent{\textbf{Flow residual problem sparsification.}
  In the flow setting, we show the following:
\begin{theorem}[Informal Flow Sparsification Result]\label{thm:flow-spars-short}
  Consider a graph $G$ with $n$ vertices and $m$ edges, equipped with
  non-negative $\ell_2$ and $\ell_p$-weights, as well as a
  gradient. In $\Otil(m)$ time, we can compute a graph $H$ with $n$
  vertices and $\Otil(n)$ edges, equipped with non-negative
  $\ell_2$ and $\ell_p$-weights, as well as a
  gradient, such that a constant factor approximation to the flow
  residual problem on $H$, when scaled by $m^{\frac{-1}{p-1}}$ results
  in an $\Otil(m^{\frac{2}{p-1}})$ approximate solution to the flow
  residual problem on $G$. The algorithm works for all $p \geq 2$ and
  succeeds with high probability.
\end{theorem}

\sushant{I don't like this line, undercuts what we do.}
Our sparsification techniques build on \cite{KyngPSW19}, require a new
bucketing scheme to deal with non-uniform $\ell_p$-weights,
as well as a prepreprocessing step to handle cycles with zero
$\ell_2$-weight and $\ell_p$-weight. This preprocessing scheme in turn
necessitates a more careful analysis of additive errors introduced
by gradient rounding, and we provide a more powerful framework for
this than \cite{KyngPSW19}.

\medskip\noindent{\textbf{Voltage residual problem sparsification.}
  In the voltage setting, we show the following.
\begin{theorem}[Voltage Sparsification Result (Informal)]\label{lem:vol-sparse-inf}
  Consider a graph $G$ with $n$ vertices and $m$ edges, equipped with
  non-negative $\ell_2$ and $\ell_p$-conductances. In $\Otil(m)$ time,
  we can compute a graph $H$ with $n$ vertices and $\Otil(n)$ edges,
  equipped with non-negative $\ell_2$ and $\ell_p$-conductances, such
  that constant factor approximation to the voltage residual problem
  on $H$, when scaled by $m^{\frac{-1}{p-1}}$ results in an
  $\Otil(m^{\frac{1}{p-1}})$ approximate solution to the voltage
  residual problem on $G$. The algorithm works for all $p \geq 2$ and
  succeeds with high probability.
\end{theorem}
Note that our voltage sparsification is slightly stronger than our
flow sparsification, as the former loses only a factor
$\Otil(m^{\frac{1}{p-1}})$ in the approximation while the latter loses
a factor $\Otil(m^{\frac{2}{p-1}})$.
Our voltage sparsification uses a few key observations: In voltage
space, surprisingly, we can treat treat the $\ell_2$ and $\ell_p$ costs
separately.
This behavior is very different than the flow case, and arises becase
in voltage space, every edge provides an ``obstacle'', i.e. adding an
edge increases cost, whereas in flow space, every edge provides an
``opportunity'', i.e. adding an edge decreases cost.
This means that in voltage space, we can separately account for the
energy costs created by our $\ell_2$ and $\ell_p$ terms, whereas in
flow space, the $\ell_2$ and $\ell_p$ weights must be highly
correlated in a sparsifier.
Armed with this decoupling observation, 
we preserve $\ell_2$ cost using standard tools for spectral
graph sparsification, and we preserve $\ell_p$ cost approximately by
a reduction to graph distance preservation, which we in turn achieve
using weighted undirected graph spanners.

\medskip\noindent{\textbf{Voltage space accelerated multiplicative
    weights solver.}
  The algorithm from \cite{AdilKPS19} for constant approximate
  solutions to the residual problem works in the flow setting.  Using
  iterative refinement, the algorithm could be used to compute
  high-accuracy solutions. Because we can use high-accuracy flow
  solutions to extract high-accuracy solutions to the dual voltage
  problem, \cite{AdilKPS19} were also able to produce solutions to
  $\ell_q$-norm minimizing voltage problems (where $\ell_{q}$ for
  $q = p/(p-1)$ is the dual norm to $\ell_{p}$).  Hence, by solving
  $\ell_p$-flow problems for all $p \in (2,\infty)$, \cite{AdilKPS19}
  were able to solve $\ell_q$-norm minimizing voltage problems for all
  $q \in (1,2)$.

  Our sparsification of flow and voltage problems works only for
  $p \geq 2.$ Thus, in order to solve for $q$-norm minimizing voltages
  for $q > 2,$ we require a solver that works directly in voltage
  space for mixed $(\ell^{2}_2 + \ell^{p}_p)$-voltage objectives.

  We develop an accelerated multiplicative weights algorithm along the
  lines of~\cite{ChristianoKMST10, ChinMMP13, AdilKPS19} that works
  directly in voltage space for mixed
  $(\ell^{2}_2 + \ell^{p}_p)$-objectives, or more generally for
  overconstrained mixed $(\ell^{2}_2 + \ell^{p}_p)$-objective
  regression.  Concretely, this directly gives an algorithm for
  computing crude solutions to the residual problems that arise from
  applying \cite{AdilKPS19} iterative refinement to
  Problem~\eqref{eq:voltageprob}. Our solver produces an improved
  $O(1)$-approximation to the residual problem rather than a
  $p^{O(p)}$-approximation from 
  \cite{AdilKPS19}.
  This
  gives an $\Otil(m^{4/3})$ high-accuracy algorithm for mixed
  $(\ell^{2}_2 + \ell^{p}_p)$-objective voltage problems for $p > 2$, unlike
  \cite{AdilKPS19}, which could only solve pure $p >2$ voltage
  problems.
  We then speed this up to a $p(m^{1+o(1)} + n^{4/3+ o(1)})$ time
  algorithm for $p = \omega(1)$ by developing a sparsification
  procedure that applies directly to mixed
  $(\ell^{2}_2 + \ell^{p}_p)$-voltage problems for $p >
  2$. \deeksha{Should also add, this implies the fastest known direct
    solver for $\ell_p$ regression with total $p$ dependence $p^3$
    without the homotopy.}



  \medskip\noindent{\textbf{Mixed $\ell_2$-$\ell_p$-norm regression.}
    Our framework can also be applied outside of a graph setting,
    where our new accelerated multiplicative weights algorithm for
    overconstrained mixed $(\ell^{2}_2 + \ell^{p}_p)$-regression gives
    new state-of-the-art results in some regimes when combined with
    new sparsification results.  In this setting we develop
    sparsification techniques based on the Lewis weights sampling from
    the work of Cohen and Peng~\cite{CohenP15}.  We focus on the case
    $2 < p < 4$, where \cite{CohenP15} provided fast algorithms for
    Lewis weight sampling.
\begin{theorem}[General Matrices Sparsification Result]\label{thm:LewisWt}
Let $p \in [2,4)$, let $\MM \in \mathbb{R}^{m_1\times n}, \NN\in \mathbb{R}^{m_2 \times n}$ be matrices, $m_1,m_2 \geq n$, and let $\LSS(\BB)$ denote the time to solve a linear system in $\BB^\top \BB$. Then, we may compute $\widetilde{\MM}, \widetilde{\NN} \in \mathbb{R}^{O(n^{p/2}\log(n))\times n}$ such that with probability at least $1 -\frac{1}{n^{\Omega(1)}}$, for all $\Delta \in \rea^n$,
\[
\|\widetilde{\MM}\Delta\|_2^2 + \|\widetilde{\NN}\Delta\|_p^p \approx_{O(1)} \|\MM\Delta\|_2^2 + \|\NN\Delta\|_p^p,
\]
in time $\tilde{O}\left(\nnz(\MM) + \nnz(\NN) + \LSS(\MMhat) + \LSS(\NNhat)\right)$, for some $\MMhat$ and $\NNhat$ each containing $O(n\log(n))$ rescaled rows of $\MM$ and $\NN$, respectively.
\end{theorem}

\begin{theorem}[General Matrices Algorithmic Result]\label{thm:LewisMain}
For $p \in [2,4)$, with high probability we can find an $\eps$-approximate solution to \eqref{eq:Problem} in time
\[
\Otil\left(\left(\nnz(\MM) + \nnz(\NN) + \left(\LSS(\MMtil) + \LSS(\NNtil)\right)n^{\frac{p(p-2)}{6p-4}}\right)\log^2(1/\eps)\right),
\]
for some $\MMtil$ and $\NNtil$ each containing
$O(n^{p/2}\log(n))$ rescaled rows of $\MM$ and $\NN$, respectively, where
$\LSS(\AA)$ is the time required to solve a linear
equation in $\AA^{\top} \AA$ to quasipolynomial accuracy.
\end{theorem}
Note that for all $p \in (2,4),$ we have that the exponent
 $\frac{p(p-2)}{6p-4} \le 0.4.$
\begin{remark}
  By \cite{CohenLMMPS15}, a linear equation in $\AA^{\top} \AA$, where
  $\AA \in \rea^{m \times n}$ can be
  solved to quasipolynomial accuracy in time $\Otil(\nnz(\AA) +
  n^{\omega})$.
\end{remark}
Using the above result for solving the required linear systems, we get
a running time of
$\Otil(\nnz(\MM) + \nnz(\NN) + (n^{p/2} + n^{\omega})
n^{\frac{p(p-2)}{6p-4}}),$ matching an earlier input
sparsity result by Bubeck~\etal~\cite{BubeckCLL18} that achieves
$\Otil((\nnz(\MM) + \nnz(\NN))(1 + n^{\frac{1}{2}} m^{-\frac{1}{p}}) +
m^{\frac{1}{2}-\frac{1}{p}}n^2 + n^{\omega}),$ where
$\MM \in \rea^{m_1 \times n}, \NN \in \rea^{m_2 \times n}$ and
$m = \max\{m_1, m_2\}.$

\newcommand{\lewis}[1]{\ttau_{#1}}
\newcommand{\approxlewis}[1]{\widetilde{\ttau}_{#1}}

\section{Main Algorithm}
\label{sec:Meta-Algo}

%
%
%
In this section, we prove Theorems \ref{thm:flow-short},
\ref{thm:voltage-short}, and \ref{thm:LewisMain}. We first design an
algorithm to solve the following general problem:
 \begin{definition}
   For matrices $\MM \in \mathbb{R}^{m_1 \times n}$,
   $\NN \in \mathbb{R}^{m_2 \times n}$ and
   $\AA \in \mathbb{R}^{d \times n}$, $m_1,m_2 \geq n$,$d \leq n$, and
   vectors
   $\bb \perp \left\{\ker(\MM) \cap \ker(\NN) \cap \ker(\AA)\right\}$
   and $\cc \in \text{im}(\AA)$, we want to solve
 \begin{align}\label{eq:Problem}
 \min_{\xx}& \quad \bb^{\top}\xx + \|\MM\xx\|_2^2 + \|\NN\xx\|_p^p\\
 \text{s.t.} & \quad \AA\xx = \cc. \nonumber
 \end{align}
 \end{definition}
 In order to solve the above problem, we use the iterative refinement framework from \cite{AdilPS19} to obtain a residual problem which is defined as follows.
\begin{definition}
   \label{def:residual}
   For any $p \geq 2$, we define the residual problem $\res(\Delta)$,
   for \eqref{eq:Problem} at a feasible $\xx$ as,
   \[
     \max_{\AA\Delta = 0} \quad \res(\Delta) \defeq \gg^{\top}\Delta -
     \Delta^{\top}\RR\Delta - \|\NN\Delta\|_p^p, \text{    where,}
   \] 
   \[\gg = \frac{1}{p}\bb + \frac{2}{p}\MM^{\top} \MM\xx +
     |\NN\xx|^{p-2} \NN\xx \quad \text{and} \quad\RR =\frac{2}{p^2}\MM^{\top}\MM+ 2\NN^{\top} Diag(|\NN\xx|^{p-2})\NN.\]
\end{definition}
This residual problem can further be reduced by moving the term linear in $\xx$ to the constraints via a binary search. This leaves us with a problem of the form, 
  \begin{align*}
  \begin{aligned}
  \min_{\Delta} & \quad \Delta^{\top}\RR\Delta + \|\NN\Delta\|_p^p \\
  \text{s.t.} & \quad  \gg^{\top}\Delta = a, \AA\Delta = 0,
  \end{aligned}
\end{align*}
for some constant $a$.

In order to solve the above problem with $\ell_2^2 + \ell_p^p$
objective, we reduce the instance size via a sparsification routine,
and then solve the smaller problem by a multiplicative weights
algorithm. We adapt the multiplicative-weights algorithm from
\cite{AdilKPS19} to work in the voltage space while improving the $p$
dependence of the runtime from $p^{O(p)}$ to $p$, and the
approximation quality from $p^{O(p)}$ to $O(1)$. The precise
sparsification routines are described in later sections.


For large $p,$ i.e., $p > \log m$, in order to get a linear
dependence on the running time on $p,$ we need to reduce the residual
problem in $\ell_p$-norm to a residual problem in $\log m$-norm
by using the framework from \cite{AdilS20}.

The entire meta-algorithm is described formally in Algorithm
\ref{alg:Meta-Algo}, and its guarantees are described by the next
theorem. Most proof details are deferred to Appendix
\ref{sec:lp-reg}.
\begin{restatable}{theorem}{MetaThm}
\label{thm:Meta-Thm}
For an instance of Problem \eqref{eq:Problem}, suppose we are given a starting solution $\xx^{(0)}$ that satisfies $\AA\xx^{(0)} = \cc$ and is a $\kappa$ approximate solution to the optimum. Consider an iteration of the while loop, line \ref{alg:line:while} of Algorithm \ref{alg:Meta-Algo} for the $\ell_p$-norm residual problem at $\xx^{(t)}$. We can define $\mu_1$ and $\kappa_1$ such that if $\Dbar$ is a $\beta$ approximate solution to a corresponding $p'$-norm residual problem, then $\mu_1\Dbar$ is a $\kappa_1$-approximate solution to the $p$-residual problem. Further, suppose we have the following procedures,
\begin{tight_enumerate}
\item  {\sc Sparsify}: Runs in time $K$, takes as input any matrices $\RR,\NN$ and vector $\gg$ and returns $\widetilde{\RR},\widetilde{\NN},\widetilde{\gg}$ having sizes at most $\tilde{n} \times n$ for the matrices \bab{Need to be careful about $\widetilde{\RR}$ vs. $\widetilde{\RR}^\top \widetilde{\RR}$}, such that if $\Dtil$ is a $\beta$ approximate solution to,
 \[
 \max_{\AA\Delta = 0} \quad  \widetilde{\gg}^{\top}\Delta - \|\widetilde{\RR}\Delta\|_2^2 - \|\widetilde{\NN}\Delta\|_{p'}^{p'}, 
 \]
 for any $p' \geq 2$, then $\mu_2\Dtil$, for a computable $\mu_2$ is a $\kappa_2\beta$-approximate solution for,
  \[
 \max_{\AA\Delta = 0} \quad  \res(\Delta) \defeq \gg^{\top}\Delta - \|\RR^{1/2}\Delta\|_2^2 - \|\NN\Delta\|_{p'}^{p'}.
 \] 
\item {\sc Solver}: Approximately solves \eqref{eq:InsideProb} to
  return $\Dbar$ such that
  $ \|\widetilde{\RR}\Dbar\|_2^2 \leq \kappa_3 \nu$ and
  $\|\widetilde{\NN}\Delta\|_{p}^{p} \leq \kappa_4 \nu$ in time
  $\tilde{K}(\tilde{n})$ for instances of size at most $\tilde{n}$.
\end{tight_enumerate} 
Algorithm \ref{alg:Meta-Algo} finds an $\epsilon$-approximate solution for Problem \eqref{eq:Problem} in time 
\[
 \widetilde{O}\left(p\kappa_4^{1/(p-1)}\kappa_3 \kappa_2\kappa_1 (K + \tilde{K}(\tilde{n}))\log\left(\frac{\kappa p}{\epsilon}\right)^2\right).
\]
\end{restatable}

\begin{algorithm}[h]
\caption{Meta-Algorithm for $\ell_p$ Flows and Voltages}
\label{alg:Meta-Algo}
 \begin{algorithmic}[1]
 \Procedure{\textsc{Sparsified-p-Problems}}{$\AA, \MM,\NN, \cc, \bb, p$}
 \State $\xx \leftarrow \xx^{(0)}$, such that $\ff\left(\xx^{(0)}\right) \leq \kappa \opt$
 \State $T  \leftarrow \widetilde{O}\left(p \kappa_1\kappa_2\kappa_3 \log\left(\frac{\kappa}{\epsilon}\right)\right)$
\For{$t = 0$ to $T$}
 \State At $\xx^{(t)}$ define $\gg,\RR,\NN$ and $\res(\Delta)$, the residual problem (Definition \ref{def:residual})
 \State $a  \leftarrow \frac{1}{2},b \leftarrow 1, \mu_1  \leftarrow 1, \kappa_1  \leftarrow 1$
 \State $\nu  \leftarrow \ff\left(\xx^{(0)}\right)$
\While{$ \nu \geq \epsilon \frac{ \ff\left(\xx^{(0)}\right)}{\kappa p}$}\label{alg:line:while}
\If{$p > \log m$} \Comment{Convert $\ell_p$-norm residual to $\log m$-norm residual}
\State $p'  \leftarrow \log m$
\State $\NN'  \leftarrow \frac{1}{2^{1/p'}} \left(\frac{\nu}{m} \right)^{\frac{1}{p'}-\frac{1}{p}}\NN$
\State $ a  \leftarrow \frac{1}{33}, b  \leftarrow O(1)m^{o(1)}$
\State $ \mu_1 \leftarrow m^{-o(1)}, \kappa_1 \leftarrow m^{o(1)}$  \Comment{Lose $\kappa_1$ in approx. when scaled by $\mu_1$}
\State $(\widetilde{\gg},\widetilde{\RR},\widetilde{\NN})  \leftarrow \text{{\sc Sparsify$(\gg,\RR,\NN')$}}$ \Comment{Lose $\kappa_2$ in approx. when scaled by $\mu_2$}
\Else
\State $(\widetilde{\gg},\widetilde{\RR},\widetilde{\NN})  \leftarrow \text{{\sc Sparsify$(\gg,\RR,\NN)$}}$ \Comment{Lose $\kappa_2$ in approx. when scaled by $\mu_2$}
\State $p' \leftarrow p$
\EndIf
\State Use {\sc Solver} to compute $\kappa_3,\kappa_4$ approximate solution to
 \begin{align}
\label{eq:InsideProb}
\centering
\begin{aligned}
\Dtil^{(\nu)} \leftarrow \arg\min_{\Delta}& \quad \|\widetilde{\RR}^{1/2}\Delta\|_2^2 + \|\widetilde{\NN}\Delta\|_{p'}^{p'}\\
\text{s.t.}& \quad \widetilde{\gg}^{\top}\Delta = a\nu, \quad\AA\Delta = 0.
\end{aligned}
\end{align}
\State $\Dbar^{(\nu)}  \leftarrow\frac{a}{2b\kappa_3 \kappa_4^{1/(p'-1)}}\mu_2\mu_1 \Dtil^{(\nu)}$
\State $ \nu  \leftarrow\nu/2$
\EndWhile
\State $\Delta \leftarrow \arg\min_{\Dbar^{(\nu)}}\quad \ff\left(\xx-\frac{\Dbar^{(\nu)}}{p}\right)$
\State $\xx  \leftarrow \xx-\frac{\Delta}{p}$
\EndFor
 \State\Return $ \xx$
 \EndProcedure 
 \end{algorithmic}
\end{algorithm}

\subsection{Algorithms for $\ell_p$-norm
  Problems}\label{sec:MainAlgoProofs}
The problems discussed in Section \ref{sec:results} are special cases
of Problem \eqref{eq:Problem}, which means we can use Algorithm
\ref{alg:Meta-Algo}. To prove our results, we will utilize Theorem
\ref{thm:Meta-Thm}, with the respective sparsification procedures and
the following multiplicative-weights based algorithm for solving
problems of the form,
\begin{align}
  \label{eq:res}
 \min_{\Delta}  & \quad \Delta^{\top}\MM^{\top}\MM\Delta + \|\NN\Delta\|_p^p \\
 & \text{s.t.} \quad \AA\Delta = \cc. \nonumber
 \end{align}
 We describe our solver formally and prove the following theorem about
 its guarantees in Appendix \ref{sec:AKPSOracle}.
\begin{restatable}{theorem}{AKPSAlgo}
\label{cor:ResidualDecision}
Let $p\geq 2$.
Consider an instance of Problem \eqref{eq:res} described by matrices
$\AA \in \mathbb{R}^{d \times n},\NN \in \mathbb{R}^{m_1 \times n},\MM
\in \mathbb{R}^{m_2 \times n}$, $d \le n \le m_1,m_2$, and vector
$\cc \in \rea^{d}$.
If the optimum of this problem is at most $\nu$, Procedure
\textsc{Residual-Solver} (Algorithm~\ref{alg:FasterOracleAlgorithm}) returns an $\xx$ such that
$\AA \xx = \cc,$ and $\xx^{\top}\MM^{\top}\MM\xx \leq O(1)\nu$ and
$\|\NN\xx\|_p^p \leq O(3^p)\nu$.
The algorithm makes ${O} \left(pm_1^{\frac{p-2}{(3p-2)}}\right)$ calls
to a linear system solver.
\end{restatable}

We utilize Procedure \textsc{Residual-Solver} as the Procedure
\textsc{Solver} in Algorithm \textsc{Sparsified-p-Problems}. The
algorithm uses the procedure only for solving problems instances with
$p \leq \log m.$ Thus, its running time is
$\tilde{K}(\tilde{n}) = \Otil \left(\tilde{n}^{\frac{p-2}{3p-2}} \cdot
  \LSS(\tilde{n})\right) \leq \Otil \left(\tilde{n}^{1/3} \cdot
  \LSS(\tilde{n})\right) $, where $\LSS(\tilde{n})$ denotes the time
required to solve a linear system in matrices of size $\tilde{n}$. We
also have, $\kappa_3 = O(1), \kappa_4^{1/(p-1)} = O(1)$.

We next estimate the values of $\kappa_1$ and $\mu_1$. If
$p \leq \log m$, we have $\mu_1 = 1$ and $\kappa_1 =
1$. Otherwise, $\mu_1 = \Otil(1)$ and $\kappa_1 = O(m^{o(1)})$ (Refer
to Lemma \ref{lem:p-to-q} in Appendix \ref{sec:lp-reg}).

In order to obtain an initial solution, we usually solve an
$\ell_2$-norm problem. This gives an $m^{p/2}$ approximate initial
solution which results in a factor of $p^2$ in the running time. To
avoid this, we can do a homotopy on $p$ similar to \cite{AdilS20},
i.e., start with an $\ell_2$ solution and solve the $\ell_{2^2}$
problem to a constant approximation, followed by
$\ell_{2^3},..\ell_p$. We note that a constant approximate solution to
the $\ell_{p/2}$-norm problem gives an $O(m)$ approximation to the
$\ell_p$ problem and thus, we can solve $\log p$ problems where we can
assume $\kappa = O(m)$.

We now complete the proof of our various algorithmic results by
utilizing sparsification procedures specific to each problem.

\subsubsection*{$\ell_p$ Flows}
We will prove Theorem \ref{thm:flow-short} (Flow Algorithmic Result), with explicit $p$ dependencies.
%
\begin{proof}
From Theorem \ref{thm:flow-spars-short}, we obtain a sparse graph in $K = \Otil(m)$ time with $\tilde{n} = \Otil(n)$ edges. A constant factor approximation to the flow residual problem on this sparse graph when scaled by $\mu_2 = m^{-\frac{1}{p-1}}$ gives a $\kappa_2 = \Otil\left(m^{\frac{2}{p-1}}\right)$-approximate solution to the flow residual problem on the original graph. We can solve linear systems on the sparse graph in $\Otil(\tilde{n}) = \Otil(n)$ time using fast Laplacian solvers. Using all these values in Theorem \ref{thm:Meta-Thm}, we get the final runtime to be $pm^{\frac{2}{p-1}+o(1)}\left(m+n^{1+\frac{p-2}{3p-2}}\right) \log^2\left(\frac{pm}{\epsilon}\right)$ as claimed. We prove Theorem \ref{thm:flow-spars-short} in Section \ref{sec:Flow-Sparse}.
\end{proof}

\subsubsection*{$\ell_p$ Voltages}
We will prove Theorem \ref{thm:voltage-short} (Voltage Algorithmic Result), with explicit $p$ dependencies.
%
\begin{proof}
From Theorem \ref{lem:vol-sparse-inf}, we obtain a sparse graph in $K = \Otil(m)$ time with $\tilde{n} = \Otil(n)$ edges. A constant factor approximation to the voltage residual problem on this sparse graph when scaled by $\mu_2 = m^{-\frac{1}{p-1}}$ gives a $\kappa_2 = \Otil\left(m^{\frac{1}{p-1}}\right)$-approximate solution to the voltage residual problem on the original graph. We can solve linear systems on the sparse graph in $\Otil(\tilde{n}) = \Otil(n)$ time using fast Laplacian solvers. Using  these values in Theorem \ref{thm:Meta-Thm}, we get the final runtime to be $pm^{\frac{1}{p-1}+o(1)}\left(m+n^{1+\frac{p-2}{3p-2}}\right) \log^2\left(\frac{pm}{\epsilon}\right)$ as claimed. We prove Theorem \ref{lem:vol-sparse-inf} in Section \ref{sec:voltage-spars}.
\end{proof}

\subsubsection*{General Matrices}
We will now prove Theorem \ref{thm:LewisMain}.
\begin{proof}
  We assume Theorem \ref{thm:LewisWt}, which we prove in Appendix
  \ref{sec:lewis-wt}. From the theorem, we have $\kappa_2 = O(1)$ and
  $\mu_2 = O(1)$. Note that $K = \LSS(\MMhat) + \LSS(\NNhat)$ for some
  $\MMhat, \NNhat \in \rea^{O(n\log(n)) \times n}$, which is the time
  required to solve linear systems in $\MMhat^\top \MMhat$ and
  $\NNhat^\top \NNhat$, respectively. Since, by
  Theorem~\ref{thm:LewisWt}, the size of $\MMtil$ and $\NNtil$ is
  $\tilde{n} = O(n^{p/2}\log(n))$, the cost from the solver in
  Theorem~\ref{cor:ResidualDecision} is
  $\tilde{O}_p\left(\left(\LSS(\MMtil) +
      \LSS(\NNtil)\right)n^{\frac{p(p-2)}{6p-4}}\right)$.
\end{proof}

\section{Construction of Sparsifiers for $\ell_2^{2} + \ell_{p}^{p}$
  Voltages}\label{sec:voltage-spars}
In this section, we prove a formal version of the voltage sparsification result (Theorem \ref{lem:vol-sparse-inf}):

\begin{theorem}
  \label{thm:spannersparsify}
  Consider a graph $G = (V,E)$ with
  non-negative $2$-weights $\ww \in \rea^E$
  and
  non-negative $p$-weights $\ss  \in \rea^E$,
  with $m$
  and $n$ vertices.
  We can produce a graph $H = (V,F)$ with edges $F \subseteq E$,
  $\ell_2$-weights $\uu \in \rea^F$, 
  and
  $\ell_p$-weights $\tt \in \rea^F$,
  such that with probability at least $1-\delta$ the graph  $H$ has $O(n \log(n/\epsilon))$ edges and
  \begin{equation}
    \frac{1}{1.5} \norm{\WW \BB_{G}\xx}_2 
    \leq 
    \norm{\UU \BB_{H}\xx}_2
    \leq
    1.5 
    \norm{\WW \BB_{G}\xx}_2
    \label{eq:Hspectralsparsify}
  \end{equation}
  and for any $p \in [1,\infty]$
 \begin{equation}
    \frac{1}
    {m^{1/p} \log(n)}
    \norm{\SS \BB_{G}\xx}_p    
    \leq 
    \norm{\TT \BB_{H}\xx}_p 
    \leq 
    \norm{\SS \BB_{G}\xx}_p
     \label{eq:Hpsparsify}
  \end{equation}
  where $\WW = \diag{\ww}, \UU = \diag{\uu}, \SS = \diag{\ss}, \TT = \diag{\tt}$.
  We denote the routine computing $H$ and $\uu,\tt$ by
  $\textsc{SpannerSparsify}$, so that $(H, \uu, \tt) =
  \textsc{SpannerSparsify}(G, \ww, \ss)$.
  This algorithm runs in $\Otil(m\log(1/\delta))$ time.
\end{theorem}


We will first define some terms required for our result. Given a undirected graph $G = (V,E)$, with edge lengths $\ll \in \rea^{E}$,
and $u,v \in V$, we let $d_G(u,v)$ denote the shortest path distance in $G$ w.r.t $\ll$,
so that if $P$ is the shortest path w.r.t $\ll$ then 
\[
  d_{G,\ll}(u,v)
  =
  \sum_{e \in P} \ll(e)
\]

\begin{definition}
  Given a undirected graph $G = (V,E)$ with edge lengths $\ll \in \rea^{E}$,
  a $K$-spanner is a subgraph $H$ of $G$ with the same edge lengths s.t.
  $d_H(u,v) \leq K d_G(u,v)$.
\end{definition}
Baswana and Sen showed the following result on spanners~\cite{BaswanaS07}.
\begin{theorem}
  \label{thm:spanner}
  Given an undirected graph $G = (V,E,\ll)$ with $m$ edges and $n$ vertices,
  and an integer $k > 1$, we can compute a $(2k-1)$-spanner $H$ of $G$
  with $O(n^{1+1/k})$ edges 
  in expected time $O(km)$.
\end{theorem}
\begin{lemma}
  \label{lem:spannerinfnorm}
  Given an undirected graph $G = (V,E)$ with positive edge lengths $\ll \in \rea^{E}$,
  and a $K$-spanner $H = (V,F)$ of $G$, for all $\xx \in \rea^V$ we
  have 
\[
  \max_{(u,v) \in F} \frac{1}{\ll(u,v)} \abs{\xx(u) - \xx(v)}
    \leq
    \max_{(u,v) \in E} \frac{1}{\ll(u,v)}\abs{\xx(u) - \xx(v)}
    \leq
    K \max_{(u,v) \in F} \frac{1}{\ll(u,v)}\abs{\xx(u) - \xx(v)}
    \]
  \end{lemma}
\begin{proof}
  The inequality $\max_{(u,v) \in F} \frac{1}{\ll(u,v)} \abs{\xx(u) - \xx(v)}
    \leq
    \max_{(u,v) \in E} \frac{1}{\ll(u,v)}  \abs{\xx(u) - \xx(v)}$
    is immediate from $F \subseteq E$.
    
 To prove the second inequality, we note that if $(u,v) \in E$ has
 shortest path $P$ in $H$ then
 \[
   \frac{1}{\ll(u,v)}\abs{\xx(u) - \xx(v)}
   \leq
   \frac{K}{\sum_{(z,y) \in P} \ll(z,y)}\abs{\sum_{(z,y) \in P}
     \xx(z)-\xx(y)}
   \leq
   \max_{(z,y) \in P}
   \frac{K}{\ll(z,y)}\abs{\xx(z) - \xx(y)}
   .
\]
\end{proof}

  


\begin{definition}
  Given a undirected graph $G = (V,E)$ with $m$ edges and $n$ vertices
  with positive edge $\ell_2$-weights $\ww \in \rea^{E}$,
  a spectral $\epsilon$-approximation of $G$ is a graph $H = (V,F)$ with
  $F \subseteq E$ with positive edge $\ell_2$-weights $\uu \in \rea^F$ s.t.
    \[
    \frac{1}{1+\epsilon} \norm{\WW \BB_{G}\xx}_2 
    \leq 
    \norm{\UU \BB_{H}\xx}_2
    \leq
    (1+\epsilon)
    \norm{\WW \BB_{G}\xx}_2 
  \]
  where $\WW = \diag{\ww}$ and $\UU = \diag{\uu}$.
\end{definition}

The following result on spectral sparsifiers was shown by Spielman and
Srivastava~\cite{SpielmanS11} (see also \cite{Spielman15}).
\begin{theorem}
  \label{thm:spectralsparsify}
  Given a graph $G = (V,E)$ with positive $\ell_2$-weights $\ww \in
  \rea^E$ with $m$ edges
  and $n$ vertices,
  for any $\epsilon \in (0,1/2]$, 
  we can produce a graph $H = (V,F)$ with edges $F \subseteq E$ and
  $\ell_2$-weights $\uu \in \rea^F$
  such that $H$ has $O(n \epsilon^{-2} \log(n/\delta))$ edges and
  with probability at least $1-\delta$ we have that
  $(H,\uu)$ is a spectral $\epsilon$-approximation of $(G,\ww
  )$.
  We denote the routine computing $H$ and $\uu$ by
  $\textsc{SpectralSparsify}$, so that
  $(H, \uu) = \textsc{SpectralSparsify}(G, \ss, \epsilon, \delta)$.
  This algorithm runs in $\Otil(m)$ time.
  Furthermore, if the weights $\ww$ are quasipolynomially bounded,
  then so are the weights of $\uu$.
\end{theorem}
 We can now prove our main result.

\begin{proof}[Proof of Theorem~\ref{thm:spannersparsify}]
  We consider a graph $G = (V,E)$ with $m$ edges
  and $n$ vertices, and with non-negative  $\ell_p$-weights $\rr
  \in \rea^E$, non-negative  $\ell_2$-weights $\ss \in \rea^E$.
  We define $\hat{E} \subseteq E$ to be the edges s.t. $\ss(e) > 0$, and
  then let $\ll \in \rea^{\hat{E}}$ by $\ll(e) = 1/\ss(e)$, and
  $\hat{G} = (V,\hat{E})$.
  We then apply Theorem~\ref{thm:spanner} to $\hat{G}$ with $\ll$
  as edge lengths, and with $k = \log(n)$.
  We turn the algorithm of Theorem~\ref{thm:spanner} into running time $\Otil(m \log(1/\delta))$,
  instead of expected time $\Otil(m)$, by applying the standard
  Las Vegas to Monte-Carlo reduction.
  With probability $1-\delta/2$, 
  this gives us a $\log n$-spanner $H_{1}$ of
  $\hat{G}
  $, and we define $\tt$ by restricting $\ss$ to the edges of $H_{1}$.
  By Lemma~\ref{lem:spannerinfnorm}, we then have
  \[
    \norm{\TT \BB_{H_{1}}\xx}_{\infty}
    \leq 
    \norm{\SS \BB_{G}\xx}_{\infty}
    \leq \log(n) \norm{\TT \BB_{H_{1}}\xx}_{\infty}
  \]
  Because $\TT \BB_{H_{1}}\xx$ is a restriction of $\SS \BB_{G}\xx$ to a
  subset of the coordinates, we always have for any $p \geq 1$ that
  $\norm{\TT \BB_{H_{1}}\xx}_{p} \leq \norm{\SS \BB_{G}\xx}_{p}$.
  
  At the same time, we also have
  \[
    \norm{\SS \BB_{G}\xx}_{p}
    \leq
    m^{1/p} \norm{\SS \BB_{G}\xx}_{\infty}
    \leq
    m^{1/p} \log(n) \norm{\TT \BB_{H_{1}}\xx}_{\infty}
    \leq
    m^{1/p} \log(n)  \norm{\TT \BB_{H_{1}}\xx}_{p}
  \]

  We define $\tilde{E} \subseteq E$ to be the edges s.t. $\rr(e) > 0$, and
  the let $\tilde{G} = (V,\tilde{E})$.
  Now, appealing to Theorem~\ref{thm:spectralsparsify}, 
  we let $(H_2, \uu) = \textsc{SpectralSparsify}(\tilde{G} , \rr, 1/2,
  \epsilon/2)$.

  Finally, we form $H$ by taking the union of the edge sets of $H_1$
  and $H_2$ and extending $\uu$ and $\tt$ to the new edge set
  by adding zero entries as needed.
  By a union bound, the approximation guarantees of
  Equations~\eqref{eq:Hspectralsparsify} and \eqref{eq:Hpsparsify}
  simultaneously hold with probability at least $1-\delta$.

  The edge set remains bounded in size by $O( n \log n)$.
\end{proof}
 To see Theorem \ref{lem:vol-sparse-inf}, note that from Theorem \ref{thm:spannersparsify}, we get,
\begin{align*}
 m^{-\frac{1}{p-1}} \left(m^{-\frac{1}{p-1}}\norm{\WW \BB_{G}\xx}_2^2 + m^{-1}  \norm{\SS \BB_{G}\xx}_p ^p \right) \leq  m^{-\frac{1}{p-1}} \left( \norm{\UU \BB_{H}\xx}_2^2 + \norm{\TT \BB_{H}\xx}_p^p \right)
\end{align*}
The other direction is easy to see.
  

\section{Extensions of Our Results and Open Problems}

\subsubsection*{Solving dual problems: $q$-norm minimizing flows and
  voltages for $ q < 2$.}
When the mixed $(\ell_2^{2} + \ell_{p}^{p})$-objective flow problem
(Problem~\eqref{eq:flowprob}) is restricted to the case $\gg = \vzero$
and $\RR = \vzero$, it becomes a pure $\ell_p$-norm minimizing flow
problem, and its dual problem can be slightly rearranged to give
\begin{equation}
  \min_{
    \vv
  }
  \dd^{\top}\vv
  +
  \norm{\SS^{-1} \BB
    \vv}_q^q
  \label{eq:dualqvoltage}
\end{equation}
where $q = p/(p-1) = 1 + 1/(p-1)$. We refer to the diagonal entries of
$\SS^{-1}$ as $\ell_q$-conductances.  Because we can solve
Problem~\eqref{eq:flowprob} to high-accuracy in near-linear time for
$p = \omega(1)$, this allows us to solve
Problem~\eqref{eq:dualqvoltage}, the dual voltage $\ell_q$-norm
minimization, in time
$p(m^{1+o(1)} + n^{4/3 + o(1)}) \log^{2} \nfrac{1}{\eps}$ (see
~\cite[Section 7]{AdilKPS19} for the reduction).  We summarize this in
the theorem below.
\begin{theorem}[Voltage Algorithmic Result, $q < 2$ (Informal)]
  Consider a graph $G$ with $n$ vertices and $m$ edges, equipped with positive $\ell_q$-conductances, as
  well as a demand vector.
  For $1 <q < 2$, when ${q=1+o(1)}$, in
  $\poly\left( \frac{1}{q-1} \right)(m^{1+o(1)} + n^{4/3 + o(1)})
  \log^{2} \nfrac{1}{\eps}$ time, we can compute an
  $\epsilon$-approximately optimal voltage solution to
  Problem~\eqref{eq:dualqvoltage} with high probability.  \todo{I'm not
    confident a simple conversion wouldn't lose more p dependence}
\end{theorem}

Similarly, we can solve $\ell_q$-norm minimizing flows for $q < 2$ as
dual to the $\ell_p$-voltage problem, a special case of the mixed
$(\ell^2_2+\ell^p_p)$-voltage problem.
%
Picking $\WW = \vzero$ in Problem~\eqref{eq:voltageprob}, we obtain
a pure $\ell_p$-norm minimizing voltage problem, and its dual problem
can be slightly rearranged to give
\begin{equation}
  \min_{
    \BB^{\top} \ff = \dd
  }
  \norm{\UU^{-1} \ff}_q^q
  \label{eq:dualqflow}
\end{equation}
where $q = p/(p-1) = 1 + 1/(p-1)$. We refer to the diagonal entries of
$\UU^{-1}$ as $q$-weights.  Again, because we can solve
Problem~\eqref{eq:voltageprob} to high-accuracy in near-linear time
for $p = \omega(1)$, this allows us to solve
Problem~\eqref{eq:dualqflow}, the dual flow $\ell_q$-norm
minimization, in time
$p(m^{1+o(1)} + n^{4/3 + o(1)}) \log^{2} \nfrac{1}{\eps}$.

\begin{theorem}[Flow Algorithmic Result, $q < 2$ (Informal)]
  Consider a graph $G$ with $n$ vertices and $m$ edges, equipped with positive $q$-weights, as
  well as a demand vector.
  For $1 <q < 2$, when ${q=1+o(1)}$, in
  $\poly\left( \frac{1}{q-1} \right) (m^{1+o(1)} + n^{4/3 + o(1)})
  \log^{2} \nfrac{1}{\eps}$ time, we can compute an
  $\epsilon$-approximately optimal flow solution to
  Problem~\eqref{eq:dualqflow} with high probability.
\end{theorem}

\subsubsection*{Open Questions}

\noindent{\textbf{Mixed $\ell_2,\ell_q$ problems for small $q < 2$.}
In this work, we provided new state-of-the-art algorithms for weighted
mixed $\ell_2,\ell_p$-norm minimizing flow and voltage problems for $p
> >2$, and
for pure $\ell_q$-norm minimizing flow and voltage problems for $q$ near
$1$.

A reasonable  definition of mixed $\ell_2,\ell_q$-norm problems for $q < 2$ is
based on gamma-functions as introduced in \cite{BubeckCLL18} and used
in \cite{AdilKPS19}.
We believe that with minor adjustments to our multiplicative weights
solver, these objectives could be handled too, by solving their dual
$\ell_2,\ell_p$-gamma function problem for $p > 2$.

\medskip\noindent{\textbf{Directly sparsifying mixed $\ell_2,\ell_q$ problems for $q < 2$.}
A second approach to developing a fast $\ell_2,\ell_q$-gamma function solver for $q
< 2$ would be to directly develop sparsification in this setting.
We believe this might be possible, and in the general matrix setting
might provide better algorithms than alternative approaches.

\medskip\noindent{\textbf{Removing the $m^{\frac{O(1)}{p-1}}$ loss in
    sparsification.} Our current approaches to graph mixed
  $\ell_2,\ell_p$-sparsification lose a factor $m^{\frac{O(1)}{p-1}}$ in their
  quality of approximation, which leads to a $m^{\frac{O(1)}{p-1}}$
  factor slowdown in running time, and makes our algorithms less
  useful for small $p$.  We believe a more sophisticated graph
  sparsification routine could remove this loss and result in
  significantly faster algorithms for $p$ close to 2.

\medskip\noindent{\textbf{Using mixed $\ell_2,\ell_p$-objectives as oracles for $\ell_{\infty}$  regression.}
The current state-of-the-art algorithm for computing maximum flow in unit
capacity graphs runs in $\Otil(m^{4/3})$ time \cite{LiuS20b}, and uses
the almost-linear-time algorithm from~\cite{KyngPSW19} for solving
\emph{unweighted} $\ell_2^{2} + \ell_{p}^{p}$ instances as a key
ingredient.

\printbibliography

\newpage
\appendix
\section{Construction of Sparsifiers for $\ell_2^{2} + \ell_{p}^{p}$ Flows}\label{sec:Flow-Sparse}
\todo{from KPSW, need to adapt!}
In this section we will prove the following formal version of Theorem \ref{thm:flow-spars-short}.
\begin{theorem}
\label{thm:Flow-Sparsification}
  Consider an instance $\calG = (V^{\calG}, E^{\calG}
  ,\rr^{\calG},\ss^{\calG},\gg^{\calG})$ with $n$ vertices and $m$
  edges. Suppose we want to solve, $\min_{\BB_G^{\top}\ff = 0} \obj^{\calG}(\ff)$. We can compute in time $\Otil(m)$ an instance
  $\calH = (V^{\calH}, E^{\calH}
  ,\rr^{\calH},\ss^{\calH},\gg^{\calH})$ such that with  probability $1-\epsilon$, $\calH$  has $n$ vertices and $m_{\calH} = n \polylog(n/(\epsilon\delta))$ edges, and for all $\ff^{\calH}$ we can compute a corresponding $\ff^{\calG}$ in $\Otil(m)$ time such that,
 \[
  \calH \circapprox_{\kappa,\delta} \calG \quad \text{and} \quad \calG \circapprox_{\kappa,\delta} \calH
  \]
 where $\kappa = m^{1/(p-1)} \polylog(n/(\epsilon\delta))$.
\end{theorem}

\subsection{Preliminaries}
\subsubsection*{Smoothed \texorpdfstring{$\ell_{p}$} - norm functions}
\label{subsec:SmoothNorm}
We consider $p$-norms smoothed by the addition of a
quadratic term. First we define such a smoothed $p\textsuperscript{th}$-power on $\rea.$
\begin{definition}[Smoothed $p\textsuperscript{th}$-power]
  Given $r, x \in \rea, r \ge 0$ define the $r$-smoothed $s$-weighted
  $p\textsuperscript{th}$-power of $x$ to be
  \[h_{p}(r, s, x) = r x^2 + s\abs{x}^p.
  \]
\end{definition}
\noindent This definition can be naturally extended to vectors to
obtained smoothed $\ell_{p}$-norms.
\begin{definition}[Smoothed $\ell_{p}$-norm]
  Given vectors $\xx \in \rea^m, \rr , \ss\in \rea^{m}_{\ge 0},$,
  define the $\rr$-smooth $\ss$-weighted $p$-norm of $\xx$ to be
  \[h_{p}(\rr, \ss, \xx) = \sum_{i=1}^m h_{p}(\rr_i, \ss_i, \xx_i) =
    \sum_{i=1}^{m} ( \rr_i \xx_i^2 + \ss_i\abs{\xx_i}^{p}).
  \]
\end{definition}

\label{sec:prelims}

\subsubsection*{Flow Problems and Approximation}
\label{sec:flow-overview}
We will consider problems where we seek to find flows minimizing
smoothed $p$-norms. We first define these problem instances.
\begin{definition}[Smoothed $p$-norm instance]
  A \emph{smoothed $p$-norm instance} is a tuple
  $\calG,$
  \[\calG\defeq (V^{\calG}, E^{\calG}, \gg^{\calG}, \rr^{\calG}, \ss^{\calG}),\]
  where $V^{\calG}$ is a set of vertices, $E^{\calG}$ is a set of
  undirected edges on $V^{\calG},$ the edges are accompanied
  by a gradient, specified by $\gg^{\calG} \in \rea^{E^{\calG}},$ 
  the edges have $\ell_2^{2}$-resistances
  given by
  $\rr^{\calG}
  \in \rea^{E^{\calG}}_{\ge 0},$ and $s \in \rea^{E^{\calG}}_{\ge 0}$ gives the
  $p$-norm scaling.
\end{definition}
\begin{definition}[Flows, residues, and circulations]
  Given a smoothed $p$-norm instance $\calG,$ a vector
  $\ff \in \rea^{E^{\calG}}$ is said to be a flow on $\calG$. A flow
  vector $\ff$ satisfies residues $\bb \in \rea^{V^{\calG}}$ if
  $\left( \BB^{\calG} \right)^{\top} \ff = \bb,$ where
  $\BB^{\calG} \in \rea^{E^{\calG} \times V^{\calG} }$ is the
  edge-vertex incidence matrix of the graph $(V^{\calG},E^{\calG}),$
  \emph{i.e.},
  $\left( \BB^{\calG} \right)^{\top}_{(u,v)} = \vone_{u} - \vone_{v}.$

  A flow $\ff$ with residue $\vzero$ is called a circulation on $\calG$.
\end{definition}
Note that our underlying instance and the edges are
undirected. However, for every undirected edge $e=(u,v) \in E$, we
assign an arbitrary fixed direction to the edge, say $u \to v,$ and
interpret $\ff_e \ge 0$ as flow in the direction of the edge from $u$
to $v,$ and $\ff_e < 0$ as flow in the reverse direction. For
convenience, we assume that for any edge $(u,v) \in E,$ we have
$\ff_{(u,v)} = - \ff_{(v,u)}.$
\begin{definition}[Objective, $\obj^{\calG}$]
  Given a smoothed $p$-norm instance $\calG,$ and a flow $\ff$ on
  $\calG,$ the associated objective function, or the energy, of $\ff$
  is given by
  \[\obj^{\calG}(\ff) = \left( \gg^{\calG} \right)^{\top} \ff -
  h_p(\rr, \ss, \ff).\]
\end{definition}

\begin{definition}[Smoothed $p$-norm flow / circulation problem]
  Given a smoothed $p$-norm instance $\calG$ and a residue vector
  $\bb \in \rea^{E^{\calG}},$ the \emph{smoothed $p$-norm flow
    problem} $(\calG, \bb)$, finds a flow $\ff \in \rea^{E^{\calG}}$
  with residues $\bb$ that maximizes $\obj^{\calG} (\ff),$
  \emph{i.e.},
  \[
    \max_{\ff: (\BB^{\calG})^{\top} \ff = \bb} \obj^{\calG} \left( \ff \right).
  \]
  If $\bb = \vzero,$ we call it a \emph{smoothed $p$-norm circulation
    problem}.
\end{definition}
Note that the optimal objective of a smoothed $p$-norm circulation
problem is always non-negative, whereas for a smoothed $p$-norm flow
problem, it could be negative.
\subsubsection*{Approximating Smoothed $p$-norm Instances}
Since we work with objective functions that are non-standard (and not
even homogeneous), we need to carefully define a new notion of
approximation for these instances.

\begin{definition}[
  $\calH \preceq_{\kappa,\delta} \calG$]
  For two smoothed $p$-norm instances, $\calG, \calH,$ we write
  $\calH \preceq_{\kappa,\delta} \calG$ if there is a linear map
  $\map{\calH}{\calG}: \rea^{E^{\calH}} \rightarrow \rea^{E^{\calG}}$
  such that for every flow $\ff^{\calH}$ on $\calH,$ we have that
  $\ff^{\calG} = \map{\calH}{\calG} (\ff^{\calH})$ is a flow on
  $\calG$ such that
  \begin{enumerate}
  \item $\ff^{\calG}$ has the same residues as $\ff^{\calH}$ \emph{i.e.},
    $(\BB^{\calG})^{\top} \ff^{\calG} = (\BB^{\calH})^{\top}
    \ff^{\calH},$ and
  \item has energy bounded by:
    \[
      \frac{1}{\kappa} \left( \obj^{\calH} \left( \ff^{\calH} \right) -
             \delta  \norm{\ff^{\calH}}_1 \right)
      \leq
      \obj^{\calG} \left( \frac{1}{\kappa} \ff^{\calG} \right).
    \]
  \end{enumerate}
\end{definition}
For some of our transformations on graphs, we will be able to prove
approximation guarantees only for circulations. Thus, we define the
following notion restricted to circulations.

\begin{definition}[$\calH \circapprox_{\kappa,\delta} \calG$]
  For two smoothed $p$-norm instances, $\calG, \calH,$ we write
  $\calH \circapprox_{\kappa} \calG$ if there is a linear map
  $\map{\calH}{\calG} : \rea^{E^{\calH}} \rightarrow \rea^{E^{\calG}}$
  such that for any circulation $\ff^{\calH}$ on $\calH$, \emph{i.e.},
  $(\BB^{\calH})^{\top}\ff^{\calH} = \vzero,$ the flow
  $\ff^{\calG} = \map{\calH}{\calG}(\ff^{\calH})$ is a circulation,
  \emph{i.e.},
  $(\BB^{\calG})^{\top} \ff^{\calG} = \vzero, $
  and satisfies
    \[
      \frac{1}{\kappa} \left( \obj^{\calH} \left( \ff^{\calH} \right) -
             \delta  \norm{\ff^{\calH}}_1 \right)
      \leq
      \obj^{\calG} \left( \frac{1}{\kappa} \ff^{\calG} \right).
    \]
  Observe that $\calH \preceq_{\kappa,\delta} \calG$ implies
  $\calH \circapprox_{\kappa,\delta} \calG.$
\end{definition}

We define the usual induced matrix $1$-to-$1$ norm as
\[
  \norm{\calM}_{1 \to 1} = \max_{ \ff \in \rea^E }
  \frac{ \norm{\calM
    \ff}_{1} }{\norm{\ff}_{1} }
  \]

We define a special matrix $1$-to-$1$ norm over circulations by 
\[
  \norm{\calM}^{\text{cycle}}_{1 \to 1} = \max_{ \BB \ff = \vzero }
  \frac{ \norm{\calM
    \ff}_{1} }{\norm{\ff}_{1} }
  \]
  
\begin{definition}[
  $\calH \preceq_{\kappa} \calG$ and $\calH \circapprox_{\kappa} \calG$]
  We abbreviate $\calH \preceq_{\kappa, 0} \calG$ as   $\calH
  \preceq_{\kappa} \calG$,
  and $\calH \circapprox_{\kappa,0
  } \calG$ as $\calH \circapprox_{\kappa} \calG$
\end{definition}

\begin{definition}
  In the context of a problem with $m$ vertices and $n$ edges,
  we say a real number $x$ is quasi-polynomially
  bounded 
  \[
    2^{-\polylog(n)} \leq \abs{x} \leq 2^{\polylog(n)}
    \]
\end{definition}

\begin{definition}
  Consider a smoothed $p$-norm flow
  problem $(\calG, \bb)$ where 
  $\calG = (V, E,\rr,\ss,\gg)$
  with $n$ vertices and $m$ edges.
  We say that the instance is
  quasipolynomially-bounded
  if the entries of $\rr$ and $\ss$ are quasipolynomially bounded and 
  \[
    \max_{\ff: (\BB^{\calG})^{\top} \ff = \bb} \obj^{\calG} \left( \ff
    \right)
    \leq
    2^{\polylog(n)}.
  \]
\end{definition}

\begin{definition}[Touched and untouched cycles]
  We say a cycle of edges in $E$ is \emph{touched} if it contains an
  edge $e$ s.t. $\rr(e) \neq 0$ or
  $\ss(e) \neq 0$.
  Otherwise, we say the cycle is \emph{untouched}
\end{definition}

\begin{definition}[Cycle-touching instance]
  We say an instance $\calG = (V, E,\rr,\ss,\gg)$ is
  cycle-touching if every cycle of edges in $e$ is touched.
\end{definition}

\subsection{Additional Properties of Flow Problems and Approximation}
The definitions in Section \ref{sec:flow-overview} satisfy most properties that we want from
comparisons. The following lemma, slightly extends
a similar statement in \cite{KyngPSW19}.
\begin{restatable}[Reflexivity]{lemma}{identity}
  \label{lem:approximations:identity}
  For every smoothed $p$-norm instance $\calG,$ and every
  $\kappa \ge 1$, $\delta \geq 0$, we have $\calG \preceq_{\kappa,\delta} \calG$
  with the identity map.
\end{restatable}
\begin{proof}
  Consider the map $\map{\calG}{\calG}$ such that for every flow
  $\ff^{\calG}$ on $\calG,$ we have $\map{\calG}{\calG}(\ff^{\calG}) =
  \ff^{\calG}.$ Thus,
  \begin{align*}
    \obj^{\calG}\left( {\kappa}^{-1}
    \map{\calG}{\calG}(\ff^{\calG}) \right)
    & = \obj^{\calG} \left( {\kappa}^{-1} \ff^{\calG} \right) \\
    & = \left( \gg^{\calG}  \right)^{\top} \left( {\kappa}^{-1}
      \ff^{\calG} \right) - h_p(\rr, \kappa^{-1} \ff^{\calG}) \\
    & \ge \kappa^{-1}  \left( \gg^{\calG}  \right)^{\top} \ff^{\calG}
      - \kappa^{-2} h_p( \rr, \ff^{\calG})
    & \\
    & \ge \kappa^{-1}  \left( \gg^{\calG}  \right)^{\top} \ff^{\calG}
      - \kappa^{-1} h_p( \rr, \ff^{\calG})
          & \\
    & =\kappa^{-1} \obj^{\calG}(\ff^{\calG})
          & \\
    & \ge
      \kappa^{-1} \obj^{\calG}(\ff^{\calG}) - \delta \norm{\ff^{\calG}}_{1}
  \end{align*}
  Moreover $(\BB^{\calG})^{\top} \map{\calG}{\calG}(\ff^{\calG}) = \BB^{\calG}
  \ff^{\calG}.$ Thus, the claims follow.
\end{proof}


It behaves well under composition.
\begin{restatable}[Composition]{lemma}{composition}
  \label{lem:approximations:composition}
  
  Given two smoothed $p$-norm instances, $\calG_1, \calG_2,$ such that
  $\calG_1 \preceq_{\kappa_1,\delta_1} \calG_2$ with the map
  $\map{\calG_1}{\calG_2}$ and $\calG_2 \preceq_{\kappa_2,\delta_2} \calG_3$
  with the map $\map{\calG_2}{\calG_3}$, then
  $\calG_1 \preceq_{\kappa, \delta} \calG_3$ with the map
  $ \map{\calG_1}{\calG_3} = \map{\calG_2}{\calG_3} \circ
  \map{\calG_1}{\calG_2}$
  and
  $\kappa = \kappa_1 \kappa_2$
  and
  $\delta = \delta_1+\delta_2
    \norm{\map{\calG_1}{\calG_2}}_{1 \to 1}$.

   Similarly, given two smoothed $p$-norm instances, $\calG_1, \calG_2,$ such that
  $\calG_1 \circapprox_{\kappa_1,\delta_1} \calG_2$ with the map
  $\map{\calG_1}{\calG_2}$ and $\calG_2 \circapprox_{\kappa_2,\delta_2} \calG_3$
  with the map $\map{\calG_2}{\calG_3}$, then
  $\calG_1 \circapprox_{\kappa, \delta} \calG_3$ with the map
  $ \map{\calG_1}{\calG_3} = \map{\calG_2}{\calG_3} \circ
  \map{\calG_1}{\calG_2}$
  and
  $\kappa = \kappa_1 \kappa_2$
  and
  $\delta = \delta_1+\delta_2
  \norm{\map{\calG_1}{\calG_2}}^{\text{cycle}}_{1 \to 1}$.
\end{restatable}
\begin{proof}
  We can simply chain together the guarantees to see that:
  \begin{align*}
    \obj^{\calG_1} \left( \ff^{\calG_1} \right)
    &\leq
      \kappa_1 \obj^{\calG_2} \left( \frac{1}{\kappa_1} \ff^{\calG_2}
      \right)
      +
      \delta_1  \norm{\ff^{\calG_1}}_1
    \\
    & \leq
      \kappa_1  \kappa_2
      \left(
      \obj^{\calG_3} \left( \frac{1}{\kappa_1 \kappa_2} \ff^{\calG_3}\right)
      +\delta_2 \norm{\frac{1}{\kappa_1} \ff^{\calG_2}}_1
      \right)
      +
      \delta_1  \norm{\ff^{\calG_1}}_1
     \\
      &=
      \kappa_1  \kappa_2 \obj^{\calG_3} \left( \frac{1}{\kappa_1
      \kappa_2} \ff^{\calG_3}\right)
      +\delta_2 \norm{\ff^{\calG_2}}_1
      +\delta_1  \norm{\ff^{\calG_1}}_1
      \\
      &\leq
        \kappa_1  \kappa_2 \obj^{\calG_3} \left( \frac{1}{\kappa_1
      \kappa_2} \ff^{\calG_3}\right)
      +\left(
      \delta_2 \norm{\map{\calG_1}{\calG_2}}_{1 \to 1}
      +\delta_1
      \right)\norm{\ff^{\calG_1}}_1
  \end{align*}
  A similar calculation gives the cycle composition guarantee, but now
  allows us to bound norms using the
  $\norm{\map{\calG_1}{\calG_2}}^{\text{cycle}}_{1 \to 1}$ norm.
  In all cases, our maps preserve that flows route the correct demands.
\end{proof}

The most important property of this is that this notion of
approximation is also additive, \emph{i.e.}, it works well with graph
decompositions.
\begin{definition}[Union of two instances]
  \label{def:union}
  Consider smoothed $p$-norm instances, $\calG_1, \calG_2,$ with the
  same set of vertices, \emph{i.e.}  $V^{\calG_1} = V^{\calG_2}.$
  Define $\calG = \calG_1 \cup \calG_2$ as the instance on the same
  set of vertices obtained by taking a disjoint union of the edges
  (potentially resulting in multi-edges). Formally,
  \[ \calG = (V^{\calG_1}, E^{\calG_1} \cup E^{\calG_2}, (\gg^{\calG_1},
    \gg^{\calG_2}), (\rr^{\calG_1}, \rr^{\calG_2}), (\ss^{\calG_1},
    \ss^{\calG_2})).\]
\end{definition}
We prove the following lemma, which closely follows an analogous
statement in \cite{KyngPSW19}.
\begin{restatable}[Union of instances]{lemma}{unionlem}
  \label{lem:union}
  Consider four smoothed $p$-norm instances,
  $\calG_1, \calG_2, \calH_1, \calH_2,$ on the same set of vertices,
  \emph{i.e.} $V^{\calG_1} = V^{\calG_2} = V^{\calH_1} = V^{\calH_2},$
  such that for $i=1,2,$ $\calH_i \preceq_{\kappa,\delta} \calG_i$ with the
  map $\map{\calH_i}{\calG_i}.$ Let
  $\calG \defeq \calG_1 \cup \calG_2,$ and
  $\calH \defeq \calH_1 \cup \calH_2.$
  Then, $\calH \preceq_{\kappa,\delta} \calG$ with the map
  \[
    \map{\calH}{\calG} \left( \ff^{\calH} =
      (\ff^{\calH_1},\ff^{\calH_2}) \right) \defeq \left(
      \map{\calH_1}{\calG_1} \left( \ff^{\calH_1} \right),
      \map{\calH_2}{\calG_2} \left( \ff^{\calH_2} \right) \right),
  \]
  where $(\ff^{H_1}, \ff^{H_2})$ is the decomposition of
  $\ff^{H}$ onto the supports of $H_1$ and $H_2$.
\end{restatable}
\begin{proof}
Let $\ff^{\calH}$ be a flow on $\calH.$ We write $\ff^{\calH} =
  (\ff^{\calH_1}, \ff^{\calH_2}).$ Let $\ff^{\calG} \defeq
  \map{\calH}{\calG}(\ff^{\calH}).$ If $\ff^{\calG_i}$ denotes
  $\map{\calH_i}{\calG_i}(\ff^{\calH_i})$ for $i=1,2,$ then we know
  that $\ff^{\calG} = (\ff^{\calG_1}, \ff^{\calG_2}).$ Thus, the
  objectives satisfy
  \begin{align*}
    \obj^{\calG}(\kappa^{-1} \ff^{\calG})
    & = \obj^{\calG_1}(\kappa^{-1} \ff^{\calG_1}) +
      \obj^{\calG_2}(\kappa^{-1} \ff^{\calG_2}) \\
    & \ge 
      \kappa^{-1}\left( \obj^{\calH_1}(\ff^{\calH_1})
      -\delta \norm{\ff^{\calH_1}}_1
      \right)
      +
      \kappa^{-1}\left( \obj^{\calH_2}(\ff^{\calH_2})
            -\delta \norm{\ff^{\calH_2}}_1
      \right)
      = \kappa^{-1} \left(  \obj^{\calH}(\ff^{\calH})
      - \delta \norm{\ff^{\calH}}_1
      \right)
  \end{align*}
  For the residues, we have,
  \begin{align*}
    (\BB^{\calG})^{\top}(\ff^{\calG})
    & = (\BB^{\calG_1})^{\top}(\ff^{\calG_1}) +
      (\BB^{\calG_2})^{\top}(\ff^{\calG_2}) \\
    & = (\BB^{\calH_1})^{\top}(\ff^{\calH_1}) +
      (\BB^{\calH_2})^{\top}(\ff^{\calH_2}) = (\BB^{\calH})^{\top}(\ff^{\calH}).
  \end{align*}
  Thus, $\calH \preceq_{\kappa,\delta} \calG.$
\end{proof}

This notion of approximation also behaves nicely with scaling of
$\ell_2$ and $\ell_p$ resistances.
\begin{restatable}{lemma}{PerturbResistances}
  \label{lem:PerturbResistances}
  For all $\kappa \ge 1,$ and for all pairs of smoothed $p$-norm
  instances, $\calG, \calH,$ on the same underlying graphs,
  \emph{i.e.}, $(V^{\calG}, E^{\calG}) = (V^{\calH}, E^{\calH}),$ such
  that,
  \begin{tight_enumerate}
  \item the gradients are identical, $\gg^{\calG} = \gg^{\calH},$
  \item the $\ell_{2}^2$ resistances are off by at most $\kappa,$ \emph{i.e.},
    $\rr_e^{\calG} \leq \kappa \rr_{e}^{\calH}$ for all edges $e,$ and
  \item the $p$-norm scaling  is off by at most
    $\kappa^{p - 1},$ \emph{i.e.},
    $s^{\calG} \leq \kappa^{p - 1} s^{\calH},$
  \end{tight_enumerate}
  then $\calH \preceq_{\kappa} \calG$ with the identity map.
\end{restatable}
\begin{proof}
Follows from Lemma 2.13 in \cite{KyngPSW19}.
\end{proof}


\subsection{Orthogonal Decompositions of Flows}
\label{sec:CyclePotential}

At the core of our graph decomposition and sparsification
procedures is a decomposition of the gradient $\gg$
of $\calG$ into its cycle space and potential flow space.
We denote such a splitting using
\begin{equation}
\gg^{\calG}
= 
\gghat^{\calG} + \BB^{\calG} \ppsi^{\calG},
\text{ s.t. }~{\BB^{\calG}}^{\top} \gghat^{\calG}=\vzero.
\end{equation}
Here $\gghat$ is a circulation, while
$\BB \ppsi$ gives a potential induced edge value.
We will omit the superscripts when the context is clear.

The following minimization based formulation of this
splitting of $\gg$ is critical to our method of bounding
the overall progress of our algorithm
\begin{fact}
\label{fact:Projection}
The projection of $\gg$ onto the cycle space is obtained
by minimizing the Euclidean norm of $\gg$ plus a potential flow.
Specifically,
\[
\norm{\gghat}_2^2
=
\min_{\xx} \norm{\gg + \BB \xx }_2^2.
\]
\end{fact}

\begin{lemma}
\label{lem:EnergyDecrease}
Given a graph/gradient instance $\calG$, consider
$\calH$ formed from a subset of its edges.
The projections of $\gg^{\calG}$ and $\gg^{\calH}$ onto
their respective cycle spaces, $\gghat^{\calG}$
and $\gghat^{\calH}$ satsify:
\[
\norm{\gghat^{\calH}}_2^2
\leq
\norm{\gghat^{\calG}}_2^2
\leq
\norm{\gg^{\calG}}_2^2.
\]
\end{lemma}

\subsection{Numerics, Conditioning, Inexact Laplacian Solvers}

Because we allow instances $\calG = (V, E,\rr,\ss,\gg)$ with $\rr(e) =
0$ and $\ss(e) = 0$ for some edges $e \in E$ (and this is important
for our sparsificaton procedures), we need to be somewhat careful
about disallowing instances with a cycle that has zero $\rr$ and $\ss$
values and non-zero gradient: In that case, our 
``energy'' can diverge to $+\infty$.

%
%
%
%

\begin{definition}[Unbounded and constant cycles]
  An untouched cycle $C \subseteq E$ is \emph{unbounded} if the sum of terms around
  the cycle is non-zero, i.e. $\sum_{e \in C} \gg(e) \neq 0$.
  If the sum \emph{is} zero, we call the unbounded cycle a \emph{constant cycle}.
\end{definition}

\begin{lemma}
  Consider a smoothed $p$-norm flow problem $(\calG, \bb)$, i.e.
     \begin{equation*}
      \max_{\ff: (\BB^{\calG})^{\top} \ff = \bb} \obj^{\calG} \left(
        \ff \right).
    \end{equation*}
 The problem is unbounded, i.e. has objective value $\obj^{\calG} \to
 \infty$ if and only if $\calG$ contains an unbounded cycle.
\end{lemma}
\begin{proof}
  The problem is unbounded exactly when the gradient has non-zero
  inner product with some element of the subspace $\ker(\BB^{\top}) \cap \ker(\RR) \cap \ker(\SS)$.
    The first condition means it $\ker(\BB^{\top})$ tells us the
      element must be in the cycle space, while the latter two tells
      us it must be supported edges with non-zero $\RR$ and $\SS$.
      Writing the element as a linear combination of particular
      cycles, the gradient must have a non-zero inner product with
      one untouched cycle.
\end{proof}

\begin{lemma}
  \label{lem:detectUnbounded}
  The algorithm $\textsc{detectUnbounded}$ takes as input a
  smoothed $p$-norm problem $(\calG, \bb)$ and if $\calG$ contains an unbounded cycle, the algorithm detects this and
  output an unbounded cycle (an arbitrary one if there are
  multiple).
\end{lemma}

\begin{proof}
We run a DFS only using edges outside the union of the support of $\rr$ and $\ss$ (just delete the  others before), and assign voltages. On an edge to already visited
  vertex, check if new voltage agrees, if output cycle on stack: it is
  an unbounded cycle.
  At the end, consistent voltages certify that gradient is an
  electrical flow and hence objective is bounded.
  Overall time is linear, as each edge is processed only once.
\end{proof}


\begin{lemma}
  \label{lem:cyclePreprocess}
  The algorithm $\textsc{constantCycleContraction}$ takes as input a
  smoothed $p$-norm problem $(\calG, \bb)$ which has bounded objective
  value and
  \begin{enumerate}
  \item Returns an equivalent instance $(\calG',\bb')$ where is $\calG'$ is cycle-touching.
    $\calG'$ is equals $\calG$ with all constant cycle contracted
    into a vertex.
  \item The problems $(\calG, \bb)$ and $(\calG', \bb')$ are
    equivalent in that
    \begin{enumerate}
    \item Any feasible solution $\ff$ to the first
      problem can be trivially mapped to a feasible solution to the
      second and vice versa.
    \item Mapping a solution from $(\calG,\bb)$ to $(\calG',\bb')$ always changes the
      objective value by the \emph{same} scalar (say $c$) for all
      solutions, and simiarly mapping the other way changes the
      objective by $-c$.
    \item The maps can be applied in $O(m)$ time.
    \item If the entries of $(\calG,\bb)$ are quasi-polynomially
      bounded, then so are the entries of $(\calG',\bb')$.
    \end{enumerate}
  \end{enumerate}
  Furthermore, if $\bb = \vzero$ then $\bb' = 0$ and the solution
  maps preserve the objective value exactly, i.e.  $c=0$ and hence
  \[
    \calG \circapprox_{1} \calG' \text{ and } \calG' \circapprox_{1}
    \calG,
  \]
  and for zero-demand flows the maps $\calM_{G \to G'}$ and $\calM_{G \to G'}$
  satisfy
  \[
    \norm{\calM_{G \to G'}}^{\text{cycle}}_{1 \to 1} \leq 1 \text{ and } \norm{\calM_{G' \to
        G}}^{\text{cycle}}_{1 \to 1} \leq m.
    \]
\end{lemma}


\begin{proof}
Once we know the problem is bounded, and hence contains no unbounded
cycles, we can look at the connected components consisting of
untouched edges, and we can repeatedly contract cycles of these parts.
In the case where demands are zero initially, we produce a smaller
instance where again demands are zero. To lift to the larger instance,
we can simply put no flow on these edges. Because having a cycle
flow on these edges does not increase the objective, our mapping
guarantees hold.
  Overall time is linear, as each edge is processed only once.
\end{proof}

  

%


\begin{lemma}
  Consider a smoothed $p$-norm circulation problem
  problem $(\calG, \vzero)$,  where 
  $\calG = (V, E,\rr,\ss,\gg)$.
 Suppose the entries of $\rr$ and $\ss$ are quasipolynomially
 bounded and $\calG$ is cycle-touching, and suppose $\norm{\gg}_{\infty}
    \leq 2^{\polylog(n)} $.
  Then $(\calG, \vzero)$ is quasi-polynomially bounded.
\end{lemma}

\begin{proof}
If every cycle that every non-zero $g$ edges appears in contains
  a non-zero entry of $r$ or $s$, then increasing flow along that
  cycle will eventually lead to a decrease in objective.
\end{proof}

\begin{remark}
  We are analyzing our algorihtm in the Real RAM model, but by
  applying the tools from this section, it can also be analyzed in
  fixed precision arithmetic with polylogarithmic bit complexity per
  number:
  In this model, our $\textsc{detectUnbounded}$ and $\textsc{constantCycleContraction}$ procedures still
  work and returns a cycle-touching instance. Once an instance is
  cycle-touching and the non-zero vectors it returns are not too small or
  big, it is possible to manage errors from fixed point arithmetic.
  \todo{be less vague?}
  This can also allow us to work with inexact Laplacian solvers using
  quasipolynomial errors, which we can compute in this model in nearly-linear time.
\end{remark}




\subsection{Main Sparsification Theorem for Flows}

\begin{theorem}[Instance Sparsification]
  \label{thm:instancesparse}
  Consider an instance $\calG = (V^{\calG}, E^{\calG}
  ,\rr^{\calG},\ss^{\calG},\gg^{\calG})$ with $n$ vertices and $m$
  edges, with $\rr^{\calG}$ and $\ss^{\calG}$ quasipolynomially bounded, and
  $\norm{\gg^{\calG}}_{\infty} \leq 2^{\polylog(n)}$, and suppose the
  instance is cycle-touching.
  We can compute an instance
  $\calH = (V^{\calH} E^{\calH}
  ,\rr^{\calH},\ss^{\calH},\gg^{\calH})$ 
  with $n$ vertices and $m_{\calH} = n \polylog(n/(\epsilon\delta))$
  edges,
  again with $\rr^{\calH}$ and $\ss^{\calH}$ quasipolynomially bounded, and
  ${\norm{\gg^{\calH}}_{\infty} \leq 2^{\polylog(n)}}$,
  in time $\Otil(m)$ such that with probability $1-\epsilon$
  the maps
  $\map{\calG}{\calH}$ and $\map{\calH}{\calG}$
  certify 
  \[
    \calH \preceq_{\kappa,\delta} \calG \text{ and } \calG \preceq_{\kappa,\delta}
    \calH
    ,
  \]
  where $\kappa = m^{1/(p-1)} \polylog(n/(\epsilon\delta))$.
  Furthermore, these maps can be applied in time $\Otil(m)$.
\end{theorem}

\begin{definition}
  \label{defn:uniformexpander}
  A graph~\footnote{We use an instance and its underlying graph interchangeably in our discussion.} $G$ is a $\alpha$-uniform $\phi$-expander (or {\em uniform expander} when parameters not spelled out explicitly) if
  \begin{tight_enumerate}
  \item $\rr$ on all edges are the same.
  \item $\ss$ on all edges are the same.
  \item $G$ has {\em conductance}\footnote{$\rr$ are uniform, so conductance is defined as in unweighted graphs. We use the standard definition of conductance. For graph $G=(V,E)$, the conductance of any $\emptyset\neq S\subsetneq V$ is $\phi(S)=\frac{\delta(S)}{\min\left(vol(S),vol(V\setminus S)\right)}$ where $\delta(S)$ is the number of edges on the cut $(S,V\setminus S)$ and $vol(S)$ is the sum of the degree of nodes in $S$. The conductance of a graph is $\phi_G=\min_{S\neq \emptyset,V}\phi(S)$.} at least $\phi$.
  \item The projection of $\gg$ onto the cycle space of $G$,
    $\gghat^G = (\II-\BB \LL^{\dag} \BB^{\top})\gg$, is {\em $\alpha$-uniform} (see next definition), where $\BB$ is the edge-vertex incidence matrix of $G$, and $\LL=\BB^{\top}\BB$ is the Laplacian.
  \end{tight_enumerate}
\end{definition}
\begin{definition}
  \label{defn:alphauniform}
  A vector $\yy \in \rea^m$ is said to be {\em $\alpha$-uniform} if
  \[\norm{\yy}_{\infty}^{2} \le \frac{\alpha}{m}
    \norm{\yy}^{2}_{2}.\]
  We abuse the notation to also let the all zero vector $\vzero$ be $1$-uniform.
\end{definition}

The next 2 theorems are from \cite{KyngPSW19}
\begin{restatable}[Decomposition into Uniform Expanders]{theorem}{Decompose}
  \label{thm:Decompose}
  Given any graph/gradient/resistance instance $\calG$
  with $n$ vertices, $m$ edges, all equal to $r$, $p$-weights all equal to
  $s$, and gradient $\gg^{\calG}$,
  along with a parameter $\delta$, $\textsc{Decompose}(\calG, \delta)$
  returns disjoint vertex subsets
  $V_1, V_2, \ldots $
  in $O(m \log^{7}n \log^2(n / \delta))$ time such that
  if we let $\calG_1, \calG_2, \ldots $ be the instances obtained
  by restricting $\calG$ to the induced graphs
  on the $V_i$ sets,
  then
  at least $m/2$ edges are contained in these subgraphs,
  and each $\calG_i$ satisfies (for some absolute constant
  $c_{partition}$):
  \begin{enumerate}
  \item The graph $(V^{\calG_i}, E^{\calG_i})$ has conductance
    at least
    \[
      \phi = 
      \frac{1}{c_{partition}
        \cdot \log^{3}n
        \cdot \log\left(n / \delta\right)},
    \]
    and degrees at least $\phi \cdot \frac{m}{3n}$,
    where $c_{partition}$ is an absolute constant.
  \item The projection of its gradient $\gg^{\calG_i}$ into the cycle space
    of $\calG_i$, $\gghat^{\calG_i}$ satisfies one of:
    \begin{enumerate}
    \item $\gghat^{\calG_i}$ is $O(\log^{8}{n} \log^{3}(n / \delta))$-uniform,
      \[
        \left(\gghat_{e}^{\calG_i}\right)^2
        \leq
        \frac{O\left( \log^{14}n \log^{5}\left( n / \delta\right)\right)}
        {m_{i}} \norm{\gghat^{\calG_i}}_2^2
        \qquad
        \forall e \in E\left( \calG_i \right).
      \]
      Here $m_i$ is the number of edges in $\calG^{\calG_i}$.
      \label{case:Uniform}
    \item The $\ell_2^2$ norm of $\gghat^{\calG_i}$ is smaller by a factor
      of $\delta$ than the unprojected gradient:
      \[
        \norm{\gghat^{\calG_i}}^2_2
        \leq
        \delta \cdot \norm{\gg^{\calG}}_2^2.
      \]
      \label{case:Tiny}
    \end{enumerate}
  \end{enumerate}
\end{restatable}

\begin{restatable}[Sampling Uniform Expanders]{theorem}{ExpanderSparsify}
  \label{thm:sampAndFixGrad}
  Given an $\alpha$-uniform $\phi$-expander
  $\calG = (V^{\calG} E^{\calG} , r^{\calG},s^{\calG},\gg^{\calG})$
  with $m$ edges and vertex degrees at least $d_{\min}$,
  for any sampling probability $\tau$ satisfying
  \[
    \tau \geq c_{sample} \log(n/\epsilon) \cdot
    \left( \frac{\alpha}{m}
      + \frac{1}{\phi^2 d_{\min}} \right),
  \]
  where $c_{sample}$ is some absolute constant,
  $\textsc{SampleAndFixGradient}(\calG, \tau)$ 
  with probability at least $1-\epsilon$ \todo{confirm $\epsilon$ dep} returns 
  a partial instance $\calH = (H,r^{\calH},s^{\calH},\gg^{\calH})$ and maps
  $\map{\calG}{\calH}$ and $\map{\calH}{\calG}$.
  The graph $H$ has the same vertex
  set as $G$, and $H$ has at most $2 \tau m$ edges.
  Furthermore, $r^{\calH} = \tau \cdot r^{\calG}$ and $s^{\calH} = \tau^p \cdot s^{\calG}$.
  The maps
  $\map{\calG}{\calH}$ and $\map{\calH}{\calG}$
  certify 
  \[
    \calH \preceq_{\kappa} \calG \text{ and } \calG \preceq_{\kappa}
    \calH
    ,
  \]
  where $\kappa = m^{1/(p-1)} \phi^{-9} \log^{3} n  $ \todo{check
    $\epsilon$ dep}
  This map can be applied in time $\Otil(m)$.
\end{restatable}

\begin{definition}[$2$-Rounded instance]
  We call an instance $\calG = (V^{\calG}, E^{\calG}
  ,\rr^{\calG},\ss^{\calG},\gg^{\calG}),$ \emph{$2$-rounded},
  if every non-zero entry of $\ss$ and $\rr$ has absolute value equal to
  a power of two (which can be negative).
\end{definition}

\todolow{reformat into pseudo-code? }
Given an instance $\calG = (V^{\calG}, E^{\calG}
,\rr^{\calG},\ss^{\calG},\gg^{\calG}),$
we compute a $2$-rounded instance
$\calG' = (V^{\calG}, E^{\calG}
,\rr^{\calG'},\ss^{\calG'},\gg^{\calG}),$.
We round each $\rr^{\calG}_e$ of edges $e\in
E^{\calG} $ down to the nearest power of $2$ (can be less
than $1$).
Similarly, we round each $\ss^{\calG}_e
$ of edges $e\in
E^{\calG} $ down to the nearest power of $2$ (can be less
than $1$).
We denote this rounding procedure by
$\textsc{instanceRound}$, s.t.
$\calG' = \textsc{instanceRound}(\calG)$.
We will only need to apply this procedure to quasipolynomially bounded
numbers, which ensures it can be implemented in logarithmic time in the Real
RAM with comparisons.
\begin{remark}
  Becaues it is applied to quasipolynomially bounded entries, the rounding can be implemented using a polylogarithmic number of
  bit operations in fixed point arithmic.
\end{remark}

\begin{lemma}[$2$-rounding]
  Consider an instance
  $\calG = (V, E,\rr,\ss,\gg)$.
  Let $\calG' = \textsc{instanceRound}(\calG)$.
  Then the identity map between the instances 
  certifies
  \[
    \calG \preceq_{1} \calG' \text{ and } \calG' \preceq_{2}
    \calG
    .
    \todo{check dir and values}
  \]
  and $\calG'$ is $2$-rounded.
\end{lemma}

\todolow{summarize outline? write pseudo-code}

\begin{proof}[Proof of
  Theorem~\ref{thm:instancesparse}]
We are given an instance $\calG = (V^{\calG}, E^{\calG}
,\rr^{\calG},\ss^{\calG},\gg^{\calG})$ with $n$ vertices and $m$
edges, with $\rr^{\calG}$ and $\ss^{\calG}$ quasipolynomially bounded, and
  $\norm{\gg^{\calG}}_{\infty} \leq 2^{\polylog(n)}$, and cycle-touching.
  

  We then compute $\calG' = \textsc{instanceRound}(\calG)$, to get
  the $2$-rounded, cycle-touching instance $\calG' = (V^{\calG'}, E^{\calG'}
  ,\rr^{\calG'},\ss^{\calG'},\gg^{\calG'})$ with $\rr^{\calG'}$ and $\ss^{\calG'}$ quasipolynomially bounded, and
  ${\norm{\gg^{\calG'}}_{\infty} \leq 2^{\polylog(n)}}$
  with using the identity map between the instances
  \begin{equation}
    \calG \preceq_{1} \calG' \text{ and } \calG' \preceq_{2}
    \calG.\label{eq:roundapx}
  \end{equation}
  Because $\calG'$ is $2$-rounded and $\rr^{\calG'}$ and
  $\ss^{\calG'}$ quasipolynomially bounded, the entries of these two
  vectors only take on $\polylog(n)$ different values.
  Thus we can divide the edges $E^{\calG'}$ into $\polylog(n)$
  buckets such that in every bucket $i$, all edges have the same
  $2$-resistance value $\rr_e = r^{(i)}$ and the same $p$-weight value
  $\ss(e) = s^{(i)}$.
  
  Let $\calG^{(i)} = (V^{\calG^{(i)}}, E^{\calG^{(i)}}
  ,\rr^{\calG^{(i)}},\ss^{\calG^{(i)}},\gg^{\calG^{(i)}})$ be the
  instance arising from restricting $\calG'$ to the edges of bucket
  $i$, while letting $V^{\calG^{(i)}}$ to be the set of vertices
  incident on edges of $E^{\calG^{(i)}}$.
  There will be exactly one bucket containing all the edges $e$ with
  $\rr(e) = \ss(e) = 0$. 
  We let this bucket have index $i = 0$.
  This bucket cannot contain a cycle of edges, because if it did then
  $\calG''$ would contain an untouched cycle, contradiction that it is
  cycle-touching.
  Hence $\calG^{(0)}$ contains at most $n-1$ edges.

    We now sparsify the edges in every bucket separately, except bucket
  $i = 0$, which we do not sparsify.
  For $i > 0$, we define $\calG^{(i,0)} = \calG^{(i)}$, and let $j \gets 0$.
  As long as $\calG^{(i,j)}$ contains more than $n \log^{10} n$ edges
  we then repeat the following: 
  Appealing to Theorem~\ref{thm:Decompose},
  we now call $\textsc{Decompose}(\calG^{(i,j)}, \tilde{\delta})$ with $\tilde{\delta}
  = 2^{-\log^c n} \delta$ for some universal constant $c$ large enough
  that $\norm{\gg^{\calG}}_{\infty} \tilde{\delta} \leq \delta$.
  This produces a partition of $V^{(i,j)}$ into disjoint
  $V^{(i,j)}_1, \ldots, V^{(i,j)}_k$.
  The $\textsc{Decompose}$ algorithm defines  $\calG^{(i,j)}_l$ to be the
  instance given by restricting $\calG^{(i,j)}$ to the induced graph
  on $V^{(i,j)}_l$ and let $\calG^{(i,j+1)}$ be the instance arising
  from restricting $\calG^{(i,j)}$ to the graph consisting of edges
  crossing between the vertex partitions, and the vertices incident on
  these edges.
  We then let $j \gets j+1$ and repeat the decomposition if necessary.
    By Theorem~\ref{thm:Decompose}, $\calG^{(i,j+1)}$ contains at most
  half of the edges of $\calG^{(i,j+1)}$, we call $\textsc{Decompose}$
  at most $\log n$ times as $\log_2(m/n) \leq \log(n)$.
 For each bucket $i$, we let $j_i$ denote the last instance produced,
 i.e. $\calG^{(i,j_i))}$ is the this final instance, which is not
 included in any expander. 
  For every instance $\calG^{(i,j)}$,
  we let $m^{(i,j)} =
  \abs{E^{\calG^{(i,j)}}}$ denote the number of edges of the graph,
  and $n^{(i,j)} =
  \abs{V^{\calG^{(i,j)}}}$ the number of vertices.
  Note that the edges of $\calG^{(i)}$ are partitioned between the
  $\calG^{(i,j))}_l$ instances and the $\calG^{(i,j_i))}$ instances,
  i.e. each edges of $\calG^{(i)}$ is contained in exactly instance,
  either a 
  $\calG^{(i,j))}_l$ or a $\calG^{(i,j_i))}$.
  Thus the union of all these in the sense of
  Definition~\ref{def:union} is exactly $\calG^{(i)}$, and the union
$\calG^{(i)}$ is $\calG'$.

  We will not sparsify the final $\calG^{(i,j_i)}$, which for each bucket $i$
  have at most $n \log^{10} n$ edges.
  Again, for every instance $\calG^{(i,j)}_l$,
  we let $m^{(i,j)}_l = \abs{E^{\calG^{(i,j)}_l}}$ denote the number
  of edges of the graph.
  
By Theorem~\ref{thm:Decompose}, 
 \begin{itemize}
  \item The graph associated with each $\calG^{(i,j)}_l$ has
    conductance at least
    $\phi \geq \frac{1}{C_1\cdot \log^{3}n \cdot \log\left(n /
        \tilde{\delta}\right)}$,
    for some universal constant $C_1$.
  \item  For each $\calG^{(i,j)}_l$ either
    \begin{itemize}
    \item (``Uniform case'') the instance is $\alpha$-uniform with $\alpha \leq C_2
      \log^{8}{n} \log^{3}(n / \tilde{\delta})$ for some universal constant
      $C_2$. 
    \item (``Small
      case'') The $\ell_2^2$ norm of $\gghat^{\calG^{(i,j)}_l}$, the
      gradient projected to the cycle space,
      is smaller by a factor
      of $\tilde{\delta}$ than the unprojected gradient of the original graph 
      \[
        \norm{\gghat^{\calG^{(i,j)}_l}}^2_2
        \leq
        \tilde{\delta} \cdot \norm{\gghat^{\calG^{(i,j)}}}
        _2^2  \leq
        \tilde{\delta} \cdot \norm{\gghat^{\calG}}_2^2
        .
      \]
      \label{case:Tiny}
    \end{itemize}
  \item The minimum degree of each $\calG^{(i,j)}_l$ graph is at least
    $\phi \cdot \frac{m^{(i,j)}}{3n^{(i,j)}}$.
    
  \end{itemize}
    We let $\tilde{\epsilon} = \epsilon/m$.
  For instances $\calG^{(i,j)}_l$ in the ``Uniform case'', we let
  \[
    \tau^{(i,j)}_l = c_{sample} \log(n/\tilde{\epsilon}) 
    \left( \frac{\alpha}{m^{(i,j)}_l}      
      + \frac{1}{\phi^2 d_{\min}} \right)
    \leq C_3 \log(n/\tilde{\epsilon}) \log^{12}(n / \tilde{\delta})
    \left(  \frac{1}{m^{(i,j)}_l}  +
      \frac{n^{(i,j)}}{m^{(i,j)}}
    \right)
  \]
  for some universal constant $C_3$,
and we call $\textsc{SampleAndFixGradient}(\calG^{(i,j)}_l,
\tau^{(i,j)}_l)$, which 
returns instance $\calH^{(i,j)}_l$ 
and maps
$\map{\calG^{(i,j)}_l}{\calH^{(i,j)}_l}$ and $\map{\calH^{(i,j)}_l}{\calG^{(i,j)}_l}$.

For instances $\calG^{(i,j)}_l$ in the ``small case'',
  we define $\calG'^{(i,j)}_l$ to be the instance $\calG^{(i,j)}_l$
  with all gradient entries set to zero.
  We let 
  \[
    \tau^{(i,j)}_l = c_{sample} \log(n/\tilde{\epsilon}) 
    \frac{1}{\phi^2 d_{\min}} 
    \leq C_3 \log(n/\tilde{\epsilon}) \log^{12}(n / \tilde{\delta})
    \frac{n^{(i,j)}}{m^{(i,j)}}
  \]
  And call 
  $\textsc{SampleAndFixGradient}(\calG'^{(i,j)}_l, \tau^{(i,j)}_l)$
  which again returns an instance $\calH^{(i,j)}_l$ 
and maps
$\map{\calG^{(i,j)}_l}{\calH^{(i,j)}_l}$ and
$\map{\calH^{(i,j)}_l}{\calG^{(i,j)}_l}$.

The instance $\calH^{(i,j)}_l$ has the same vertex set as
$\calG^{(i,j)}_l$ and at most $\tau^{(i,j)}_l m^{(i,j)}_l$ edges.

In the ``uniform case'', with probability at least $1-\tilde{\epsilon}$, the maps 
certify 
  \[
    \calH^{(i,j)}_l \preceq_{\kappa} \calG^{(i,j)}_l\text{ and } \calG^{(i,j)}_l \preceq_{\kappa}
    \calH^{(i,j)}_l
    ,
  \]
  
  where $\kappa \leq  m^{1/(p-1)} \log^{42}\left(n / \tilde{\delta}\right) $,
  directly by Theorem~\ref{thm:sampAndFixGrad}.
In the ``small case'', with probability at least $1-\tilde{\epsilon}$, the maps 
certify with the same $\kappa$ that 
\begin{align}
  \label{eq:smallcaseadditiveexpanderapx}
    \calH^{(i,j)}_l  \preceq_{\kappa,\delta} \calG^{(i,j)}_l \text{ and } \calG^{(i,j)}_l \preceq_{\kappa,\delta}
    \calH^{(i,j)}_l ,
\end{align}
  because by Theorem~\ref{thm:sampAndFixGrad},
  \[
    \calH^{(i,j)}_l \preceq_{\kappa} \calG'^{(i,j)}_l\text{ and } \calG'^{(i,j)}_l \preceq_{\kappa}
    \calH^{(i,j)}_l
    ,
  \]
  and our rounding of the gradients implies (certified by the identity map)
  that 
  \[
    \calG^{(i,j)}_l \preceq_{1,\tilde{\delta} \norm{\gg}_{\infty} } \calG'^{(i,j)}_l
    \text{ and }
    \calG'^{(i,j)}_l \preceq_{1,\tilde{\delta} \norm{\gg}_{\infty}
    } \calG^{(i,j)}_l
    ,
  \]
  with finally our choice of $\tilde{\delta} = 2^{\log^c(n)} \delta$ for some
  large enough universal constant $c$ ensures
  $\tilde{\delta} \norm{\gg}_{\infty} \leq \delta$, and by
  Lemma~\ref{lem:approximations:composition},
  the guarantees compose to give Equation~\eqref{eq:smallcaseadditiveexpanderapx}.
  By Lemma~\ref{lem:approximations:composition}, the guarantees
  compose and 
we get that the total edge count in the sampled instances $\calH^{(i,j)}_l$ 
associated with $\calG^{(i,j)}$ is bounded by
\begin{align*}
  \sum_{l} 2\tau^{(i,j)}_l m^{(i,j)}_l
  & \leq
  \sum_{l} 
  2C_3 \log(n/\tilde{\epsilon}) \log^{12}(n / \tilde{\delta})
    \left(  1 + n^{(i,j)}  
      \frac{m^{(i,j)}_l  }{m^{(i,j)}}
    \right)
 \\
  &\leq
    3C_3 \log(n/\tilde{\epsilon}) \log^{12}(n / \tilde{\delta}) n^{(i,j)}
 \\
  &\leq
    3C_3 \log(n/\tilde{\epsilon}) \log^{12}(n / \tilde{\delta}) n
    .
\end{align*}
Since the number of buckets (indexed by $i$) is bounded by
$\polylog(n)$ and the number $\textsc{Decompose}$ calls for each
bucket (indexed by $j$) is bounded by $\log(n)$, we get that the total
number of edges summed across all the $\calH^{(i,j)}_l$ for all
$i,j,l$ is bounded by
\[
  n \log(1/\tilde{\epsilon}) \log^{12}(1/\tilde{\delta}) \polylog(n)
\]
Summed over all $i$, the total number of edges in the
$\calG^{(i,j_i)}$ instances is 
\[
  n \polylog(n)
  .
\]
We return a sparsifier instance $\calH$ consisting of the union over all $i$
and over all $j$ of
the $\calH^{(i,j)}_l$ and the $\calG^{(i,j_i)}$, and $\calG^{(0)}$
(the bucket with edges where $2$ and $p$ weights are both zero),
with a total number
of edges bounded by
\[
  n \log(1/\tilde{\epsilon}) \log^{12}(1/\tilde{\delta}) \polylog(n)
  \leq
  n \polylog(n/(\epsilon \delta))
  .
\]
Because the union of the original instances $\calG^{(i,j_i)}$, and
$\calG^{(0)}$ gives us $\calG'$,
by Lemma~\ref{lem:union},
\[
  \calG' \preceq_{\kappa,\delta} \calH
  \text{ and }
  \calH \preceq_{\kappa,\delta} \calG'
\]
where again
$\kappa \leq m^{1/(p-1)} \log^{42}\left(n / \tilde{\delta}\right)
\leq m^{1/(p-1)} \polylog(n/\delta)$.
Then, because Equation~\eqref{eq:roundapx} holds using the identity
map between $\calG$ and $\calG'$, we have
\[
  \calG \preceq_{2\kappa,\delta} \calH
  \text{ and }
  \calH \preceq_{2\kappa,\delta} \calG
\]
\todo{need ot check $\kappa$ $\delta$ interaction under composition}

We sparsified at most $m$ different expanders, since each contains a
distinct edges of $\calG$, and each sparsification fails with
probability at most $\tilde{\epsilon} = \epsilon/m$,
so by a union bound, the probability none of the sparsifications fail
is at least $1-\epsilon$.
That weights of $\calH$ are quasipolynomially bounded follows from the explicit
weights given in Theorem~\ref{thm:sampAndFixGrad}.
The overall time bound to apply the union map also follows immediately
from Theorem~\ref{thm:sampAndFixGrad}.
\end{proof}

\subsection*{Proof of Theorem \ref{thm:Flow-Sparsification}}
\begin{proof}
We will use Theorem \ref{thm:instancesparse}. This requires that the
instance is cycle-touching. So we first convert out instance $\calG$
to $\calG'$ using Lemma \ref{lem:cyclePreprocess}. We thus have
maps $\calM_{\calG \to \calG'}$ and $\calM_{\calG' \to \calG}$ that can be applied in
$O(m)$ time, where $\norm{\calM_{\calG \to \calG'}}_{1 \to 1}^{\text{cycle}}\leq 1$ and $\calG' \circapprox_{1,0} \calG$, $\calG \circapprox_{1,0}\calG'$ (since we are solving the residual problem, our demand vector is $0$).  We can apply Theorem \ref{thm:instancesparse} on $\calG'$ to get  an instance $\calH$ with at most $m_{\calH} = n \polylog(n/(\epsilon\delta))$ edges, maps $\calM'$ that can be computed in $\Otil(m)$ time, and 
 \[
  \calH \circapprox_{\kappa,\delta} \calG' \quad \text{and} \quad \calG' \circapprox_{\kappa,\delta} \calH
  \]
  where $\kappa = m^{-1/(p-1)}$. Now to go from $\calG$ to $\calH$, we will compose these two approximations and we thus have from Lemma \ref{lem:approximations:composition},
  \[
    \calG \circapprox_{\kappa,\delta \norm{\calM_{\calG \to \calG'}}_{1 \to 1}} \calH    
    \text{ and }
    \calH \circapprox_{\kappa,\delta} \calG
  \]
Finally, as $\norm{\calM_{\calG \to \calG'}}_{1 \to 1} \leq 1$, this
completes our proof. We remark that any quasipolynomial blow-up in
this error would also be acceptable.
\end{proof}


\section{Sparsification for General $\ell_2^{2} + \ell_{p}^{p}$
  Objectives Using Lewis Weights}\label{sec:lewis-wt}

We will prove Theorem \ref{thm:LewisWt}.

 \subsection{Leverage Scores and Lewis Weights}
 For $\alpha \geq 1$, and $x, y > 0,$ we say $x \approx_{\alpha} y$ if
 $\frac{1}{\alpha}x \leq y \leq \alpha x$.
 The statistical leverage score of a row $\aa_i$ of a matrix $\AA$ is
 defined as
 \[
   \lewis{2, i}(\AA) \defeq \aa_i^\top (\AA^\top \AA)^{-1} \aa_i =
   \norm{(\AA^\top \AA)^{-1/2}\aa_i}_2^2,
 \]

 The generalization of statistical leverage scores to $\ell_{p}$-norms
 is given by $\ell_p$ Lewis weights~\cite{Lewis78}, which are defined as
 follows:
 \begin{definition}
   For a matrix $\AA$ and for $p \geq 1$, we define the $\ell_{p}$
   Lewis weights $\{\lewis{p,i}\}_{i}$ to be the unique weights such
   that
   \[
     \lewis{p, i} = \lewis{2, i}(\diag{\lewis{p, i}}^{1/2-1/p} \AA).
   \]
   Equivalently,
   \[
     \aa_i^\top \left(\AA^\top \diag{\lewis{p, i}}^{1-2/p} \AA\right)^{-1}
     \aa_i = \lewis{p, i}^{2/p}.
   \]
   When the matrix $\AA$ is not obvious from the context, we will
   denote the Lewis weights by $\lewis{p,i}(\AA).$
 \end{definition}


 We use $\approxlewis{p,i}$ to denote $\beta$-approximate Lewis
 weights, i.e., $\approxlewis{p,i} \approx_{\beta} \lewis{p,i}$.
 %

 \begin{lemma}[Foster's Theorem \cite{Foster53}]\label{lem:foster} For any matrix
   $\AA \in \rea^{m \times n},$ $m\geq n$, we have $\sum_i \lewis{2, i}(\AA) = rank(\AA) \leq n.$
 \end{lemma}
 As a simple corollary, we get that the $\ell_p$ Lewis weights also
 sum to $n.$
  \begin{corollary}\label{cor:FosterLewis}
  For any matrix
    $\AA \in \rea^{m \times n},$ $m \geq n,$ and any $p,$ we have
    $\sum_i \lewis{p, i}(\AA) \leq n.$
  \end{corollary}
 \begin{proof}
 By definition and existence of the Lewis weights,
 \[
 \sum_i \lewis{p, i}(\AA) = \sum_i \lewis{2, i}(\diag{\lewis{p, i}}^{1/2-1/p}\AA) = rank(\diag{\lewis{p, i}}^{1/2-1/p}\AA) \leq n,
 \]
 where the second equality follows from Lemma \ref{lem:foster}.
 \end{proof}

As we will see, having access to $\lewis{2,i}$ would allow us to determine a spectral approximation to $\AA$, though with many fewer rows. Unfortunately, the naive approach to calculating them requires computing $(\AA^\top \AA)^+$, which would defeat the purpose of finding a smaller spectral approximation in the first place. Thus, the key insight of work by \cite{CohenLMMPS15} is that a certain \emph{uniform sampling-based} approach is sufficient to determine approximate leverage scores, as established by the following important lemma.

 \begin{lemma}[Lemma 7, \cite{CohenLMMPS15}]\label{lem:ApproxLeverageScores}
Given matrix $\AA$, $p \in [2,4)$, $\theta < 1$, and a matrix $\BB$ containing $O(n\log(n))$ rescaled rows of $\AA$, there is an algorithm that, w.h.p. in $n$, computes $n^\theta$-approximate $\lewis{2,i}$ for $\AA$ in time $O((\LSS(\BB) + \nnz(\AA))\theta^{-1})$, where $\LSS(\BB)$ is the time required to solve a linear equation in $\BB^{\top} \BB$ to quasipolynomial accuracy.
 \end{lemma}

While the previous lemma provides approximations to $\lewis{2,i}$, it was later shown by \cite{CohenP15} that in fact we may use such a routine as a black-box for determining approximations $\lewis{p,i}(\AA)$, for $p \in [2,4)$.

 \begin{lemma}[Lemma 2.4, \cite{CohenP15}]\label{lem:ApproxLewisWeights}
   For any fixed $p<4$, given a routine {\sc ApproxLeverageScores} for
   computing, with high probability in $n$, $\beta$-approximate statistical leverage scores of rows
   of matrices of the form $\WW\AA$ for $\beta = n^{\theta\frac{1-|p/2 -1|}{p}}$,
   there is an algorithm {\sc ApproxLewisWeights} $n^{\theta}$-approximate $\ell_p$ Lewis
   weights for $\AA$ with $O\left(\frac{\log(\theta^{-1})}{1-|p/2-1|}\right)$ calls to {\sc
     ApproxLeverageScores}.
 \end{lemma}

Combining these two lemmas, we arrive at the following overall computational cost for finding $\approxlewis{p,i}$.

  \begin{theorem}\label{thm:ComputeApproxLewis}
Given matrix $\AA$, $p \in [2,4)$, $\theta < 1$, there is an algorithm that computes $n^\theta$-approximate $\lewis{p,i}$ for $\AA$ in time 
\[
O\left(\frac{p}{(1-|p/2-1|)^2}(\LSS(\BB) + \nnz(\AA))\theta^{-1}\log(\theta^{-1})\right),
\]
where $\LSS(\BB)$ is the time required to solve a linear
equation in $\BB^{\top} \BB$ to quasipolynomial accuracy for some matrix $\BB$ containing $O(n\log(n))$ rescaled rows of $\AA$.
 \end{theorem}
\begin{proof}
The theorem follows immediately from Lemmas \ref{lem:ApproxLeverageScores} and \ref{lem:ApproxLewisWeights}.
\end{proof}

 \begin{lemma}[$\ell_2$ Matrix Concentration Bound (Lemma 4, \cite{CohenLMMPS15})]
   \label{lem:l2-wts}
   There exists an absolute constant $C_{2}$ such that for any matrix
   $\AA \in \rea^{m \times n},$ and any set of sampling values $\nu_i$
   satisfying
   \[
     \nu_i \geq \lewis{2, i}(\AA) \cdot C_{2}\epsilon^{-2} \log n,
   \]
   if we generate a matrix $\SS$ with $N = \sum_i \nu_i$ rows, each
   chosen independently as the $i^{th}$ standard basis vector times
   $\frac{1}{\sqrt{\nu_i}}$, with probability $\frac{\nu_i}{N}$, then
   with probability at least $1-\frac{1}{n^{\Omega(1)}}$ we have
   \[
     \forall \xx \in \rea^{n},\ \|\SS\AA\xx\|_2
     \approx_{1+\epsilon}\|\AA\xx\|_2.
   \]
 \end{lemma}

 \begin{lemma}[\cite{BourgainLM89},~\cite{CohenP15}(Lemma 7.1)]
   \label{lem:lp-wts}
   For $p \ge 2$, there exists an absolute constant $C_{p}$ such that
   for any matrix $\AA \in \rea^{m \times n},$ and any set of sampling
   values $\nu_i$ satisfying
   \[
     \nu_i \geq \lewis{p, i}(\AA) \cdot C_{p} n^{\frac{p}{2}-1}
     \epsilon^{-5} \log n \log\nfrac{1}{\epsilon},
   \]
   if we generate a matrix $\SS$ with $N = \sum_i \nu_i$ rows, each
   chosen independently as the $i^{th}$ standard basis vector times
   $\frac{1}{\nu_i^{1/p}}$, with probability $\frac{\nu_i}{N}$, then
   with probability at least $1-\frac{1}{n^{\Omega(1)}}$ we have
   \[
     \forall \xx \in \rea^{n},\ \|\SS\AA\xx\|_p
     \approx_{1+\epsilon}\|\AA\xx\|_p.
   \]
 \end{lemma}

  \begin{lemma}\label{lem:LewisSparse}
    For $p \in [2,4),$ given matrices
    $\CC, \DD \in \rea^{m \times n},$ there exist
    $\nu_i > 0, i \in [m]$ with
    $N = \sum_i \nu_i \leq O(1) n^{\nfrac{p}{2}} \log n$ such that, if
    we generate a matrix $\SS$ with $N$ rows, each chosen
    independently as the $i^{th}$ standard basis vector times
    $\frac{1}{\nu_i^{1/p}}$ with probability $\frac{\nu_i}{N}$, then
    we can compute a diagonal matrix
    $\RR \in \rea^{N \times N}_{\ge 0}$ such that with probability at
    least $1- \frac{1}{n^{\Omega(1)}}$,
    \[
      \forall \xx \in \rea^{n},\ \|\RR\SS\CC\xx\|_2
      \approx_{2} \|\CC\xx\|_2 \textrm{ and }
      \|\SS\DD\xx\|_p \approx_{2} \|\DD\xx\|_p.
    \]
  \end{lemma}
  \begin{proof}
    Let $\approxlewis{2,i}(\CC)$ be $2$-approximate leverage scores of
    $\CC$ and $\approxlewis{p,i}(\DD)$ be $2$-approximate $\ell_p$
    Lewis weights for $\DD$. Define
    \[
      \nu_i = C_{2,p} \max\left\{ \approxlewis{2, i}(\CC) \cdot\log
        n,\ \approxlewis{p, i}(\DD) \cdot n^{\frac{p}{2}-1} \log n
      \right\},
    \]
    where $C_{2,p}$ is a large enough absolute constant we specify later.
    Since
    $ \sum_i \approxlewis{2,i}(\CC) \le 2 \sum_i \lewis{2,i}(\CC) \leq
    2 n$ and
    $ \sum_i \approxlewis{p,i}(\DD) \le 2 \sum_i \lewis{p,i}(\DD) \leq
    2 n$ from Corollary~\ref{cor:FosterLewis}, we get
    $ N = \sum_i \nu_i \leq O(C_{2,p}) n^{\frac{p}{2}-1} \log n.$ Let $\SS$
    be as defined in the lemma statement, i.e.
    \[
      \SS_{ab} = \begin{cases} \frac{1}{\nu_b^{1/p}} & \text{if $b^{th}$
          basis vector is chosen for row $a$},\\ 0& \text{otherwise}.
      \end{cases}
    \]
    Let us assume for row $a$, we have chosen the $b^{th}$ basis
    vector. Now define the diagonal matrix $\RR$ as
    \[
      \RR_{aa} = \nu_b^{\frac{1}{p}-\frac{1}{2}}.
    \]
    Note that $\widetilde{\SS} = \RR\SS$ is a matrix with $N$ rows,
    each chosen independently as the $i^{th}$ standard basis vector
    times $\frac{1}{\nu_i^{1/2}}$ with probability
    $\frac{\nu_i}{N}$. We can pick $C_{2,p}$ large enough so that
    $\nu_i \geq \lewis{2, i}(\CC) \cdot C_{2} \log n,$ and we can
    apply Lemma~\ref{lem:l2-wts} for some constant $\epsilon < 1$ to
    obtain that with probability at least
    $1- \frac{1}{n^{\Omega(1)}},$ we have
    \begin{equation}
      \label{eq:1}
      \forall \xx \in \rea^{n},\
      \| \RR\SS\CC\xx\|_2 =
      \| \widetilde{\SS}\CC\xx\|_2
      \approx_{2} \|\CC\xx\|_2.
    \end{equation}
    Similarly, we can pick $C_{2,p}$ large enough so that we have
    $\nu_i \geq \lewis{p, i}(\DD) \cdot n^{\frac{p}{2}-1} \log
    n$. Thus, using Lemma \ref{lem:lp-wts}, we get that with
    probability at least $1- \frac{1}{n^{\Omega(1)}},$ we have
    \begin{equation}
      \label{eq:2}
      \forall \xx \in \rea^{n},\ 
      \|\SS\DD\xx\|_p
      \approx_{2} \|\DD\xx\|_p.
    \end{equation}
    Combining the above two claims, and applying a union bound, we
    obtain our lemma.
\end{proof}

\begin{lemma}\label{lem:ResProblemLewisApprox}
Let $p \in [2,4)$, let $\MM, \NN, \AA$ be matrices such that $\MM \in \rea^{m_1\times n}$, $\NN \in \rea^{m_2 \times n}$, $m_1, m_2 \geq n$, and $\AA \in \rea^{d \times n}$, $d \leq n$, and consider the problem
 \begin{align}\label{eq:QuadPlusP}
 \min_{\Delta}  & \quad \Delta^{\top}\MM^{\top}\MM\Delta + \|\NN\Delta\|_p^p \\
 & \text{s.t.} \quad \AA\Delta = \cc, \nonumber
 \end{align}
 with optimum at $\Delta^*$. Then, with high probability we may compute $\MMtil, \NNtil\in \rea^{O(n^{p/2}\log(n))\times n}$ such that, for a $\kappa$-approximate solution $\Dtil$ to the problem
 \begin{align}
 \min_{\Delta}  & \quad \Delta^{\top}\MMtil^{\top}\MMtil\Delta + \|\NNtil\Delta\|_p^p \\
 & \text{s.t.} \quad \AA\Delta = \cc \nonumber
 \end{align}
with optimum at $\Dtil^*$, $\Dtil$ is a $O(\kappa)$-approximate solution to \eqref{eq:QuadPlusP}.
\end{lemma}

\begin{proof}
Let $\MMtil = \RR\SS\MM$ and $\NNtil = \SS\NN$ be as provided by Lemma \ref{lem:LewisSparse}. It follows that
\begin{align*}
\Dtil^\top \MM^\top \MM \Dtil + \|\NN \Dtil\|_p^p \leq 2^p\left(\Dtil^\top \MMtil^\top \MMtil \Dtil + \|\NNtil \Dtil\|_p^p\right) &\leq 2^p\kappa\left(\Dtil^{*\top}\MMtil^{\top}\MMtil\Dtil^* + \|\NNtil\Dtil^*\|_p^p\right)\\
&\leq 2^p\kappa\left(\Delta^{*\top}\MMtil^{\top}\MMtil\Delta^* + \|\NNtil\Delta^*\|_p^p\right)\\
&\leq 2^{2p}\kappa\left(\Delta^{*\top}\MM^{\top}\MM\Delta^* + \|\NN\Delta^*\|_p^p\right)\\
&\leq 256\kappa\left(\Delta^{*\top}\MM^{\top}\MM\Delta^* + \|\NN\Delta^*\|_p^p\right),
\end{align*}
where the last inequality follows from our bound on $p$.
\end{proof}

We now recall the main Lewis weights-based sparsification result, Theorem \ref{thm:LewisMain} which was proven in Section~\ref{sec:MainAlgoProofs}.


This result gives us the following corollaries which distinguish the general problem from the more structured graph problem, whereby the latter may take advantage of fast Laplacian solvers. Note that, for Theorem~\ref{thm:LewisMain} and its corollaries, we use $\tilde{O}_p(\cdot)$ to suppress a $(1-|p/2-1|)^{-2}$ term, which will become large as $p$ approaches $4$.

\begin{proof}[Proof of Theorem \ref{thm:LewisWt}]
Follows from Lemmas \ref{thm:ComputeApproxLewis} and \ref{lem:ResProblemLewisApprox}.
\end{proof}

\begin{corollary}[General matrix setting]
Consider \eqref{eq:Problem} for arbitrary $\MM \in \rea^{m_1\times n}, \NN \in \rea^{m_2 \times n}$. Then, for $p \in [2,4)$, with high probability, we can find an $\eps$-approximate solution in time
\[
\tilde{O}_p\left(\left(\nnz(\MM) + \nnz(\NN) + \left(\nnz(\MMtil)+\nnz(\NNtil)+n^\omega\right)n^{\frac{p(p-2)}{6p-4}}\right)\log^2(1/\eps)\right),
\]
for some $\MMtil$ and $\NNtil$ each containing $O(n^{p/2}\log(n))$ rescaled rows of $\MM$ and $\NN$, respectively.
\end{corollary}

\begin{corollary}[Graph setting]
Consider \eqref{eq:Problem} for $\MM \in \rea^{m_1\times n}, \NN \in \rea^{m_2 \times n}$, $m_1, m_2 \geq n$, given as the edge-vertex incidence matrices for some graphs. Then, for $p \in [2,4)$, with high probability, we can find an $\eps$-approximate solution to \eqref{eq:Problem} in time
\[
\tilde{O}_p\left(\left(m_1 + m_2 + n^{\frac{p}{2}\left(1+\frac{p-2}{3p-2}\right)}\right)\log^2(1/\eps)\right).
\]
\end{corollary}


\renewcommand{\union}{\cup}
\section{Width-Reduced Approximate Solver for
  $\ell_{2}^{2} + \ell_{p}^{p}$ Problems}
 \label{sec:AKPSOracle}
 We will solve problems of the form,
  \begin{align}\label{eq:Un-ScaledProblem}
 \min_{\Delta}  & \quad \Delta^{\top}\MM^{\top}\MM\Delta + \|\NN\Delta\|_p^p \\
 & \text{s.t.} \quad \AA\Delta = \cc, \nonumber
  \end{align}

 which have an optimum at most $\nu$. We first scale the problem down to a new problem with optimum at most $1$. Note that there exists $\Dopt$ such that $\AA\Dopt = \cc$ and $\Dopt^{\top}\MM^{\top}\MM\Dopt + \|\NN\Dopt\|_p^p \leq \nu$. Let $\tilde{\MM} = \nu^{-{\frac{p-2}{2}}} \MM$ and $\Dtil = \nu^{-1/p}\Dopt$. The following problem has optimum at most $1$ since $\Dtil$ is a feasible solution.
   \begin{align}\label{eq:Scaled}
 \min_{\Delta}  & \quad \Delta^{\top}\tilde{\MM}^{\top}\tilde{\MM}\Delta + \|\NN\Delta\|_p^p \\
 & \text{s.t.} \quad \AA\Delta = \nu^{-1/p}\cc, \nonumber
 \end{align}
 Now, let $\Dbar$ denote a feasible solution such that $ \Dbar^{\top}\tilde{\MM}^{\top}\tilde{\MM}\Dbar \leq \alpha$ and $\|\NN\Delta\|_p^p \leq \beta$. Note that $\Delta = \nu^{1/p}\Dbar$ satisfies the constraints of \eqref{eq:Un-ScaledProblem} and,
 \[
 \Delta^{\top}\MM^{\top}\MM\Delta \leq \alpha \nu, \quad \text{and}, \quad   \|\NN\Delta\|_p^p \leq \beta \nu.
 \]
 It is thus sufficient to solve Problem \eqref{eq:Scaled} to an  $\alpha,\beta$ approximation. We thus have the following result which follows from Theorem \ref{lem:FasterAlgo},
 \deeksha{At some point we can make the general oracle work for non scaled problems. We do not need any scaling and can include the $\nu$'s in the parameters.}
 \AKPSAlgo*
\subsection{Solving Scaled Problem} 
 
We will show that we can solve problems of the form 
 \begin{align}\label{eq:ScaledProblem}
 \min_{\Delta}  & \quad \Delta^{\top}\MM^{\top}\MM\Delta + \|\NN\Delta\|_p^p \\
 & \text{s.t.} \quad \AA\Delta = \cc, \nonumber
 \end{align}
which have an optimum value at most $1$. We will use the following oracle in our algorithm.

\begin{algorithm}[H]
\caption{Oracle}
\label{alg:oracle}
 \begin{algorithmic}[1]
 \Procedure{\textsc{Oracle}}{$\AA, \MM,\NN, \cc, \ww$}
\State $\rr_e \leftarrow \ww_e^{p-2}$
\label{algline:resistance}
\State Compute,
\[
  \Delta = \arg\min_{\AA \Delta' =  \cc} \quad m_1^{\frac{p-2}{p}}{\Delta'}^{\top}\MM^{\top}\MM\Delta' + \frac{1}{3^{p-2}} \sum_e \rr_e \left(\NN\Delta^{'}\right)^2_e
\]
 \State\Return $ \Delta$
 \EndProcedure 
 \end{algorithmic}
\end{algorithm}
We can use now use Algorithm 4 from \cite{AdilKPS19}.
\begin{algorithm}
\caption{Algorithm for the Scaled down Problem}
\label{alg:FasterOracleAlgorithm}
 \begin{algorithmic}[1]
 \Procedure{Residual-Solver}{$\AA, \MM,\NN, \cc$}
 \State $\ww^{(0,0)}_e \leftarrow 1$
 \State $\xx \leftarrow 0$
 \State $\rho \leftarrow {\Theta}\left(m_1^{\frac{(p^2-4p+2)}{p(3p-2)}}\right)$\Comment{width parameter}\label{algline:defrho}
 \State $\beta \leftarrow {\Theta} \left(m_1^{\frac{p-2}{3p-2}}\right)$\Comment{resistance threshold}
 \State $\alpha \leftarrow {\Theta}\left(p^{-1} m_1^{-\frac{p^2-5p+2}{p(3p-2)}}\right)$\Comment{step size}
 \State $\tau \leftarrow {\Theta}\left(m_1^{\frac{(p-1)(p-2)}{(3p-2)}}\right)$\Comment{$\ell_p$ energy threshold}\label{algline:defalpha}
\State $T \leftarrow \alpha^{-1} m_1^{1/p} = {\Theta} \left( p m^{\frac{p-2}{3p-2}}\right)$        
\State{$i \leftarrow 0, k \leftarrow 0$}
\While{$i < T$} 
\State $\Delta = \textsc{Oracle}(\AA, \MM,\NN, \cc, \ww)$ \label{algline:InvokeOracle}
\If{$\norm{\NN\Delta}_{p}^p \leq \tau$} \hfill \Comment{flow step} \label{algline:CheckWidth}
\State $\ww^{(i+1,k)} \leftarrow \ww^{(i,k)} + \alpha \abs{\NN\Delta}$ \label{algline:LowWidth}
\State $\xx \leftarrow \xx +  \alpha \Delta$
\State $ i \leftarrow i+1$ 
\Else \hfill \Comment{width reduction step}
\State{For all edges $e$ with $|\NN\Delta|_e \geq \rho$ and $\rr_e \leq \beta$\label{lin:WidthReduceEdge}}
\Statex{\qquad \qquad \qquad $\ww_e^{(i,k+1)} \leftarrow 2^{\frac{1}{p-2}}\ww_e$}  \label{lin:widthReductionStepIncr}
\State{$ k \leftarrow k+1$}
\EndIf
\EndWhile
 \State\Return $m_1^{- \frac{1}{p}}{\xx}$
 \EndProcedure 
 \end{algorithmic}
\end{algorithm}

\paragraph*{Notation}
We will use $\Dopt$ to denote the optimum of \eqref{eq:ScaledProblem} and $\Dtil$ to denote the solution returned by the oracle (Algorithm \ref{alg:oracle}). We thus have, 
\begin{tight_itemize}
\item $\Dopt^{\top}\MM^{\top}\MM\Dopt \leq 1$,
\item $\norm{\NN\Delta^*}_p \leq 1$
\item $\rr_e \geq 1, \forall e$.
\end{tight_itemize}
\noindent We will prove the following main theorem:
\begin{theorem}\label{lem:FasterAlgo}
Let $p\geq 2$.  Given matrices $\AA \in \mathbb{R}^{d \times n},\NN \in \mathbb{R}^{m_1 \times n},\MM \in \mathbb{R}^{m_2 \times n}$, $m_1,m_2 \geq n$, $d \leq n$, and vector $\cc$, Algorithm \ref{alg:FasterOracleAlgorithm} uses ${O}\left(p m_1^{\frac{p-2}{(3p-2)}}\right)$, calls to the oracle (Algorithm \ref{alg:oracle}) and returns a vector $\xx$ such that 
$\AA \xx = \cc,$ $\xx^{\top}\MM^{\top}\MM\xx \leq O(1)$ and $\|\NN\xx\|_p^p = O(3^p)$.
\end{theorem}
\subsection*{Analysis of Algorithm \ref{alg:FasterOracleAlgorithm}}
Similar to \cite{AdilKPS19} we will track two potentials $\Phi$ and $\Psi$ which we define as,
\[
  \Phi\left(\ww^{\left( i \right)} \right) \defeq \norm{\ww}_p^p
\]
\[
\Psi(\rr)\defeq \min_{\Delta: \AA \Delta = \cc} m_1^{\frac{p-2}{p}}{\Delta}^{\top}\MM^{\top}\MM\Delta + \frac{1}{3^{p-2}}\sum_e \rr_e \left(\NN\Delta\right)^2_e.
\]
Note that these potentials have a similar idea as \cite{AdilKPS19} but are defined differently. Our proof will follow the following structure,
\begin{enumerate}
\item Provided the total number of width reduction steps, $K$, is not too big, $\Phi(\cdot)$ is small.
This in turn helps upper bound the value of the solution returned by the algorithm.
\item Showing that $K$ cannot be too big, because each width reduction step cause large growth in ${\Psi(\cdot)}$, while we can bound the total growth in ${\Psi(\cdot)}$ by relating it to $\Phi(\cdot)$.
\end{enumerate}

We start by proving some results that we need to prove our final result, Theorem \ref{lem:FasterAlgo}. The proofs of all lemmas are in Section \ref{sec:MissingProofsAppendix}.

\begin{restatable}{lemma}{Oracle}
\label{lem:Oracle}
Let $p\geq 2$. For any $\ww$, let $\Dtil$ be the solution returned by Algorithm \ref{alg:oracle}. Then,
\[
\sum_e (\NN\Dtil)_e^2 \leq \sum_e \rr_e(\NN\Dtil)_e^2 \leq  \norm{\ww}^{p-2}
\]
\end{restatable}

We next show through the following lemma that the $\Phi$ potential does not increase
too rapidly. The proof is through induction and can be found in the Appendix.
\begin{restatable}{lemma}{ReduceWidthGammaPotential}
  \label{lem:ReduceWidthGammaPotential}
  After $i$ \emph{flow} steps, and $k$ width-reduction steps,
  provided
  \begin{enumerate}
 \item \label{enu:pPowerStep} $p^p \alpha^p \tau \leq p \alpha m_1^{\frac{p-1}{p}}$, (controls $\Phi$ growth in flow-steps)
 \item \label{enu:widthStepCondition}$k \leq 2^{-\frac{p}{p-2}}\rho^2 m_1^{2/p} \beta^{-\frac{2}{p-2}}$ ,(acceptable number of width-reduction steps)
  \end{enumerate}
  the potential $\Phi$ is bounded as follows:
  \begin{align*}
 \Phi\left(\ww^{(i,k)}\right) \leq \left(2\alpha i + m_1^{\nfrac{1}{p}} \right)^{p} \exp{\left(2^{\frac{p}{p-2}} \frac{k}{\rho^2 m_1^{2/p} \beta^{-\frac{2}{p-2}}} \right)}.
 \end{align*}
\end{restatable}
We next show how the potential $\Psi$ changes with a change in resistances. The proof is in the Appendix.
\begin{restatable}{lemma}{ckmstResIncrease}
  \label{lem:ckmst:res-increase}
Let $\Dtil = \arg\min_{\AA\Delta = c} m_1^{\frac{p-2}{p}}\Delta^{\top}\MM^{\top}\MM\Delta + \frac{1}{3^{p-2}}\sum_e \rr_e (\NN\Delta)_e^2$. Then one has for any $\rr'$ and $\rr$ such that $\rr' \geq \rr$,
\[
{\Psi({\rr'})} \geq {\Psi\left({\rr}\right)} + \sum_{e}\left(1-\frac{\rr_e}{\rr'_e}\right) \rr_e (\NN\Dtil)_e^2.
\]
\end{restatable}
\deeksha{for any constant $z$ in the diffrence of $\rr'$ and $\rr$, the factor just becomes $z+1$. Can maybe change that at some point. }
The next lemma gives a lower bound on the energy in the beginning and an upper bound on the energy at each step.

\begin{restatable}{lemma}{ElectricalPotentialStartFinishBounds}
  \label{lem:ElectricalPotentialStartFinishBounds}
 Initially, we have,
 \[
{\Psi\left({\rr^{\left(0,0\right)}}\right)} \ge \frac{\|\MM+\NN\|_{\min}^2\|\cc\|_2^2}{\norm{\AA}^2} \defeq L,
\]
where $\|\MM+\NN\|_{\min} = \min_{\AA\xx = \cc}\|\MM\xx\|_2^2 + \|\NN\xx\|_2^2$ and $\norm{\AA}$ is the operator norm of $\AA$. Moreover, at any step $(i,k)$ we have,
\[
{\Psi\left({\rr^{(i,k)}}\right) }
  \le   m_1^{\frac{p-2}{p}} + \frac{1}{3^{p-2}}\Phi(i,k)^{\frac{p-2}{p}} .
\]  
\end{restatable}

We next bound the change in energy in a flow step and a width reduction step. This lemma is directly from \cite{AdilKPS19} and the proof is also very similar. We include it here for completeness.
\begin{restatable}{lemma}{ReduceWidthElectricalPotential}
\label{lem:ReduceWidthElectricalPotential}
Suppose at step $(i,k)$ we have $\norm{\NN\Dtil}^p_{p} >\tau$ so that we perform a width reduction step (line~\ref{lin:widthReductionStepIncr}). If
  \begin{enumerate}
  \item \label{eq:parametersEnsuringElectricalPotentialGrowth1} $\tau^{2/p} \geq\Omega\left( 1\right) \frac{\Psi(\rr)}{\beta}$, and
\item \label{eq:parametersEnsuringElectricalPotentialGrowth2} $\tau \geq \Omega\left(1\right)\Psi(\rr) \rho^{p-2}$.
\end{enumerate}
Then
\[
 {\Psi\left({\rr^{(i,k+1)}}\right)} \geq {\Psi\left({\rr^{(i,k)}}\right)} +\Omega\left(1\right) \tau^{2/p}.
\]
Furthermore, if at $(i,k)$ we have $\norm{\NN\Dtil}^p_{p} \leq \tau$ so that we perform a flow step,
then
\[
    {\Psi\left({\rr^{(i,k+1)}}\right)} \geq {\Psi\left({\rr^{(i,k)}}\right)}.
\]
\end{restatable}
\subsection*{Proof of Theorem \ref{lem:FasterAlgo}}
\deeksha{p factors are all off}
\begin{proof}
We begin by setting all our parameter values.
\begin{itemize}
\item $\alpha \leftarrow \Theta \left(p^{-1}m_1^{-\frac{p^2-5p+2}{p(3p-2)}}\right)$
\item $\tau \leftarrow \Theta\left(m_1^{\frac{(p-1)(p-2)}{(3p-2)}}\right)$
\item $\beta = \Theta\left(m_1^{\frac{p-2}{3p-2}}\right)$
\item $\rho = \Theta\left(m_1^{\frac{(p^2-4p+2)}{p(3p-2)}}\right)$
\end{itemize}
Note that the above values satisfy the relations $p^p \alpha^p \tau = p \alpha m_1^{\frac{p-2}{p}}$. 
%

Let $m_1^{-1/p}\xx$ be the solution returned by Algorithm \ref{alg:FasterOracleAlgorithm}. Note that this satisfies the linear constraint required. We will now bound the values of $m_1^{-2/p}\xx^{\top}\MM^{\top}\MM\xx$ and $m_1^{-1}\|\NN\xx\|_p^p$. If the algorithm terminates in $T = \alpha^{-1} m_1^{1/p}$ flow steps and $K \leq 2^{-\frac{p}{p-2}} \rho^2 m_1^{2/p} \beta^{-\frac{2}{p-2}}$ width reduction steps, then from Lemma \ref{lem:ReduceWidthGammaPotential},
\[
\Phi\left(\ww^{(T,K)}\right) \le  O(3^p) m_1 e^{1} = O\left(3^p m_1\right)
\]
Note that throughout the algorithm $\ww \geq |\NN\xx|$. This means that the algorithm returns $m_1^{- \frac{1}{p}} \xx$ with 
\begin{align*}
m_1^{-1}\|\NN\xx\|_p^p \leq \frac{1}{m_1} \norm{\ww^{(T,K)}}_p^p  =\frac{1}{m_1} \Phi\left(\ww^{(T,K)}\right) \leq O(3^p).
\end{align*}
To bound the other term, let $\Dtil^{(t)}$ denote the solution returned by the oracle in iteration $t$. Note that, since $\Phi \leq O(3^p) m_1$ for all iterations, we always have $\Psi(\rr) \leq O(1) m_1^{\frac{p-2}{p}}$. We claim that $\left(\Dtil^{(t)}\right)^{\top}\MM^{\top}\MM\Dtil^{(t)} \leq O(1)$ for all $t$. To see this, note that from Lemma \ref{lem:ElectricalPotentialStartFinishBounds},
\[
m_1^{\frac{p-2}{p}} \left(\Dtil^{(t)}\right)^{\top}\MM^{\top}\MM\Dtil^{(t)} \leq \Psi(\rr) \leq O(1) m_1^{\frac{p-2}{p}}.
\]
We also know that $\xx = \sum_t \frac{\alpha}{p}\Dtil^{(t)}$. Combining this and the convexity of $\|\xx\|_2^2$, we get
\begin{align*}
m_1^{-2/p}\|\MM\xx\|_2^2 \leq {\alpha^2}m_1^{-2/p} T \sum_t \|\MM\Dtil^{(t)}\|_2^2 \leq \alpha^2 m_1^{-2/p} T^2 O(1) \leq O(1).
\end{align*}
This concludes the first part of the proof that if the number of width reduction steps are bounded, then we return a solution with the required values. We will now show that we cannot have more width reduction steps.

 Suppose to the contrary, the algorithm takes a width reduction step starting from step $(i,k)$ where $i < T$ and  $k = 2^{-\frac{p}{p-2}} \rho^2 m_1^{2/p}\beta^{-\frac{2}{p-2}}$. Since the conditions for Lemma~\ref{lem:ReduceWidthGammaPotential}
hold for all preceding steps, we must have $\Phi\left(\ww^{(i,k)}\right) \leq O(3^p) m_1$. We note that our parameter values satisfy $ \tau^{2/p} \geq \Omega(1)  \frac{\Psi}{\beta}$ and $\tau \geq \Omega(1)\rho^{p-2} \Psi$ since $\Psi \leq O(1) m_1^{\frac{p-2}{p}}$.

This means that at every step $(j,l)$ preceding the current step,
the conditions of Lemma~\ref{lem:ReduceWidthElectricalPotential} are
satisfied, so we can prove by a simple induction that 
\deeksha{Need to fix the $p$ dependencies}
\[
 {\Psi\left({\rr^{(i,k+1)}}\right)} \geq {\Psi\left({\rr^{(0,0)}}\right)}  +  \Omega\left(1\right)\tau^{2/p} k.  
\]
Since our parameter choices ensure $\tau^{2/p}k > \Theta\left(m_1\right)$,
\[
 {\Psi\left({\rr^{(i,k+1)}}\right)} - {\Psi\left({\rr^{(0,0)}}\right)}   >  \Omega\left(m_1\right).
\]
Since $\Phi\left(\ww^{(i,k)}\right) \leq O(3^p) m_1$, 
\[
\Psi\left({\rr^{(i,k+1)}}\right) - \Psi\left({\rr^{(0,0)}}\right)
\leq O\left(m_1^{\frac{p-2}{p}} \right),
\]
which is a contradiction. We can thus conclude that we can never have more than $K = 2^{\frac{-p}{p-2}}\rho^2 m_1^{2/p} \beta^{-\frac{2}{p-2}}$ width reduction steps, thus concluding the correctness of the returned solution. We next bound the number of oracle calls required. The total number of iterations is at most,
\[
T + K \leq \alpha^{-1}m_1^{1/p} +  2^{-p/(p-2)}\rho^2 m_1^{2/p} \beta^{-\frac{2}{p-2}} \leq O \left(p m_1^{\frac{p-2}{3p-2}}\right).
\]
\end{proof}

\subsection{Missing Proofs}
\label{sec:MissingProofsAppendix}
\Oracle*
\begin{proof}
Since $\Dtil$ is the solution returned by Algorithm \ref{alg:oracle}, and $\Dopt$ satisfies the constraints of the oracle, we have,
\begin{align*}
  \sum_e \rr_e (\NN\Dtil)_e^2 \leq \sum_e \rr_e (\NN\Delta^*)_e^2 & = \sum_e \ww_e^{p-2} (\NN\Delta^*)_e^2 \leq   \norm{\ww}^{p-2}_p.
\end{align*}
In the last inequality we use,
\begin{align*}
\sum_e \ww_e (\NN\Dopt)_e^2 & \leq \left(\sum_e (\NN\Dopt)_e^{2\cdot \frac{p}{2}} \right)^{2/p}  \left(\sum_e \abs{\ww_e}^{(p-2)\cdot \frac{p}{p-2}}\right)^{(p-2)/p}\\
& = \|\NN\Dopt\|_p^2 \|\ww\|_p^{(p-2)/p} \\
& \leq \|\ww\|_p^{(p-2)/p}, \text{since $\norm{\NN\Delta^*}_p \leq 1$ }.
 \end{align*}
 Finally, using $\rr_e \ge  1,$ we have
$\sum_e (\NN \Delta)_e^{2} \le \sum_e \rr_e (\NN\Delta)_e^{2},$ concluding the proof.
\end{proof}
%

\ReduceWidthGammaPotential*

\begin{proof}
  We prove this claim by induction. Initially, $i = k = 0,$
  and $\Phi\left(\ww^{(0,0)}\right) = m_1,$ and thus, the claim holds trivially. Assume that the claim holds for some $i,k \ge 0.$
We will use $\Phi$ as an abbreviated notation for $\Phi\left(\ww^{(i,k)}\right)$ below.
\paragraph*{Flow Step.}
For brevity, we use $\ww$ to denote $\ww^{(i,k)}$. If the next step is a \emph{flow} step, 

\begin{align*}
\Phi\left(\ww^{( i+1,k)} \right) = &\norm{ \ww^{(i,k)} + \alpha \abs{\NN\Dtil}}_p^p \\
\leq & \norm{\ww}_p^p +  \alpha \abs{(\NN\Dtil)}^{\top}\abs{\nabla \|\ww\|_p^p}  +  2p^2\alpha^2  \sum_e |\ww_e|^{p-2} \abs{\NN\Dtil}_e^2 + \alpha^p p^p \|\NN\Dtil\|_p^p\\
&   \quad \text{by Lemma B.1 of \cite{AdilPS19}}
\end{align*}
We next bound $\abs{(\NN\Dtil)}^{\top} \abs{\nabla \|\ww\|_p^p}$ as,
\[
\textstyle \sum_e \abs{(\NN\Dtil)_e} \abs{\nabla_e \|\ww\|_p^p}  \leq  p \norm{\ww}_p^{p-1}.
\]
Using Cauchy Schwarz's inequality,
\begin{align*}
 \left(\sum_e \abs{\NN\Dtil}_e \abs{\nabla_e \|\ww\|_p^p}\right)^2  = & p^2\left(\sum_e  \abs{\NN\Dtil}_e  |\ww_e|^{p-2}\abs{\ww_e}\right)^2\\
 \leq & p^2 \left(\sum_{e} |\ww_e|^{p-2} \ww_e^2\right) \left(\sum_e |\ww_e|^{p-2}(\NN\Dtil)_e^2\right) \\
  = & p^2 \|\ww\|_p^p\sum_e \rr_e (\NN\Dtil)_e^2\\
  \leq & p^2 \|\ww\|_p^{2p-2}
\end{align*}
We thus have,
\begin{align*}
 \sum_e \abs{\NN\Dtil}_e\abs{\nabla_e \|\ww\|_p^p} &\leq  p \|\ww\|_p^{p-1}.
\end{align*}

Using the above bound, we now have,

\begin{align*}
\Phi\left(\ww^{( i+1,k)} \right)\leq &  \norm{\ww}_p^p +  p\alpha \norm{\ww}_p^{p-1} +  2p^2 \alpha^2  \norm{\ww}_p^{p-2} + p^p \alpha^p  \|\NN\Dtil\|_p^p\\
\leq &  \norm{\ww}_p^p +  p\alpha \norm{\ww}_p^{p-1} +  2p^2 \alpha^2  \norm{\ww}_p^{p-2} + p \alpha m_1^{\frac{p-1}{p}},\\
      & \quad \text{by Assumption~\ref{enu:pPowerStep} of this
      Lemma}\\
 \intertext{ Recall $\|\ww\|_p^p = \Phi(\ww).$ Since $\Phi \geq m_1$, we have,}
    \leq
    & \Phi(\ww) +  p \alpha \Phi(\ww)^{\frac{p-1}{p}} +  2 p^2 \alpha^{2}  \Phi(\ww)^{\frac{p-2}{p}} +  p \alpha \Phi(\ww)^{\frac{p-1}{p}} \\
  \leq
  & (\Phi(\ww)^{1/p} + 2  \alpha)^{p}.
\end{align*}

  From the inductive assumption, we have
  \begin{align*}
  \Phi(\ww) & \le
    \left({2\alpha  i} + m_1^{\nfrac{1}{p}} \right)^{p}
    \exp\left(O_p(1) \frac{k}{\rho^2 m_1^{2/p} \beta^{-\frac{2}{p-2}}} \right).
       \end{align*}
    Thus,
    \[
      \Phi(i+1,k)  \le (\Phi(\ww)^{1/p} + 2 \alpha )^{p}
      \le \left({2\alpha (i+1)} + m_1^{\nfrac{1}{p}}\right)^{p} \exp\left(O_p(1) \frac{k}{\rho^2 m_1^{2/p} \beta^{-\frac{2}{p-2}}} \right)
    \]
    proving
    the inductive claim.

 \paragraph*{Width Reduction Step.}
 We have the following:
\[
\sum_{e \in H} \rr_e \leq \rho^{-2} \sum_{e \in H} \rr_e (\NN\Delta)_e^2 \leq\rho^{-2}  \sum_{e} \rr_e (\NN\Delta)_e^2 \leq \rho^{-2}  \|\ww\|_p^{p-2}\leq \rho^{-2}\Phi^{\frac{p-2}{p}},
\]
and
\begin{align*}
\Phi(i,k+1) & \le \Phi + \sum_{e \in H} \abs{ \ww_e^{k+1}}^p\\
&  \leq \Phi + 2^{\frac{p}{p-2}} \sum_{e \in H}  |\ww_e|^{p}\\
& \leq \Phi + 2^{\frac{p}{p-2}} \sum_e \rr_e^{\frac{p}{p-2}} \\
& \leq  \Phi + 2^{\frac{p}{p-2}} \left( \sum_{e \in H} \rr_e \right)\left( \max_{e \in H} \rr_e \right) ^{\frac{p}{p-2}-1}  \\
& \leq \Phi + 2^{\frac{p}{p-2}}  \rho^{-2}\Phi^{\frac{p-2}{p}}\beta^{\frac{2}{p-2}}.
\end{align*}
Again, since $\Phi(\ww) \geq m_1$,
 \[
 \Phi(i,k+1)  \le \Phi \left( 1+ 2^{\frac{p}{p-2}} \rho^{-2} m_1^{-\frac{2}{p}}\beta^{\frac{2}{p-2}} \right) \le \left(2\alpha  i + m_1^{\nfrac{1}{p}}\right)^{p} \exp\left(2^{\frac{p}{p-2}}  \frac{k+1}{\rho^2 m_1^{2/p} \beta^{-\frac{2}{p-2}}} \right)
    \]
    proving
    the inductive claim.
\end{proof}

\ckmstResIncrease*

\begin{proof}
\[
\Psi(\rr) = \min_{\AA\xx = \cc} \xx^{\top}\MM^{\top}\MM\xx + \xx^{\top}\NN^{\top}\RR\NN\xx.
\]
Constructing the Lagrangian and noting that strong duality holds,
\begin{align*}
\Psi(\rr) &= \min_{\xx}\max_{\yy} \quad \xx^{\top}\MM^{\top}\MM\xx + \xx^{\top}\NN^{\top}\RR\NN\xx + 2\yy^{\top}(\cc-\AA\xx)\\
& =\max_{\yy} \min_{\xx} \quad \xx^{\top}\MM^{\top}\MM\xx + \xx^{\top}\NN^{\top}\RR\NN\xx + 2\yy^{\top}(\cc-\AA\xx).
\end{align*}
Optimality conditions with respect to $\xx$ give us,
\[
2\MM^{\top}\MM\xx^{\star} + 2\NN^{\top}\RR\NN\xx^{\star} = 2\AA^{\top}\yy.
\]
Substituting this in $\Psi$ gives us,
\[
\Psi(\rr) = \max_{\yy}\quad 2\yy^{\top}\cc  - \yy^{\top}\AA \left(\MM^{\top}\MM+ \NN^{\top}\RR\NN\right)^{-1} \AA^{\top}\yy.
\]
Optimality conditions with respect to $\yy$ now give us,
\[
2\cc  =  2 \AA \left(\MM^{\top}\MM+ \NN^{\top}\RR\NN\right)^{-1} \AA^{\top} \yy^{\star},
\]
which upon re-substitution gives,
\[
\Psi(\rr) =  \cc^{\top}\left(\AA \left(\MM^{\top}\MM+ \NN^{\top}\RR\NN\right)^{-1} \AA^{\top}\right)^{-1} \cc.
\]
We also note that 
\begin{equation}\label{eq:optimizer}
\xx^{\star} = \left(\MM^{\top}\MM+ \NN^{\top}\RR\NN\right)^{-1}\AA^{\top}\left(\AA \left(\MM^{\top}\MM+ \NN^{\top}\RR\NN\right)^{-1} \AA^{\top}\right)^{-1}\cc.
\end{equation}
We now want to see what happens when we change $\rr$. Let $\RR$ denote the diagonal matrix with entries $\rr$ and let $\RR' = \RR+\SS$, where $\SS$ is the diagonal matrix with the changes in the resistances. We will use the following version of the Sherman-Morrison-Woodbury formula multiple times,
\[
(\XX + \UU\CC\VV)^{-1} = \XX^{-1} - \XX^{-1}\UU(\CC^{-1} + \VV\XX^{-1}\UU)^{-1}\VV\XX^{-1}.
\]
We begin by applying the above formula for $\XX = \MM^{\top}\MM+ \NN^{\top}\RR\NN$, $\CC = \II$, $\UU = \NN^{\top}\SS^{1/2}$ and $\VV = \SS^{1/2}\NN$. We thus get,
\begin{multline}
\left(\MM^{\top}\MM+ \NN^{\top}\RR'\NN\right)^{-1} = \left(\MM^{\top}\MM+ \NN^{\top}\RR\NN\right)^{-1} -  \left(\MM^{\top}\MM+ \NN^{\top}\RR\NN\right)^{-1}\NN^{\top}\SS^{1/2} \\
\left(\II + \SS^{1/2}\NN \left(\MM^{\top}\MM+ \NN^{\top}\RR\NN\right)^{-1}\NN^{\top}\SS^{1/2}\right)^{-1}\SS^{1/2}\NN \left(\MM^{\top}\MM+ \NN^{\top}\RR\NN\right)^{-1}.
\end{multline}
We next claim that, 
\[
\II + \SS^{1/2}\NN \left(\MM^{\top}\MM+ \NN^{\top}\RR\NN\right)^{-1}\NN^{\top}\SS^{1/2} \preceq \II + \SS^{1/2}\RR^{-1}\SS^{1/2} ,
\]
which gives us,
\begin{multline}
\left(\MM^{\top}\MM+ \NN^{\top}\RR'\NN\right)^{-1} \preceq \left(\MM^{\top}\MM+ \NN^{\top}\RR\NN\right)^{-1} -  \\ \left(\MM^{\top}\MM+ \NN^{\top}\RR\NN\right)^{-1}\NN^{\top}\SS^{1/2}(\II + \SS^{1/2}\RR^{-1}\SS^{1/2})^{-1}\SS^{1/2}\NN \left(\MM^{\top}\MM+ \NN^{\top}\RR\NN\right)^{-1}.
\end{multline}
This further implies,
\begin{multline}
\AA\left(\MM^{\top}\MM+ \NN^{\top}\RR'\NN\right)^{-1}\AA^{\top} \preceq \AA\left(\MM^{\top}\MM+ \NN^{\top}\RR\NN\right)^{-1} \AA^{\top} -  \\ \AA\left(\MM^{\top}\MM+ \NN^{\top}\RR\NN\right)^{-1}\NN^{\top}\SS^{1/2}(\II + \SS^{1/2}\RR^{-1}\SS^{1/2})^{-1}\SS^{1/2}\NN \left(\MM^{\top}\MM+ \NN^{\top}\RR\NN\right)^{-1}\AA^{\top}.
\end{multline}
We apply the Sherman-Morrison formula again for, $\XX =\AA\left( \MM^{\top}\MM+ \NN^{\top}\RR\NN\right)^{-1}\AA^{\top}$, $\CC = -(\II + \SS^{1/2}\RR^{-1}\SS^{1/2})^{-1}$, $\UU = \AA\left(\MM^{\top}\MM+ \NN^{\top}\RR\NN\right)^{-1}\NN^{\top}\SS^{1/2}$ and $\VV = \SS^{1/2}\NN \left(\MM^{\top}\MM+ \NN^{\top}\RR\NN\right)^{-1}\AA^{\top}$. Let us look at the term $\CC^{-1} + \VV\XX^{-1}\UU$.
\[
-\left(\CC^{-1} + \VV\XX^{-1}\UU\right)^{-1} =    \left(\II + \SS^{1/2}\RR^{-1}\SS^{1/2} -  \VV\XX^{-1}\UU\right)^{-1}  \succeq   (\II + \SS^{1/2}\RR^{-1}\SS^{1/2})^{-1}.
\]
Using this, we get,
\[
\left(\AA\left(\MM^{\top}\MM+ \NN^{\top}\RR'\NN\right)^{-1}\AA^{\top}\right)^{-1} \succeq \XX^{-1} + \XX^{-1}\UU(\II + \SS^{1/2}\RR^{-1}\SS^{1/2})^{-1}\VV\XX^{-1},
\]
which on multiplying by $\cc^{\top}$ and $\cc$ gives,
\[
\Psi(\rr') \geq \Psi(\rr) +  \cc^{\top} \XX^{-1}\UU(\II + \SS^{1/2}\RR^{-1}\SS^{1/2})^{-1}\VV\XX^{-1} \cc.
\]
We note from Equation \eqref{eq:optimizer} that $\xx^{\star} = \left(\MM^{\top}\MM+ \NN^{\top}\RR\NN\right)^{-1}\AA^{\top} \XX^{-1}\cc$.
We thus have,
\begin{align*}
\Psi(\rr') &\geq \Psi(\rr) +  \left(\xx^{\star}\right)^{\top} \NN^{\top}\SS^{1/2} (\II + \SS^{1/2}\RR^{-1}\SS^{1/2})^{-1}\SS^{1/2}\NN\xx^{\star}\\
& = \Psi(\rr) + \sum_e \left(\frac{\rr'_e -\rr_e}{\rr'_e}\right)\rr_e (\NN\xx^{\star})_e \\
\end{align*}

\end{proof}

\ElectricalPotentialStartFinishBounds*
\begin{proof}
For the lower bound in the initial state, since $\rr_e \geq 1$, for any solution $\Delta$, we have
\[
\Psi(\rr^{(0, 0)}) \geq m^{\frac{p-2}{p}}\Dtil^{\top}\MM^{\top}\MM\Dtil + \frac{1}{3^{p-2}}\sum_{e} \rr^{\left( 0, 0 \right)} (\NN\Dtil)_e^2 = \|\MM\Dtil\|_2^2 + \frac{1}{3^{p-2}} \|\NN\Dtil\|_2^2 \geq \frac{1}{3^{p-2}} \|\MM+\NN\|_{\min}^2\|\Dtil\|_2^2 ,
\]
where $\|\MM+\NN\|^2_{\min} = \min_{\AA\xx = \cc} \|\MM\xx\|_2^2 + \|\NN\xx\|_2^2$. Note that if $\|\MM+\NN\|^2_{\min} = 0$ then the oracle has returned the optimum in the first iteration.
On the other hand, because
\[
\|\cc\|_2 = \norm{\AA \Dtil}_{2} \leq
\norm{\AA} \norm{\Dtil}_2,
\]
where $\|\AA\|$ is the operator norm of $\AA$. We get
\[
\norm{\Dtil}_2 \geq \frac{\|\cc\|_2}{\norm{\AA}},
\]
upon which squaring gives the lower bound on ${\Psi(\rr^{(0, 0)})}$.

\noindent For the upper bound, Lemma~\ref{lem:Oracle} implies that,
  \begin{align*}
{\Psi\left({\rr^{(i,k)}}\right) } & =m_1^{\frac{p-2}{p}} \Dtil^{\top}\MM^{\top}\MM\Dtil + \frac{1}{3^{p-2}}\sum_e \rr_e (\NN\Dtil)_e^2 \\
& \le m_1^{\frac{p-2}{p}} \Dopt^{\top}\MM^{\top}\MM\Dopt + \frac{1}{3^{p-2}}\sum_e \rr_e (\NN\Dopt)_e^2\\
& \le  m_1^{\frac{p-2}{p}} + \frac{1}{3^{p-2}} \norm{\ww}_p^{p-2} \\
& \le m_1^{\frac{p-2}{p}} + \frac{1}{3^{p-2}}\Phi(i,k)^{\frac{p-2}{p}}.
  \end{align*}
\end{proof}

\ReduceWidthElectricalPotential*

\begin{proof}
It will be helpful for our analysis to split the index set into three
disjoint parts:
\begin{itemize}
\item $S =  \setof{e : \abs{\NN\Delta_e} \leq \rho }$
\item $H = \setof{e : \abs{\NN\Delta_e} > \rho \text{ and } \rr_e \leq \beta }$
\item $B = \setof{e : \abs{\NN\Delta_e} > \rho \text{ and } \rr_e > \beta }$.
\end{itemize}
Firstly, we note 
\begin{align*}
\sum_{e \in S} \abs{\NN\Delta}_e^{p} \leq \rho^{p-2} \sum_{e \in S} \abs{\NN\Delta}_e^{2} 
 \leq \rho^{p-2}  \sum_{e \in S} \rr_e \abs{\NN\Delta}_e^{2}  \leq \rho^{p-2} \Psi(\rr).
\end{align*}
hence, using Assumption~\ref{eq:parametersEnsuringElectricalPotentialGrowth2}
\begin{align*}
  \sum_{e \in H \union B} \abs{\NN \Delta}_e^p
  \geq \sum_{e} \abs{\NN\Delta}_e^p - \sum_{e \in S}
  \abs{\NN\Delta}_e^{p}
  \geq \tau - \rho^{p-2} \Psi(\rr)
  \geq \Omega(1) \tau.
\end{align*}
This means, 
\[
\sum_{e \in H \union B}(\NN \Delta)_e^2 \geq  \left(\sum_{e \in H \union B} |\NN\Delta|_e^p\right)^{2/p} \geq \Omega\left(1\right) \tau^{2/p}.
\]
Secondly we note that, 
\[
\sum_{e \in B} (\NN\Delta)_e^2 \leq \beta^{-1} \sum_{e \in B} \rr_e (\NN\Delta)_e^2 \leq \beta^{-1}  \Psi(\rr).
\] 
So then, using Assumption~\ref{eq:parametersEnsuringElectricalPotentialGrowth1},
\begin{align*}
\sum_{e \in H } (\NN\Delta)_e^2 = \sum_{e \in H \union B } (\NN\Delta)_e^2 - \sum_{e \in B}(\NN \Delta)_e^2  \geq 
\Omega\left(1\right) \tau^{2/p} -\beta^{-1} \Psi(\rr) \geq \Omega\left(1\right) \tau^{2/p}.
\end{align*}
As $\rr_e \geq 1$, this implies $\sum_{e \in H } \rr_e (\NN\Delta)_e^2\geq \Omega\left(1\right) \tau^{2/p}$ .
We note that in a width reduction step, the resistances change by a factor of 2.
Thus, combining our last two observations, and applying Lemma~\ref{lem:ckmst:res-increase}, we get
\[
{\Psi\left({\rr^{(i,k+1)}}\right)} \geq {\Psi\left({\rr^{(i,k)}}\right)}+  \Omega\left(1\right) \tau^{2/p}.
\]
Finally, for the ``flow step'' case, we use the trivial bound from Lemma~\ref{lem:ckmst:res-increase}, ignoring the second term, 
\[
{\Psi\left({\rr^{(i,k+1)}}\right)} \geq {\Psi\left({\rr^{(i,k)}}\right)}. 
\]
\end{proof}

 \section{$\ell_p$-Regression}
 \label{sec:lp-reg}
%

%


\begin{definition}[$\kappa$-approximate solution] Let $\kappa \geq 1$. A $\kappa$-approximate solution for the residual problem is $\Dtil$ such that $\AA\Dtil = 0$ and $\res(\Dtil) \geq \frac{1}{\kappa}\res(\Delta^\star),$ where $\Delta^\star = \argmax_{\AA\Delta = 0} \res(\Delta)$.
\end{definition}

\begin{restatable}{lemma}{IterativeRefinement}(Iterative Refinement \cite{AdilPS19}). \label{lem:IterativeRefinement}
Let $p \geq 2$, and $\kappa \geq 1$. Starting from an initial feasible solution $\xx^{(0)}$,  and iterating as $\xx^{(t+1)} = \xx^{(t)} - \frac{\Delta}{p}$, where $\Delta$ is a $\kappa$-approximate solution to the residual problem (Definition \ref{def:residual}), we get an $\epsilon$-approximate solution to \eqref{eq:Problem} in at most $O\left(p \kappa \log \left(\frac{\ff\left(\xx^{(0)}\right)-\opt}{\epsilon \opt} \right)\right)$ calls to a $\kappa$-approximate solver for the residual problem. 
\end{restatable}
\begin{proof}
Let $\ff(\xx) =  \bb^{\top}\xx + \|\MM\xx\|_2^2 + \|\NN\xx\|_p^p$ and $\res(\Delta) = \gg^{\top}\Delta -  \Delta^{\top}\RR\Delta - \|\NN\Delta\|_p^p$. Observe that,
\[
\bb^{\top}(\xx+\Delta) = \bb^{\top}\xx + \bb^{\top}\Delta,
\]
and
\[
\|\MM(\xx+\Delta)\|_2^2 = \|\MM\xx\|_2^2 + 2\Delta^{\top}\MM^{\top}\MM\xx + \|\MM\Delta\|_2^2.
\]
Using Lemma B.1 from \cite{AdilPS19} we have,
\[
 \|\NN(\xx+\Delta)\|_p^p \leq  \|\NN\xx\|_p^p + p(\NN\Delta)^{\top}|\NN\xx|^{p-2}\NN\xx + 2p^2 \Delta^{\top}\NN^{\top}Diag(|\NN\xx|^{p-2})\NN\Delta + p^p \|\NN\Delta\|_p^p, 
\]
and,
\[
 \|\NN(\xx+\Delta)\|_p^p \geq \|\NN\xx\|_p^p + p(\NN\Delta)^{\top}|\NN\xx|^{p-2}\NN\xx + \frac{p}{8} \Delta^{\top}\NN^{\top}Diag(|\NN\xx|^{p-2})\NN\Delta + \frac{1}{2^{p+1}} \|\NN\Delta\|_p^p.
\]
Using these relations, we have,
\[
\ff(\xx+\Delta) \leq \ff(\xx) + p\gg^{\top}\Delta + p^2 \Delta^{\top}\RR\Delta + p^p\|\NN\Delta\|_p^p,
\]
or 
\[
\ff\left(\xx-\frac{\Delta}{p}\right) \leq \ff(\xx) - \res(\Delta).
\]
A lower bound looks like,
\[
\ff(\xx+\Delta) \geq \ff(\xx) +p \gg^{\top}\Delta + \frac{p}{16} \Delta^{\top}\RR\Delta + \frac{1}{2^{p+1}}\|\NN\Delta\|_p^p.
\]
For $\lambda = 16p$,
\begin{align*}
\ff(\xx) - \ff\left(\xx-\lambda \frac{\Delta}{p}\right) &\leq \lambda \gg^{\top}\Delta - \frac{\lambda^2}{16p} \Delta^{\top}\RR\Delta - \frac{\lambda^p}{p^p2^{p+1}}\|\NN\Delta\|_p^p\\
& \leq \lambda \left(\gg^{\top}\Delta - \frac{\lambda}{16p} \Delta^{\top}\RR\Delta - \frac{\lambda^{p-1}}{p^p2^{p+1}}\|\NN\Delta\|_p^p\right)\\
& \leq \lambda \left(\gg^{\top}\Delta - \Delta^{\top}\RR\Delta - \|\NN\Delta\|_p^p\right)\\
& = \lambda \res(\Delta).
\end{align*}
These relations are the same as Lemma B.2 of \cite{AdilPS19}. We can follow the proof further from \cite{AdilPS19} to get our result.
\end{proof}

\subsubsection*{Solving the Residual Problem}
\begin{restatable}{lemma}{BinarySearch}
\label{lem:BinarySearch}
Let $\Dopt$ denote the optimum of the residual problem at $\xx^{(t)}$ and $\opt$ denote the optimum of Problem \eqref{eq:Problem}. We have that $\res(\Dopt) \in (\nu/2, \nu]$ for some $\nu \in \left[\epsilon \frac{\opt}{p}, \ff\left(\xx^{(0)}\right)-\opt\right]$.
\end{restatable}
\begin{proof}
From the above proof, we note that for any $\xx$, let $\Dopt$ be the optimum of the residual problem.
\[
\res(\Dopt) \leq \ff(\xx) -\ff\left(\xx-\frac{\Dopt}{p}\right) \leq \ff\left(\xx^{(0)}\right)-\opt.
\]
Let $\Delta$ be the step we take to reach the optimum from $\xx$.
\[
\res(\Dopt) \geq \res(\Delta) \geq \frac{\ff(\xx) - \opt}{\lambda} \geq \epsilon \frac{\opt}{\lambda},
\]
where the last inequality follows since otherwise $\xx$ is an $\epsilon$ approximate solution. The lemma thus follows.
\end{proof}
%
\begin{restatable}{lemma}{ptoq}
\label{lem:p-to-q}
Let $p\geq p'$ and $\nu$ be such that $\res_p(\Dopt) \in (\nu/2,\nu]$,
where $\Dopt$ is the optimum of the residual problem for $p$-norm
(Definition~\ref{def:residual}). The following problem has optimum between
$\left[\nu/32,O(1)m^{\frac{1}{p'-1}}\nu\right] \defeq (a\nu,b\nu]$.
\begin{align}
\label{eq:qProblem}
  \max_{\AA \Delta = 0} \quad & \gg^{\top}\Delta -
  \Delta^{\top}\RR\Delta - \frac{1}{2} \left(\frac{\nu}{m}
  \right)^{1-\frac{p'}{p}} \norm{\NN\Delta}_{p'}^{p'}
\end{align}
For $\beta \geq 1$, if $\Dtil$ is a $\beta$-approximate solution to the above problem, then $\alpha\Dtil$ gives an $8\frac{b\beta^2}{a^2}m^{\frac{p}{p-1}\left(\frac{1}{p'}- \frac{1}{p}\right)+\frac{1}{p'-1}}$ approximate solution to the residual problem,
where $\alpha =\frac{a}{4b\beta}m^{-\frac{p}{p-1}\left(\frac{1}{p'}- \frac{1}{p}\right)-\frac{1}{p'-1}}$. 
\end{restatable}
\begin{proof}
  We will first show that the optimum of Problem \eqref{eq:qProblem}
  is at most $O(m^{\frac{1}{p'-1}}\nu)$. Suppose the optimum $\Dopt$ is such that it
  gives an objective value of $\beta\nu$. Using a scaling argument as
  in the above proof, we can conclude that,
  \[
    \gg^{\top} \Dopt - \Dopt^{\top}\RR\Dopt - \frac{1}{2}
    \left(\frac{\nu}{m} \right)^{1-\frac{p'}{p}}
    \norm{\NN\Dopt}_{p'}^{p'} = \Dopt^{\top}\RR\Dopt + (p'-1)
    \frac{1}{2} \left(\frac{\nu}{m} \right)^{1-\frac{p'}{p}}
    \norm{\NN\Dopt}_{p'}^{p'} = \beta \nu.
  \]
  This implies that $\gg^{\top} \Dopt \geq \beta \nu$,
  $\Dopt^{\top}\RR\Dopt \leq \beta \nu$ and
  $\norm{\NN\Delta^\star}_{p'}^{p'} \leq 2 \beta \nu^{p'/p} m^{1-p'/p}$. We
  thus have,
  \[
    \norm{\NN\Delta^\star}_{p}^{p} \leq (2\beta)^{p/p'} m^{p/p'- 1}\nu .
  \]
Let $\alpha =  \frac{1}{16}m^{-\frac{p}{p-1}\left(\frac{1}{p'}- \frac{1}{p}\right)}$ be some scaling factor. Now,
\[
\alpha^2  \Dopt^{\top}\RR\Dopt  \leq \alpha\frac{\beta \nu}{8},
\]
and
\begin{equation}
 \alpha^{p} \norm{\NN\Dopt}_{p}^{p} \leq  \alpha  \frac{1}{16^{p-1}}m^{-\left(\frac{p}{p'}- 1\right)}\norm{\NN\Dopt}_{p}^{p}  \leq \frac{\alpha}{8^{p-1}} 2^{\frac{p}{p'} - (p-1)} \beta^{\frac{p}{p'} } \nu \leq \alpha \beta^{\frac{p}{p'} }\frac{\nu}{8}.
\end{equation}
Consider,
\[
\res_p\left(\beta^{-\frac{p-p'}{p'(p-1)}}\alpha\Dopt\right) \geq \left(\beta^{-\frac{p-p'}{p'(p-1)}}\alpha\right)\left(\gg^{\top}\Dopt - \frac{\beta \nu}{8} - \frac{\beta \nu}{8} \right) \geq \beta^{\frac{p(p'-1)}{p'(p-1)}}\frac{\alpha}{8} \nu.
\]
Since $\res_p(\cdot)$ has optimum at most $\nu$, we must have
\[
\beta^{\frac{p(p'-1)}{p'(p-1)}}\frac{\alpha}{8} \leq 1,
\]
which gives,
\[
\beta \leq O(1)m^{\frac{1}{p'-1}}.
\]
Thus Problem \eqref{eq:qProblem} has an optimum at most $O(1)m^{\frac{1}{p'-1}}\nu = b \nu$. To obtain a lower bound, consider $\Dtil$ obtained in Lemma 3.1 of \cite{AdilS20}. We will evaluate the objective at $\frac{1}{8}\Dtil$,
\[
\gg^{\top} \frac{1}{8}\Dtil-   \frac{1}{8^2}\Dtil^{\top}\RR\Dtil  - \frac{1}{8^p}\frac{1}{2} \left(\frac{\nu}{m} \right)^{1-\frac{p'}{p}} \norm{\NN\Dtil}_{p'}^{p'} \geq \frac{\nu}{16} - \frac{\nu}{32} = \frac{\nu}{32} = a\nu.
\]
Therefore, the optimum of Problem \eqref{eq:qProblem} must be at least $\frac{\nu}{32}$. We next look at how a $\beta$-approximate solution of \eqref{eq:qProblem} translates to an approximate solution of the residual problem for $p$. Let $\Dtil$ be a $\beta$ approximate solution to \eqref{eq:qProblem} and $\Dopt$ denote its optimum. Denote the objective at $\Delta$ for \eqref{eq:qProblem} as $\res_{p'}(\Delta)$. We know that $\res_{p'}(\Dopt) \geq \res_{p'}(\Dtil) \geq \frac{1}{\beta}\res_{p'}(\Dopt)$. Note that, $\gg^{\top} \Dtil$ has to be between $a\nu/\beta$ and $ b\nu$. This ensures that, 
\[
 \Dtil^{\top}\RR\Dtil  +  \frac{1}{2} \left(\frac{\nu}{m} \right)^{1-\frac{p'}{p}} \norm{\NN\Dtil}_{p'}^{p'} \leq \gg^{\top} \Dtil  \leq b \nu.
\]
Let $\alpha =\frac{a}{4b\beta}m^{-\frac{p}{p-1}\left(\frac{1}{p'}- \frac{1}{p}\right)- \frac{1}{p'-1}}$. Following the calculations above, we have,
\[
\alpha^2  \Dtil^{\top}\RR\Dtil  \leq \alpha\frac{a\nu}{4\beta},
\]
and
\begin{equation*}
 \alpha^{p} \norm{\NN\Dtil}_{p}^{p} \leq  \alpha  \frac{a^{p-1}}{4^{p-1}b^{p-1}\beta^{p-1}}m^{-\left(\frac{p}{p'}- 1\right)- \frac{p-1}{p'-1}}\norm{\NN\Dtil}_{p}^{p}  \leq \alpha\frac{a\nu}{4\beta}.
\end{equation*}
Therefore,
\[
\res_p(\alpha\Dtil) \geq  \frac{a\alpha}{2\beta} \nu.
\]

\end{proof}

\deeksha{There should be a way to avoid the $\beta^2$. It basically comes form the fact that now the range of values for the linear and other terms has a $\beta$ factor in it. This works for now so will get to improving it later. Fixing the linear term made a lot of difference in terms of avoiding these extra factors.}
\begin{restatable}{lemma}{DecisionProblem}
\label{lem:Decision}
At $\xx^{(t)}$, let $\nu$ be such that $\res(\Dopt) \in (a\nu, b\nu]$ for some values $a$ and $b$. The following problem has optimum at most $b\nu$.
\begin{align}
  \label{eq:Decision}
  \begin{aligned}
  \min_{\Delta} & \quad \Delta^{\top}\RR\Delta + \|\NN\Delta\|_p^p \\
  \text{s.t.} & \quad  \gg^{\top}\Delta = a\nu, \AA\Delta = 0.
  \end{aligned}
\end{align}
Further, if $\Dtil$ is a feasible solution to Problem \eqref{eq:Decision} such that $\Dtil^{\top}\RR\Dtil \leq \alpha b \nu$ and $\|\NN\Dtil\|_p^p \leq \beta b \nu$ , then we can pick a scalar $\mu = \frac{a}{4b\alpha \beta^{\nfrac{1}{(p-1)}}}$ such that $\mu\Dtil$ is a $\frac{16 b^2 \alpha \beta^{\nfrac{1}{(p-1)}}}{a^2}$-approximate solution to the residual problem.
\end{restatable}
Adapted from the proof of Lemma B.4 of \cite{AdilPS19}
\begin{proof}
The assumption on the residual is
\[
\res(\Dopt) = \gg^{\top}\Dopt -\Dopt^{\top}\RR\Dopt - \norm{\NN\Dopt}_p^p \in  (a\nu,b\nu].
\]
Since the last $2$ terms are strictly non-positive, we must have, $\gg^{\top}\Dopt  \geq a\nu.$
Since $\Dopt$ is the optimum and satisfies $\AA\Dopt = 0$, 
\[
\frac{d}{d\lambda}\left(\gg^{\top}\lambda \Dopt -  \lambda^2\Dopt^{\top}\RR\Dopt - \lambda^p \norm{\NN\Dopt}_p^p \right)_{\lambda = 1} =0.
\]
Thus,
\[
\gg^{\top} \Dopt -   \Dopt^{\top}\RR\Dopt  - \norm{\AA\Dopt}_p^p  =  \Dopt^{\top}\RR\Dopt  + (p-1) \norm{\NN\Dopt}_p^p.
\]
Since $p \ge 2,$ we get the following
\[
\Dopt^{\top}\RR\Dopt  +  \norm{\NN\Dopt}_p^p \leq \gg^{\top} \Dopt -   \Dopt^{\top}\RR\AA\Dopt  -  \norm{\AA\Dopt}_p^p \leq b\nu.
\]
For notational convenience, let function $ \hh_p(\RR,\Delta) =  \Delta^{\top}\RR\Delta  +  \norm{\NN\Delta}_p^p$.
Now, we know that, $\gg^{\top}\Dopt  \geq a\nu$ and $ \gg^{\top} \Dopt - \hh_p(\RR,\Dopt) \leq b \nu$. This gives, 
\[
a\nu\leq \gg^{\top} \Dopt \leq \hh_p(\RR,\Dopt) + b\nu \leq 2b\nu.
\]
Let $\Delta = \delta \Dopt$, where $\delta = \frac{a\nu}{\gg^{\top}\Dopt}$. Note that $\delta \in [a/2b, 1]$. Now, $\gg^{\top}\Delta = a\nu$ and, 
\[
\hh_p(\RR,\Delta) \leq \max\{\delta^2,\delta^p\} \hh_p(\RR,\Dopt) \leq  b\nu.
\]
Note that this $\Delta$ satisfies the constraints of program \eqref{eq:Decision} and has an optimum at most $b\nu$. So the optimum of the program must have an objective at most $b\nu$. Now let $\Dtil$ be a $(\alpha,\beta)$ approximate solution to \eqref{eq:Decision}, i.e.,
\[
\Dtil^{\top}\RR\Dtil  \leq \alpha b\nu \quad \text{and} \quad \|\NN\Dtil\|_p^p \leq \beta b\nu.
\]
Let $\mu = \frac{a}{4 \alpha \beta^{1/(p-1)} b}$. We have,
\begin{align*}
\gg^{\top}(\mu\Dtil) - \hh_p(\RR,\mu\Dtil) & \geq \mu\left( a\nu - \mu  \Dtil^{\top}\RR\Dtil - \mu^{p-1}\|\NN\Dtil\|_p^p\right)\\
&\geq \mu \left(a\nu - \frac{a\nu}{4\beta^{1/(p-1)}} - \frac{a \nu}{4\alpha}\right) \\
& \geq \mu \left(a\nu - \frac{a\nu}{4} - \frac{a \nu}{4}\right)\\
& \geq \frac{a\mu}{4b} b\nu \geq \frac{a^2}{16b^2\alpha \beta^{1/(p-1)}} \opt. 
\end{align*}
\end{proof}

\MetaThm*
\begin{proof}
From Lemma \ref{lem:IterativeRefinement} we know that given an instance of Problem \eqref{eq:Problem}, at every $\xx^{(t)}$ we can define a residual problem and if $\Dbar$ is a $\beta$ approximate solution of the residual problem, updating $\xx^{(t)}$ to $\xx^{(t)} -\frac{\Dbar}{p}$, we can find the required $\epsilon$ approximate solution in 
\[
O\left(p\beta \log \left(\frac{\ff\left(\xx^{(0)}\right)-\opt}{\epsilon\opt} \right)\right) \leq O\left(p\beta \log \left(\frac{\kappa}{\epsilon} \right)\right)
\]
iterations. The last inequality follows since $\frac{\ff\left(\xx^{(0)}\right)-\opt}{\epsilon\opt} \leq \frac{\kappa \opt}{\epsilon\opt} = \frac{\kappa}{\epsilon}$. It is thus sufficient to solve a residual problem at $\xx$ that looks like,
 \[
 \max_{\AA\Delta = 0} \quad  \gg^{\top}\Delta -  \Delta^{\top}\RR\Delta - \|\NN\Delta\|_p^p.
 \] 
 Here $\gg$ and $\RR$ depend on $\xx^{(t)},\MM,\NN$. Now, suppose we have $\res(\Dopt) \in (\nu/2,\nu]$. We will consider the following cases:
 \begin{enumerate}
 \item $p \leq \log m$:\\
 We apply {\sc Sparsify} to $\gg,\RR,\NN$ to get $\widetilde{\gg},\widetilde{\RR},\widetilde{\NN}$. Now is $\Delta$ is a $\beta$ approximate solution to 
  \begin{equation}\label{eq:Spar_res}
 \max_{\AA\Delta = 0} \quad \widetilde{\gg}^{\top}\Delta -  \Delta^{\top}\widetilde{\RR}\Delta - \|\widetilde{\NN}\Delta\|_{p}^{p}, 
 \end{equation}
 then $\mu_2\Delta$ is a $\kappa_2\beta$ approximate solution to the residual problem. We will now solve the above problem. Note that the size of this problem is $m$. Let $\Dtil$ be a $\kappa_3,\kappa_4$ approximate solution to Problem \eqref{eq:InsideProb}. From \ref{lem:Decision}, $\frac{a}{2b\kappa_3\kappa_4^{1/(p-1)}}\Dtil$ is a $\frac{4\kappa_3 \kappa_4^{1/(p-1)} b^2}{a^2} $ approximation to \eqref{eq:Spar_res}. Now, going back, $\frac{a}{2b\kappa_3\kappa_4^{1/(p-1)}}\mu_2\Dtil$ is a $\frac{4\kappa_3 \kappa_4^{1/(p-1)}b^2}{a^2}\kappa_2 $-approximate solution to the residual problem. Now, $\mu_1 = \kappa_1 = 1$ and thus we have that, in this case if $\res(\Dopt) \in (\nu/2,\nu]$, then $\frac{a}{2b\kappa_3\kappa_4^{1/(p-1)}}\mu_2\Dtil$ is a $\frac{4\kappa_3 \kappa_4^{1/(p-1)}b^2}{a^2}\kappa_2 $-approximate solution to the residual problem.

\item $ p>\log m$:\\
In this case, there is an additional step that converts the residual problem to an instance of the previous case for $p = \log m$. We apply {\sc Sparsify} to $\gg,\RR,\NN'$ and set $p = \log m$. If $\Dbar$ is the $\alpha$ solution returned by the previous case on $p = \log m$ and $\NN = \NN'$, then $\mu_1 \Dbar$ is a $\kappa_1\alpha$-approximation for the residual problem. Thus $\frac{a}{2b\kappa_3\kappa_4^{1/(p-1)}}\mu_2\mu_1\Dtil$ is a $\frac{4\kappa_3 \kappa_4^{1/(p-1)}b^2}{a^2}\kappa_2\kappa_1 $-approximate solution to the residual problem.
 \end{enumerate}
From the above discussion we conclude that, if $\res(\Dopt) \in (\nu/2,\nu]$, then we get a solution $\Dbar$ such that $\res(\Dbar) \geq \frac{a^2}{4\kappa_4^{1/(p-1)}\kappa_3 \kappa_2\kappa_1b^2}\res(\Dopt) $. From Lemma \ref{lem:BinarySearch}, we know that $\res(\Dopt) \in (\nu/2,\nu]$ for some $\nu \in \left[\epsilon \frac{\opt}{p}, \ff\left(\xx^{(0)}\right)-\opt\right]$. Since $\opt \geq \frac{\ff\left(\xx^{(0)}\right)}{\kappa}$ and $\opt \geq 0$, it is sufficient to change $\nu$ in the range $\left[\epsilon \frac{\ff\left(\xx^{(0)}\right)}{\kappa p}, \ff\left(\xx^{(0)}\right)\right]$.

We finally look at the running time. We start with a residual problem. We require time $K$ to apply procedure {\sc Sparsify}. We next require time $\tilde{K}(m)$ to solve \eqref{eq:InsideProb}. Now this either gives us an $\frac{4\kappa_4^{1/(p-1)}\kappa_3 \kappa_2\kappa_1b^2}{a^2}$-approximate solution to the residual problem. For every residual problem we repeat the above process at most $\log\frac{\kappa p}{\epsilon}$ times (corresponding to the number of values of $\nu$). We use the fact that $a,b \leq m^{o(1)}$. Thus the total running time is,
\begin{align*}
&O\left(\frac{4p\kappa_4^{1/(p-1)}\kappa_3 \kappa_2\kappa_1b^2}{a^2} (K + \tilde{K}(m))\log\left(\frac{\kappa}{\epsilon} \right)\log\left(\frac{\kappa p}{\epsilon}\right)\right)\\
& \leq \widetilde{O}\left(p\kappa_4^{1/(p-1)}\kappa_3 \kappa_2\kappa_1 (K + \tilde{K}(m))\log\left(\frac{\kappa p}{\epsilon}\right)^2\right)
\end{align*}
\end{proof}

\end{document}